\documentclass[12pt]{article}
\usepackage{amsmath}
\usepackage{mathtools}
\usepackage{graphicx,psfrag,epsf}
\usepackage{enumerate}
\usepackage{url}
\usepackage{amssymb}
\usepackage{amsfonts}
\usepackage{centernot}
\usepackage[nohead]{geometry}
\usepackage[singlespacing]{setspace}
\usepackage[bottom]{footmisc}
\usepackage{indentfirst}
\usepackage{endnotes}
\usepackage{rotating}
\usepackage{amsthm}
\usepackage{booktabs}
\usepackage{algorithm}
\usepackage{array}
\usepackage{multirow}
\usepackage{algpseudocode}
\usepackage{enumerate}
\usepackage{bbm}
\usepackage{dsfont}
\usepackage{verbatim}
\usepackage{subfig}
\usepackage{xr}
\usepackage{relsize}
\usepackage{color}
\RequirePackage[OT1]{fontenc} \RequirePackage{amsthm,amsmath}
\RequirePackage[numbers]{natbib}
\usepackage{tabularx}
    \newcolumntype{L}{>{\raggedright\arraybackslash}X}

\newcommand{\be} {\begin{eqnarray*}}
\newcommand{\ee} {\end{eqnarray*}}

\newcommand{\wrho} {{\widetilde \rho}}

\makeatletter
\newcases{nocases}
  {\quad}
  {$\m@th\displaystyle{##}$\hfil}
  {$\m@th\displaystyle{##}$\hfil}
  {.}{.}
\makeatother

\newcommand{\field}[1]{\mathbb{#1}}

\newcommand{\E}{\field{E}}
\newcommand{\tran}{^{\top\kern-\scriptspace}}

\def\chakraborty{\textcolor{blue}}

\theoremstyle{example} \theoremstyle{remark} \theoremstyle{lemma}
\theoremstyle{definition} \theoremstyle{corol}
\theoremstyle{proposition} \theoremstyle{condition}
\theoremstyle{assumption}
\newtheorem{assumption}{\n{Assumption}}[section]
\newtheorem{theorem}{\n{Theorem}}[section]
\newtheorem{example}{\n{Example}}[section]
\newtheorem{remark}{\n{Remark}}[section]
\newtheorem{lemma}{\n{Lemma}}[section]
\newtheorem{definition}{\n{Definition}}[section]

\newtheorem{proposition}{\n{Proposition}}[section]

\newcommand{\Rmnum}[1]{\expandafter\romannumeral #1}

\def\bb{\mathbb}
\def\cal{\mathcal}
\def\dis{\displaystyle}
\font\n=cmcsc12
\def\cov{{\mbox{cov}}}
\def\var{{\mbox{var}}}

\def\symb{\mathlarger{\mathlarger{\varepsilon}}}
\newcommand{\bigCI}{\mathrel{\text{\scalebox{1.07}{$\perp\mkern-10mu\perp$}}}}
\newcommand{\nbigCI}{\centernot{\bigCI}}

\DeclareMathOperator*{\plim}{plim}
  \textwidth = 420pt
\renewcommand{\baselinestretch}{1}

\makeatletter
\renewcommand\@biblabel[1]{}
\renewenvironment{thebibliography}[1]
     {\section*{\refname}%
      \@mkboth{\MakeUppercase\refname}{\MakeUppercase\refname}%
      \list{}%
           {\leftmargin0pt
            \@openbib@code
            \usecounter{enumiv}}%
      \sloppy
      \clubpenalty4000
      \@clubpenalty \clubpenalty
      \widowpenalty4000%
      \sfcode`\.\@m}
     {\def\@noitemerr
       {\@latex@warning{Empty `thebibliography' environment}}%
      \endlist}
\makeatother

\newcounter{alphasect}
\def\alphainsection{0}

\let\oldsection=\section
\def\section{%
  \ifnum\alphainsection=1%
    \addtocounter{alphasect}{1}
  \fi%
\oldsection}%

\renewcommand\thesection{%
 \ifnum\alphainsection=1%
   \Alph{alphasect}%
 \else
   \arabic{section}%
 \fi%
}%

\newenvironment{alphasection}{%
  \ifnum\alphainsection=1%
    \errhelp={Let other blocks end at the beginning of the next block.}
    \errmessage{Nested Alpha section not allowed}
  \fi%
  \setcounter{alphasect}{0}
  \def\alphainsection{1}
}{%
  \setcounter{alphasect}{0}
  \def\alphainsection{0}
}%

\begin{document}

\def\spacingset#1{\renewcommand{\baselinestretch}%
{#1}\small\normalsize} \spacingset{1}


  \title{\bf A New Framework for Distance and Kernel-based Metrics in High Dimensions}
  \author{Shubhadeep Chakraborty \\
    Department of Statistics, Texas A\&M University\\
    and \\
    Xianyang Zhang \\
    Department of Statistics, Texas A\&M University}
    \date{}
  \maketitle

\bigskip

\begin{abstract}

The paper presents new metrics to quantify and test for (i) the equality of distributions and (ii) the independence between two high-dimensional random vectors. We show that the energy distance based on the usual Euclidean distance cannot completely characterize the homogeneity of two high-dimensional distributions in the sense that it only detects
the equality of means and the traces of covariance matrices in the high-dimensional setup. We propose a new class of metrics which inherits the desirable properties of the energy distance and maximum mean discrepancy/(generalized) distance covariance and the Hilbert-Schmidt Independence Criterion in the low-dimensional setting and is capable of detecting the homogeneity of/completely characterizing independence between the low-dimensional marginal distributions in the high dimensional setup. We further propose t-tests based on the new metrics to perform high-dimensional two-sample testing/independence testing and study their asymptotic behavior under both high dimension low sample size (HDLSS) and high dimension medium sample size (HDMSS) setups. The computational complexity of the t-tests only grows linearly with the dimension and thus is scalable to very high dimensional data. We demonstrate the superior power behavior of the proposed tests for homogeneity of distributions and independence via both simulated and real datasets.

\end{abstract}

\noindent%

\newpage
\spacingset{1.45} 

\section{Introduction}

Nonparametric two-sample testing of homogeneity of distributions has been a classical problem in statistics, finding a plethora of applications in goodness-of-fit testing, clustering, change-point detection and so on. Some of the most traditional tools in this domain are Kolmogorov-Smirnov test, and Wald-Wolfowitz runs test, whose multivariate and multidimensional extensions have been studied by Darling\,(1957), David\,(1958) and Bickel\,(1969) among others. Friedman and Rafsky\,(1979) proposed a distribution-free multivariate generalization of the Wald-Wolfowitz runs test applicable for arbitrary but fixed dimensions. Schilling\,(1986) proposed another distribution-free test for multivariate two-sample problem based on $k$-nearest neighbor ($k$-NN) graphs. Maa et al.\,(1996) suggested a technique for reducing the dimensionality by examining the distribution of interpoint distances. In a recent novel work, Chen and Friedman\,(2017) proposed graph-based tests for moderate to high dimensional data and non-Euclidean data. The last two decades have seen an abundance of literature on distance and kernel-based tests for equality of distributions. Energy distance (first introduced by Sz\'{e}kely (2002)) and maximum mean discrepancy or MMD (see Gretton et al.\,(2012)) have been widely studied in both the statistics and machine learning communities. Sejdinovic et al.\,(2013) provided a unifying framework establishing the equivalence between the (generalized) energy distance and MMD. Although there have been some very recent works to gain insight on the decaying power of the distance and kernel-based tests for high dimensional inference (see for example Ramdas et al.\,(2015a, 2015b), Kim et al. (2018) and  Li (2018)), the behavior of these tests in the high dimensional setup is still a pretty unexplored area.

Measuring and testing for independence between two random vectors has been another fundamental problem in statistics, which has found applications in a wide variety of areas such as independent component analysis, feature selection, graphical modeling, causal inference, etc. There has been an enormous amount of literature on developing dependence metrics to quantify non-linear and non-monotone dependence in the low dimensional context. Gretton et al.\,(2005, 2007) introduced a kernel-based independence measure, namely the Hilbert-Schmidt Independence Criterion (HSIC). Bergsma and Dassios\,(2014) proposed a consistent test of independence of two ordinal random variables based on an extension of Kendall's tau. Josse and Holmes\,(2014) suggested tests of independence based on the RV coefficient. Sz\'{e}kely et al.\,(2007), in their seminal paper, introduced distance covariance (dCov) to characterize dependence between two random vectors of arbitrary dimensions. Lyons\,(2013) extended the notion of distance covariance from Euclidean spaces to arbitrary metric spaces. Sejdinovic et al.\,(2013) established the equivalence between HSIC and (generalized) distance covariance via the correspondence between positive definite kernels and semi-metrics of negative type. Over the last decade, the idea of distance covariance has been widely extended and analyzed in various ways; see for example Zhou\,(2012), Sz\'{e}kely and Rizzo\,(2014), Wang et al.\,(2015), Shao and Zhang\,(2014), Huo and Sz\'{e}kely\,(2016), Zhang et al.\,(2018), Edelmann et al.\,(2018) among many others. There have been some very recent literature  which aims at generalizing distance covariance to quantify the joint
dependence among more than two random vectors; see for example Matteson and Tsay\,(2017), Jin and Matteson\,(2017), Chakraborty and Zhang\,(2018), B\"{o}ttcher\,(2017), Yao et al.\,(2018), etc. However, in the high dimensional setup, the literature is scarce, and the behavior of the widely used distance and kernel-based dependence metrics is not very well explored till date.  Sz\'{e}kely and Rizzo\,(2013) proposed a distance correlation based t-test to test for independence in high dimensions. In a very recent work, Zhu et al.\,(2018) showed that in the high dimension low sample size (HDLSS) setting, i.e., when the dimensions grow while the sample size is held fixed, the sample distance covariance can only measure the component-wise {\it linear dependence} between the two vectors. As a consequence, the distance correlation based t-test proposed by Sz\'ekely et al.\,(2013) for independence between two high dimensional random vectors has trivial power when the two random vectors are nonlinearly dependent but component-wise uncorrelated. As a remedy, Zhu et al.\,(2018) proposed a test by aggregating the pairwise squared sample distance covariances and studied its asymptotic behavior under the HDLSS setup.

This paper presents a new class of metrics to quantify the homogeneity of distributions and independence between two high-dimensional random vectors. The core of our methodology is a new way of defining the distance between sample points (interpoint distance) in the high-dimensional Euclidean spaces. In the first part of this work, we show that the energy distance based on the usual Euclidean distance cannot completely characterize the homogeneity of two high-dimensional distributions in the sense that it only detects the {\it equality of means and the traces of covariance matrices} in the high-dimensional setup. To overcome such a limitation, we propose a new class of metrics based on the new distance which inherits the nice properties of energy distance and maximum mean discrepancy in the low-dimensional setting and is capable of detecting the {\it pairwise homogeneity of the low-dimensional marginal distributions} in the HDLSS setup. We construct a high-dimensional two sample t-test based on the U-statistic type estimator of the proposed metric, which can be viewed as a generalization of the classical two-sample t-test with equal variances. We show under the HDLSS setting that the new two sample t-test converges to a central t-distribution under the null and it has nontrivial power for a broader class of alternatives compared to the energy distance.  We further show that the two sample t-test converges to a standard normal limit under the null when the dimension and sample size both grow to infinity with the dimension growing more rapidly. It is worth mentioning that we develop an approach to unify the analysis for the usual energy distance and the proposed metrics. Compared to existing works, we make the following contribution.
\begin{itemize}
\item We derive the asymptotic variance of the generalized energy distance under the HDLSS setting and propose a computationally efficient variance estimator (whose computational cost is linear in the dimension). Our analysis is based on a pivotal t-statistic which does not require permutation or resampling-based inference and allows an asymptotic exact power analysis.
\end{itemize}

In the second part, we propose a new framework to construct dependence metrics to quantify the dependence between two high-dimensional random vectors $X$ and $Y$ of possibly different dimensions. The new metric, denoted by $\cal{D}^2(X, Y)$, generalizes both the distance covariance and HSIC. It completely characterizes independence between $X$ and $Y$ and inherits all other desirable properties of the distance covariance and HSIC for fixed dimensions. In the HDLSS setting, we show that the proposed population dependence metric behaves as an aggregation of group-wise (generalized) distance covariances. We construct an unbiased U-statistic type estimator of $\cal{D}^2(X, Y)$ and show that with growing dimensions, the unbiased estimator is asymptotically equivalent to the sum of group-wise squared sample (generalized) distance covariances. Thus it can quantify {\it group-wise non-linear dependence} between two high-dimensional random vectors, going beyond the scope of the distance covariance based on the usual Euclidean distance and HSIC which have been recently shown only to capture the componentwise linear dependence in high dimension, see Zhu et al. (2018). We further propose a t-test based on the new metrics to perform high-dimensional independence testing and study its asymptotic size and power behaviors under both the HDLSS and high dimension medium sample size (HDMSS) setups. In particular, under the HDLSS setting, we prove that the proposed t-test converges to a central t-distribution under the null and a noncentral t-distribution with a random noncentrality parameter under the alternative. Through extensive numerical studies, we demonstrate that the newly proposed t-test can capture group-wise nonlinear dependence which cannot be detected by the usual distance covariance and HSIC in the high dimensional regime. Compared to the marginal aggregation approach in Zhu et al. (2018), our new method enjoys two major advantages.
\begin{itemize}
\item Our approach provides a neater way of generalizing the notion of distance and kernel-based dependence metrics. The newly proposed metrics completely characterize dependence in the low-dimensional case and capture group-wise nonlinear dependence in the high-dimensional case. In this sense, our metric can detect a wider range of dependence compared to the marginal aggregation approach.
\item The computational complexity of the t-tests only grows linearly with the dimension and thus is scalable to very high dimensional data. 
\end{itemize}

\emph{Notation}. Let $X = (X_1, \dots X_p) \in \bb{R}^p$ and $Y = (Y_1, \dots, Y_q)$ $\in \bb{R}^q$ be two random vectors of dimensions $p$ and $q$ respectively. Denote by $\Vert\cdot\Vert_p$ the Euclidean norm of $\mathbb{R}^p$ (we shall use it interchangeably with $\Vert\cdot\Vert$ when there is no confusion). Let $0_p$ be the origin of $\bb{R}^p$. We use $X \bigCI Y$ to denote that $X$ is independent of $Y$, and use $``X \overset{d}{=} Y"$ to indicate that $X$ and $Y$ are identically distributed.
Let $(X',Y')$, $(X'', Y'')$ and $(X''', Y''')$ be independent copies of $(X,Y)$.
We utilize the order in probability notations such as stochastic boundedness
$O_p$ (big O in probability), convergence in probability $o_p$ (small o in
probability) and equivalent order $\asymp$, which is defined as follows: for a sequence of random variables $\{Z_n\}_{n=1}^{\infty}$
and a sequence of real numbers $\{a_n\}_{n=1}^{\infty}$, $Z_n \asymp_p a_n$ if and only if $Z_n/a_n = O_p(1)$ and $a_n/Z_n = O_p(1)$ as $n \to \infty$.
For a metric space $(\cal{X}, d_{\cal{X}})$, let $\cal{M}(\cal{X})$ and $\cal{M}_1(\cal{X})$ denote the set of all finite signed Borel measures on $\cal{X}$ and all probability measures on $\cal{X}$, respectively. Define $\cal{M}^1_{d_{\cal{X}}}(\cal{X}):= \{v \in \cal{M}(\cal{X}) \,:\, \exists\, x_0 \in \cal{X} \; \text{s.t.}\;\int_{\cal{X}} d_{\cal{X}}(x,x_0)\, d|v|(x) < \infty\}$. For $\theta > 0$, define $\cal{M}_{\mathcal{K}}^{\theta}(\cal{X}):= \{v \in \cal{M}(\cal{X}) \,:\, \int_{\cal{X}} \mathcal{K}^{\theta}(x,x)\, d|v|(x) < \infty\}$, where $\mathcal{K}: \cal{X} \times \cal{X} \to \bb{R}$ is a bivariate kernel function.
Define $\cal{M}^1_{d_{\cal{Y}}}(\cal{Y})$ and $\cal{M}_{\mathcal{K}}^{\theta}(\cal{Y})$ in a similar way. 
For a matrix $A = (a_{kl})_{k,l=1}^n \in \bb{R}^{n \times n}$, define its $\cal{U}$-centered version $\tilde{A} = (\tilde{a}_{kl}) \in \bb{R}^{n \times n}$ as follows
\begin{align}\label{U centering def}
\tilde{a}_{kl} = \begin{cases}
a_{kl}\, -\,\mathlarger{ \frac{1}{n-2}} \dis\sum_{j=1}^n a_{kj} \,-\, \frac{1}{n-2} \dis\sum_{i=1}^n a_{il}\, +\, \frac{1}{(n-1)(n-2)} \dis\sum_{i,j=1}^n a_{ij}, \; & k \neq l,\\
0, & k=l, \end{cases}
\end{align}
for $k, l = 1, \dots, \, n$. Define
$$(\tilde{A} \cdot \tilde{B}):=\frac{1}{n(n-3)} \dis \sum_{k\neq l} \tilde{a}_{kl} \tilde{b}_{kl}$$
for $\tilde{A} = (\tilde{a}_{kl})$ and $\tilde{B} = (\tilde{b}_{kl})\in \bb{R}^{n \times n}$.
Denote by $\textrm{tr}(A)$ the trace of a square matrix $A$. $A \otimes B$ denotes the kronecker product of two matrices $A$ and $B$. Let $\Phi(\cdot)$ be the cumulative distribution function of the standard normal distribution. 
Denote by $t_{a,b}$  the noncentral t-distribution with $a$ degrees of freedom
and noncentrality parameter $b$. Write $t_a=t_{a,0}$. Denote by $q_{\alpha, a}$ and $Z_{\alpha}$ the upper $\alpha$ quantile of the distribution of $t_{a}$ and the standard normal distribution, respectively, for $\alpha \in (0,1)$. Also denote by $\chi^2_{a}$  the chi-square distribution with $a$ degrees of freedom.
Denote $U\sim$ Rademacher\,$(0.5)$\, if\, $P(U=1) = P(U=-1) = 0.5$. Let $\mathbbm{1}_A$ denote the indicator function associated with a set $A$. Finally, denote by $\lfloor a \rfloor$ the integer part of $a\in\mathbb{R}$.

\section{An overview: distance and kernel-based metrics} \label{sec:overview}
\subsection{Energy distance and MMD}\label{ed sec}

Energy distance (see Sz\'{e}kely et al.\,(2004, 2005), Baringhaus and Franz\,(2004)) or the Euclidean energy distance between two random vectors $X,Y\in \bb{R}^p$ and $X \bigCI Y$ with $\E \Vert X \Vert_p < \infty$ and $\E \Vert Y \Vert_p < \infty$, is defined as
\begin{equation}\label{ed def}
ED(X,Y) \;=\; 2\,\E \Vert X-Y \Vert_p - \E \Vert X-X'\Vert_p - \E \Vert Y-Y'\Vert_p \; ,
\end{equation}
where $(X',Y')$ is an independent copy of $(X,Y)$. Theorem 1 in Sz\'{e}kely et al.\,(2005) shows that \,$ED(X,Y) \geq 0$ and the equality holds if and only if $X \overset{d}{=} Y$. In general, for an arbitrary metric space $(\cal{X}, d)$, the generalized energy distance between $X \sim P_X$\, and \,$Y \sim P_Y$ where $P_X, P_Y \in \cal{M}_1(\cal{X}) \cap \cal{M}^1_{d}(\cal{X})$ is defined as
\begin{equation}\label{ed def general}
ED_d(X,Y) \;=\; 2\,\E \,d(X,Y) - \E \,d(X,X') - \E \,d(Y,Y') \;.
\end{equation}
\begin{definition}[Spaces of negative type]\label{negative type}
A metric space $(\cal{X}, d)$ is said to have negative type if for all $n\geq 1$, $x_1, \dots, x_n \in \cal{X}$ and $\alpha_1, \dots, \alpha_n \in \bb{R}$ \,with $\sum_{i=1}^n \alpha_i = 0$, we have
\begin{align}\label{def:neg type}
\sum_{i, j =1}^n \alpha_i\, \alpha_j\, d(x_i, x_j) \leq 0\;.
\end{align}
The metric space $(\cal{X}, d)$ is said to be of strong negative type if the equality in (\ref{def:neg type}) holds only when $\alpha_i = 0$ for all $i \in \{1,\dots, n\}$.
\end{definition}

If $(\cal{X}, d)$ has strong negative type, then $ED_d(X,Y)$ completely characterizes the homogeneity of the distributions of $X$ and $Y$ (see Lyons\,(2013) and Sejdinovic et al.\,(2013) for detailed discussions). This quantification of homogeneity of distributions lends itself for reasonable use in one-sample goodness-of-fit testing and two sample testing for equality of distributions.

On the machine learning side, Gretton et al.\,(2012) proposed a kernel-based metric, namely maximum mean discrepancy (MMD), to conduct two-sample testing for equality of distributions. We provide some background before introducing MMD.

\begin{definition}(RKHS)\label{RKHS}
Let $\cal{H}$ be a Hilbert space of real valued functions defined on some space $\cal{X}$. A bivariate function $\mathcal{K} : \cal{X} \times \cal{X} \to \bb{R}$ is called a reproducing kernel of $\cal{H}$ if :
\begin{enumerate}
\item $\forall x \in \cal{X}, \mathcal{K}(\cdot,x) \in \cal{H}$\,
\item $\forall x \in \cal{X}, \forall f \in \cal{H},\; \langle f,\mathcal{K}(\cdot,x)\rangle_{\cal{H}} = f(x)$\,
\end{enumerate}
where $\langle \cdot,\cdot\rangle_{\cal{H}}$ is the inner product associated with $\cal{H}$. If $\cal{H}$ has a reproducing kernel, it is said to be a reproducing kernel Hilbert space (RKHS).
\end{definition}

By Moore-Aronszajn theorem, for every positive definite function (also called a kernel) $\mathcal{K}: \cal{X} \times \cal{X} \to \bb{R}$, there is an associated RKHS $\cal{H}_\mathcal{K}$ with the reproducing kernel $\mathcal{K}$. 
The map $\Pi : \cal{M}_1(\cal{X}) \to \cal{H}_\mathcal{K}$, defined as $\Pi(P) = \int_{\cal{X}} \mathcal{K}(\cdot,x) \, dP(x)$\, for $P \in \cal{M}_1(\cal{X})$\, is called the mean embedding function associated with $\mathcal{K}$. A kernel $\mathcal{K}$ is said to be characteristic to $\cal{M}_1(\cal{X})$ if the map $\Pi$ associated with $\mathcal{K}$ is injective. Suppose $\mathcal{K}$ is a characteristic kernel on $\cal{X}$. Then the MMD between $X \sim P_X$ and $Y \sim P_Y$, where $P_X, P_Y \in \cal{M}_1(\cal{X}) \cap \cal{M}^{1/2}_\mathcal{K}(\cal{X})$ is defined as
\begin{align}\label{MMD}
MMD_\mathcal{K}(X,Y)\;&=\; \Vert \,\Pi(P_X) \,-\, \Pi(P_Y) \,\Vert_{\cal{H}_\mathcal{K}}\,.
\end{align}
By virtue of $\mathcal{K}$ being a characteristic kernel, $MMD_\mathcal{K}(X,Y)=0$\, if and only if $X \overset{d}{=} Y$. Lemma 6 in Gretton et al.\,(2012) shows that the squared MMD can be equivalently expressed as
\begin{equation}\label{MMD equiv expr}
MMD^2_\mathcal{K}(X,Y) \;=\; \E \,\mathcal{K}(X,X') \,+\, \E \,\mathcal{K}(Y,Y') \,-\, 2\,\E \,\mathcal{K}(X,Y) \; .
\end{equation}

Theorem 22 in Sejdinovic et al.\,(2013) establishes the equivalence between (generalized) energy distance and MMD. Following is the definition of a kernel induced by a distance metric (refer to Section 4.1 in Sejdinovic et al.\,(2013) for more details).

\begin{definition}(Distance-induced kernel and kernel-induced distance)\label{dist-induced kernel}
Let $(\cal{X}, d)$ be a metric space of negative type and $x_0 \in \cal{X}$. Denote $\mathcal{K}: \cal{X} \times \cal{X} \to \bb{R}$ as
\begin{align}\label{def:dist-induced kernel}
\mathcal{K}(x,x')\; =\; \frac{1}{2}\, \left\{d(x,x_0) + d(x',x_0) - d(x,x')\right\} .
\end{align}
The kernel $\mathcal{K}$ is positive definite if and only if $(\cal{X}, d)$ has negative type, and thus $\mathcal{K}$ is a valid kernel on $\cal{X}$ whenever $d$ is a metric of negative type. The kernel $\mathcal{K}$ defined in (\ref{def:dist-induced kernel}) is said to be the distance-induced kernel induced by $d$ and centered at $x_0$. One the other hand, the distance $d $ can be generated by the kernel $\mathcal{K}$ through
\begin{align}\label{eq-dist}
d(x,x')=\mathcal{K}(x,x)+\mathcal{K}(x',x')-2\mathcal{K}(x,x').
\end{align}
\end{definition}

Proposition 29 in Sejdinovic et al.\,(2013) establishes that the distance-induced kernel $\mathcal{K}$ induced by $d$ is characteristic to $\cal{M}_1(\cal{X}) \cap \cal{M}^{1}_\mathcal{K}(\cal{X})$ if and only if $(\cal{X}, d)$ has strong negative type. Therefore, MMD can be viewed as a special case of the generalized energy distance in (\ref{ed def general}) with $d$ being the metric induced by a characteristic kernel.


Suppose $\{X_i\}^{n}_{i=1}$ and $\{Y_i\}^{m}_{i=1}$ are i.i.d samples of $X$ and $Y$ respectively. A U-statistic type estimator of $E_d(X,Y)$ is defined as
\begin{align}\label{unif U est ED}
E_{n,m}(X,Y)=\frac{2}{n m}\sum_{k=1}^{n}\sum^{m}_{l=1}d(X_k, Y_l)-\frac{1}{n(n-1)}\sum_{k\neq l}^{n}d(X_k,X_l)-\frac{1}{m(m-1)}\sum_{k\neq l}^{m}d(Y_k,Y_l)\,.
\end{align}

In Section \ref{sec:new-homo}, we shall propose a new class of metrics for quantifying the homogeneity of high-dimensional distributions. This new class can be viewed as a particular case of the general measures in (\ref{ed def general}) with a suitably chosen distance $d$ to accommodate the high dimensionality. It thus inherits all the nice properties of $E_{d}(X, Y)$ in the low-dimensional context (see Proposition \ref{prop ED U stat} and Theorem \ref{th ED U stat} in the supplementary material). With the specific choice of distance, the new metrics can detect a broader range of inhomogeneity between high-dimensional distributions compared to Euclidean energy distance.

\subsection{Distance covariance and HSIC}
Distance covariance (dCov) was first introduced in the seminal paper by Sz\'ekely et al. (2007) to quantify the dependence between two random vectors of arbitrary (fixed) dimensions. Consider two random vectors $X\in \bb{R}^p$ and $Y \in \bb{R}^q$ with $\E \Vert X \Vert_p < \infty$ and $\E \Vert Y \Vert_q < \infty$. The Euclidean dCov between $X$ and $Y$ is defined as the positive square root of
\begin{eqnarray*}
dCov^2(X,Y)=\frac{1}{c_{p}c_{q}}\int_{{\mathbb
R}^{p+q}}\frac{|f_{X,Y}(t,s)-f_X(t)f_Y(s)|^2}{\Vert t \Vert_p^{1+p}\,\Vert s \Vert_q^{1+q}}dtds,
\end{eqnarray*}
where $f_X$, $f_Y$ and $f_{X,Y}$ are the individual and joint
characteristic functions of $X$ and $Y$ respectively, and, $c_{p}=\pi^{(1+p)/2}/\,\Gamma((1+p)/2)$
is a constant with $\Gamma(\cdot)$ being the complete gamma function.

The key feature of dCov is that it completely
characterizes independence between two random vectors of arbitrary dimensions, or in other words $dCov(X,Y)=0$ if and
only if $X \bigCI Y$. According to Remark 3 in Sz\'ekely et al.\,(2007), dCov can be equivalently expressed as
\begin{equation}\label{alt dcov}
dCov^2(X,Y) \;=\; \E\,\Vert X-X'\Vert_p \Vert Y-Y'\Vert_q \,+\, \E\,\Vert X-X'\Vert_p\,\E\,\Vert Y-Y'\Vert_q \,-\, 2 \,\E\,\Vert X-X'\Vert_p \Vert Y-Y''\Vert_q. \\
\end{equation}
Lyons\,(2013) extends the notion of dCov from Euclidean spaces to general metric spaces. For arbitrary metric spaces $(\cal{X}, d_{\cal{X}})$ and $(\cal{Y}, d_{\cal{Y}})$, the generalized dCov between $X \sim P_X \in \cal{M}_1(\cal{X}) \cap \cal{M}^1_{d_{\cal{X}}}(\cal{X})$\, and \,$Y \sim P_Y \in \cal{M}_1(\cal{Y}) \cap \cal{M}^1_{d_{\cal{Y}}}(\cal{Y})$ is defined as
\begin{align}\label{dCov Lyons}
D^2_{d_{\cal{X}}, d_{\cal{Y}}} (X,Y) \;=\; \E\,d_{\cal{X}}(X,X') d_{\cal{Y}}(Y,Y') \,+\, \E\,d_{\cal{X}}(X,X')\,\E\,d_{\cal{Y}}(Y,Y') \,-\, 2 \,\E\,d_{\cal{X}}(X,X') d_{\cal{Y}}(Y,Y'').
\end{align}
Theorem 3.11 in Lyons\,(2013) shows that if $(\cal{X}, d_{\cal{X}})$ and $(\cal{Y}, d_{\cal{Y}})$ are both metric spaces of strong negative type, then $D_{d_{\cal{X}}, d_{\cal{Y}}} (X,Y)=0$ if and only if $X \bigCI Y$. In other words, the complete characterization of independence by dCov holds true for any metric spaces of strong negative type. According to Theorem 3.16 in Lyons\,(2013), every separable Hilbert space is of strong negative type. As Euclidean spaces are separable Hilbert spaces, the characterization of independence by dCov between two random vectors in $(\bb{R}^p, \Vert \cdot \Vert_p)$ and $(\bb{R}^q, \Vert \cdot \Vert_q)$ is just a special case.

Hilbert-Schmidt Independence Criterion (HSIC) was introduced as a kernel-based independence measure by Gretton et al.\,(2005, 2007). Suppose $\cal{X}$ and  $\cal{Y}$ are arbitrary topological spaces, $\mathcal{K}_{\cal{X}}$ and $\mathcal{K}_{\cal{Y}}$ are characteristic kernels on $\cal{X}$ and $\cal{Y}$ with the respective RKHSs $\cal{H}_{\mathcal{K}_{\cal{X}}}$ and $\cal{H}_{\mathcal{K}_{\cal{Y}}}$. Let $\mathcal{K} = \mathcal{K}_{\cal{X}} \otimes \mathcal{K}_{\cal{Y}}$ be the tensor product of the kernels $\mathcal{K}_{\cal{X}}$ and $\mathcal{K}_{\cal{Y}}$, and, $\cal{H}_\mathcal{K}$ be the tensor product of the RKHSs $\cal{H}_{\mathcal{K}_{\cal{X}}}$ and $\cal{H}_{\mathcal{K}_{\cal{Y}}}$. The HSIC between $X \sim P_X \in \cal{M}_1(\cal{X}) \cap \cal{M}^{1/2}_\mathcal{K}(\cal{X})$\, and \,$Y \sim P_Y \in \cal{M}_1(\cal{Y}) \cap \cal{M}^{1/2}_\mathcal{K}(\cal{Y})$ is defined as
\begin{align}\label{HSIC}
HSIC_{\mathcal{K}_{\mathcal{X}},\mathcal{K}_\mathcal{Y}}(X,Y) \;&=\; \Vert \,\Pi(P_{XY}) \,-\, \Pi(P_X P_Y) \,\Vert_{\cal{H}_\mathcal{K}},
\end{align}
where $P_{XY}$ denotes the joint probability distribution of $X$ and $Y$. The HSIC between $X$ and $Y$ is essentially the MMD between the joint distribution $P_{XY}$ and the product of the marginals $P_X$ and $P_Y$. Clearly, $HSIC_{\mathcal{K}_{\mathcal{X}},\mathcal{K}_\mathcal{Y}}(X,Y) = 0$ if and only if $X \bigCI Y$. Gretton et al.\,(2005) shows that the squared HSIC can be equivalently expressed as
\begin{align}\label{HSIC equiv def}
HSIC^2_{\mathcal{K}_{\mathcal{X}},\mathcal{K}_\mathcal{Y}}(X,Y) \;&=\; \E\,\mathcal{K}_{\cal{X}}(X,X') \mathcal{K}_{\cal{Y}}(Y,Y') \,+\, \E\,\mathcal{K}_{\cal{X}}(X,X')\,\E\,\mathcal{K}_{\cal{Y}}(Y,Y') \,-\, 2 \,\E\,\mathcal{K}_{\cal{X}}(X,X') \mathcal{K}_{\cal{Y}}(Y,Y'').
\end{align}
Theorem 24 in Sejdinovic et al.\,(2013) establishes the equivalence between the generalized dCov and HSIC.

For an observed random sample $(X_i,Y_i)^{n}_{i=1}$ from the joint distribution of $X$ and $Y$, a U-statistic type estimator of the generalized dCov in (\ref{dCov Lyons}) can be defined as
\begin{equation}\label{ustat dcov}
\widetilde{D_n^2}_{\,;\,d_{\cal{X}},d_{\cal{Y}}}(X,Y)\; =\; (\tilde{A} \cdot \tilde{B})\;=\; \frac{1}{n(n-3)} \dis \sum_{k\neq l} \tilde{a}_{kl} \tilde{b}_{kl} \; ,
\end{equation}
where $\tilde{A}, \tilde{B}$ are the \,$\cal{U}$-centered versions (see (\ref{U centering def})) of $A=\big(d_{\cal{X}}(X_k, X_l)\big)_{k,l=1}^n$ and $B=\big(d_{\cal{Y}}(Y_k, Y_l)\big)_{k,l=1}^n$, respectively. We denote $\widetilde{D_n^2}_{\,;\,d_{\cal{X}},d_{\cal{Y}}}(X,Y)$ by $dCov^2_n(X,Y)$  when $d_{\cal{X}}$ and $d_{\cal{Y}}$ are Euclidean distances.

\section{New distance for Euclidean space}\label{new distance}
We introduce a family of distances for Euclidean space, which shall play a central role in the subsequent developments.
For $x \in \mathbb{R}^{\tilde{p}}$, we partition $x$ into $p$ sub-vectors or groups, namely $x=(x_{(1)},\dots,x_{(p)})$,
where $x_{(i)}\in \mathbb{R}^{d_i}$ with $\sum_{i=1}^{p}d_i=\tilde{p}$. Let $\rho_i$ be a metric or semimetric (see for example Definition 1 in Sejdinovic et al.\,(2013)) defined on $\mathbb{R}^{d_i}$
for $1\leq i\leq p$. We define a family of distances for $\mathbb{R}^{\tilde{p}}$ as
\begin{equation}\label{Kdef}
K_{\bf{d}}(x,x') \, := \, \sqrt{ \,\rho_1 (x_{(1)}, x_{(1)}') \, + \,\dots \,+ \,\rho_p (x_{(p)}, x_{(p)}') \,}\,,
\end{equation}
where $x, x' \in \bb{R}^{\tilde{p}}$ with $x = (x_{(1)},\dots,x_{(p)})$ and $x' = (x'_{(1)},\dots,x'_{(p)})$, and $\textbf{d}=(d_1,d_2,\dots,d_p)$ with $d_i\in\mathbb{Z}_+$ and $\sum_{i=1}^{p}d_i=\tilde{p}$. 

\begin{proposition}\label{metric} Suppose each $\rho_i$ is a metric of strong negative type on $\mathbb{R}^{d_i}$. Then $\left(\bb{R}^{\tilde{p}}, K_{\bf{d}}\right)$ satisfies the following two properties:
\begin{enumerate}
\item $K_{\bf{d}}:\bb{R}^{\tilde{p}} \times \bb{R}^{\tilde{p}} \rightarrow [0,\infty)$ is a valid metric on $\bb{R}^{\tilde{p}}$;
\item $\left(\bb{R}^{\tilde{p}},K_{\bf{d}}\right)$ has strong negative type.
\end{enumerate}
\end{proposition}
In a special case, suppose\, $\rho_i$\, is the Euclidean distance on $\bb{R}^{d_i}$. By Theorem 3.16 in Lyons\,(2013), $(\bb{R}^{d_i}, \rho_i)$ is a separable Hilbert space, and hence has strong negative type. Then the Euclidean space equipped with the metric
\begin{align}\label{sp case}
K_{\bf{d}}(x,x') \, = \, \sqrt{ \,\Vert x_{(1)} - x_{(1)}' \Vert \, + \,\dots \,+ \,\Vert x_{(p)} - x_{(p)}' \Vert \,}\,.
\end{align}
is of strong negative type. Further, if all the components $x_{(i)}$ are unidimensional, i.e., $d_i = 1$ for $1 \leq i \leq p$, then the metric boils down to
\begin{equation}\label{Kdef_sp}
K_{\bf{d}}(x,x') \;= \; \Arrowvert x-x' \Arrowvert_1^{1/2} \;= \;\sqrt{\dis\sum_{j=1}^p |x_j - x'_j|} \; ,
\end{equation}
where $\Arrowvert x \Arrowvert_{1} = \sum_{j=1}^p |x_j|$\, is the $l_1$ or the absolute norm on $\bb{R}^p$. If
\begin{align}\label{rho for Eucl}
\rho_i(x_{(i)},x_{(i)}')=\Vert x_{(i)} - x_{(i)}' \Vert^2, \quad 1\leq i\leq p,
\end{align}
then $K_{\bf{d}}$ reduces to the usual Euclidean distance. We shall unify the analysis of our new metrics with the classical metrics by considering $K_{\bf{d}}$ which is defined in (\ref{Kdef}) with
\begin{enumerate}
  \item[S1] each $\rho_i$ being a metric of strong negative type on $\mathbb{R}^{d_i}$;
  \item[S2] each $\rho_i$ being a semimetric defined in (\ref{rho for Eucl}).
\end{enumerate}
The first case corresponds to the newly proposed metrics while the second case leads to the classical metrics based on the usual Euclidean distance. Remarks \ref{rem3.1} and \ref{rem3.2} provide two different ways of generalizing the class in (\ref{Kdef}). To be focused, our analysis below shall only concern about the distances defined in (\ref{Kdef}). In the numerical studies in Section \ref{sec:num}, we consider $\rho_i$ to be the Euclidean distance and the distances induced by the Laplace and Gaussian kernels (see Definition \ref{dist-induced kernel}) which are of strong negative type on $\mathbb{R}^{d_i}$ for $1\leq i \leq p$.

\begin{remark}\label{rem3.1}
A more general family of distances can be defined as
\begin{align*}
K_{\mathbf{d},r}(x,x')=\Big(\rho_1(x_{(1)},x_{(1)}')+\cdots+\rho_p(x_{(p)},x_{(p)}')\Big)^{r}, \quad 0<r<1.
\end{align*}
According to Remark 3.19 of Lyons (2013), the space $(\mathbb{R}^{\tilde{p}},K_{\mathbf{d},r})$ is of strong negative type. The proposed distance is a special case with $r=1/2.$
\end{remark}

\begin{remark}\label{rem3.2}
Based on the proposed distance, one can construct the generalized Gaussian and Laplacian  kernels as
$$f(K_{\mathbf{d}}(x,x')/\gamma)=\begin{cases}
\exp(-K_{\mathbf{d}}^2(x,x')/\gamma^2), \quad & f(x)=\exp(-x^2) \text{ for Gaussian kernel},\\
\exp(-K_{\mathbf{d}}(x,x')/\gamma), \quad & f(x)=\exp(-x) \text{ for Laplacian kernel}.
\end{cases}$$
If $K_{\bf{d}}$ is translation invariant, then by Theorem 9 in Sriperumbudur et al.\,(2010) it can be verified that $f(K_{\mathbf{d}}(x,x')/\gamma)$ is a characteristic kernel on $\mathbb{R}^{\tilde{p}}$. As a consequence, the Euclidean space equipped with the distance
$$K_{\mathbf{d},f}(x,x')=f(K_{\mathbf{d}}(x,x)/\gamma)+f(K_{\mathbf{d}}(x',x')/\gamma)-2f(K_{\mathbf{d}}(x,x')/\gamma)$$
is of strong negative type. 
\end{remark}

\begin{remark}\label{rem3.3}
In Sections \ref{sec:new-homo} and \ref{sec:ACdcov} we develop new classes of homogeneity and dependence metrics to quantify the pairwise homogeneity of distributions or the pairwise non-linear dependence of the low-dimensional groups. A natural question to arise in this regard is how to partition the random vectors optimally in practice. We present some real data examples in Section \ref{sub:real} of the main paper where all the group sizes have been considered to be one (as a special case of the general theory proposed in this paper), and an additional real data example in Section \ref{addl data ex} of the supplement where the data admits some natural grouping. We believe this partitioning can be very much problem specific and may require subject knowledge. We leave it for future research to develop an algorithm to find the optimal groups using the data and perhaps some auxiliary information.  
\end{remark}

\section{Homogeneity metrics}\label{sec:new-homo}
Consider $X, Y \in \bb{R}^{\tilde{p}}$. Suppose $X$ and $Y$ can be partitioned into $p$ sub-vectors or groups, viz. $X = \left(X_{(1)}, X_{(2)}, \dots, X_{(p)} \right)$ and $Y = \left(Y_{(1)}, Y_{(2)}, \dots, Y_{(p)} \right)$, where the groups $X_{(i)}$ and $Y_{(i)}$ are $d_i$ dimensional, $1\leq i\leq p$, and $p$ might be fixed or growing. We assume that $X_{(i)}$ and $Y_{(i)}$'s are finite (low) dimensional vectors, i.e., $\{d_i\}_{i=1}^p$ is a bounded sequence. Clearly $\tilde{p} = \sum_{i=1}^p d_i = O(p)$.
Denote the mean vectors and the covariance matrices of $X$ and $Y$ by $\mu_X$ and $\mu_Y$, and, $\Sigma_X$ and $\Sigma_Y$, respectively. We propose the following class of metrics $\cal{E}$ to quantify the homogeneity of the distributions of $X$ and $Y$:
\begin{align}\label{ED HD def}
\cal{E}(X,Y) \;=\;  2\,\E \,K_{\bf{d}}(X,Y)\, -\, \E \,K_{\bf{d}}(X,X') \,-\, \E \,K_{\bf{d}}(Y,Y') \; ,
\end{align}
with $\textbf{d}=(d_1,\dots,d_p)$. We shall drop the subscript $\textbf{d}$ below for the ease of notation.

\begin{assumption}\label{ass0}
Assume that $\sup_{1\leq i\leq p}\E \rho_i^{1/2}(X_{(i)},0_{d_i})< \infty$ and $\sup_{1\leq i\leq p}\E \rho_i^{1/2}(Y_{(i)},0_{d_i})< \infty$.
\end{assumption}

Under Assumption \ref{ass0}, $\cal{E}$ is finite. In Section \ref{ld E} of the supplement we illustrate that in the low-dimensional setting, $\cal{E}(X,Y)$ completely characterizes the homogeneity of the distributions of $X$ and $Y$. 

Consider i.i.d. samples $\{X_k\}^{n}_{k=1}$ and $\{Y_l\}^{m}_{l=1}$ from the respective distributions of $X$ and $Y \in \mathbb{R}^{\tilde{p}}$, where $X_k = (X_{k(1)}, \dots, X_{k(p)})$, $Y_l=(Y_{l(1)},\dots,Y_{l(p)})$ for $1 \leq k \leq n$, $1\leq l \leq m$ and $X_{k(i)}, Y_{l(i)}\in\mathbb{R}^{d_i}.$ We propose an unbiased U-statistic type estimator $\cal{E}_{n,m}(X,Y)$ of $\cal{E}(X,Y)$ as in equation (\ref{unif U est ED}) with $d$ being the new metric $K$. We refer the reader to Section \ref{ld E} of the supplement, where we show that $\cal{E}_{n,m}(X,Y)$ essentially inherits all the nice properties of the U-statistic type estimator of generalized energy distance and MMD.

We define the following quantities which will play an important role in our subsequent analysis:
\begin{align}\label{tau def original}
\tau_X^2=\E \,K(X,X')^2,\quad \tau_Y^2=\E \,K(Y,Y')^2,\quad \tau^2=\E \,K(X,Y)^2.
\end{align}
In Case S2 (i.e., when $K$ is the Euclidean distance), we have
\begin{align}\label{tau for Eucl}
\tau_X^2 = 2\textrm{tr}\Sigma_X,\quad \tau_Y^2 = 2\textrm{tr}\Sigma_Y,\quad \tau^2 = \textrm{tr}\Sigma_X + \textrm{tr} \Sigma_Y + \Vert \mu_X - \mu_Y \Vert^2.
\end{align}
Under the null hypothesis $H_0 : X \overset{d}{=} Y$, it is clear that $\tau_X^2 = \tau_Y^2 = \tau^2$. 

In the subsequent discussion we study the asymptotic behavior of $\cal{E}$ in the high-dimensional framework, i.e., when $p$ grows to $\infty$ with fixed $n$ and $m$ (discussed in Subsection \ref{subsec:ED HDLSS}) and when $n$ and $m$ grow to $\infty$ as well (discussed in Subsection \ref{subsec:ED HDMSS} in the supplement). We point out some limitations of the test for homogeneity of distributions in the high-dimensional setup based on the usual Euclidean energy distance. Consequently we propose a test based on the proposed metric and justify its consistency for growing dimension.

\subsection{High dimension low sample size (HDLSS)}\label{subsec:ED HDLSS}
In this subsection, we study the asymptotic behavior of the Euclidean energy distance and our proposed metric $\cal{E}$ when the dimension grows to infinity while the sample sizes $n$ and $m$ are held fixed. We make the following moment assumption.

\begin{assumption}\label{ass2 : ED}
There exist constants $a, a', a'', A, A', A''$ such that uniformly over $p$,
\begin{align*}
&\; 0 < a \leq \dis \inf_{1\leq i \leq p} \E \,\rho_i( X_{(i)}, X_{(i)}'\,) \leq \dis \sup_{1\leq i \leq p} \E \,\rho_i( X_{(i)}, X_{(i)}'\,) \leq A < \infty,\\
&\; 0 < a' \leq \dis \inf_{1\leq i \leq p} \E \,\rho_i( Y_{(i)}, Y_{(i)}'\,) \leq \dis \sup_{1\leq i \leq p} \E \,\rho_i( Y_{(i)}, Y_{(i)}'\,) \leq A' < \infty,\\
&\; 0 < a'' \leq \dis \inf_{1\leq i \leq p} \E \,\rho_i( X_{(i)}, Y_{(i)}\,) \leq \dis \sup_{1\leq i \leq p} \E \,\rho_i( X_{(i)}, Y_{(i)}\,) \leq A'' < \infty.
\end{align*}
\end{assumption}
Under Assumption \ref{ass2 : ED}, it is not hard to see that $\tau_X, \tau_Y, \tau \asymp p^{1/2}$. The proposition below provides an expansion for $K$ evaluated at random samples.

\begin{proposition}\label{K taylor : ED}
Under Assumption \ref{ass2 : ED}, we have
\begin{align}
&\frac{K(X,X')}{\tau_X} = 1 + \frac{1}{2} L_X(X,X') + R_X(X,X'),\\
&\frac{K(Y,Y')}{\tau_Y} = 1 + \frac{1}{2} L_Y(Y,Y') + R_Y(Y,Y'),
\end{align}
and
\begin{equation}
\frac{K(X,Y)}{\tau} = 1 + \frac{1}{2} L(X,Y) + R(X,Y),
\end{equation}
where $$L_X(X,X') := \frac{K^2(X,X') - \tau_X^2}{\tau_X^2}, \;\; L_Y(Y,Y') := \frac{K^2(Y,Y') - \tau_Y^2}{\tau_Y^2},\;\;L(X,Y) := \frac{K^2(X,Y) - \tau^2}{\tau^2}, $$ and \,$R_X(X,X'), R_Y(Y,Y'), R(X,Y)$ are the remainder terms. In addition, if $L_X(X,X'), L_Y(Y,Y')$ and $L(X,Y)$ are $o_p(1)$ random variables as $p \to \infty$, then \,$R_X(X,X') = O_p\left(L^2_X(X,X')\right)$, $R_Y(Y,Y') = O_p\left(L^2_Y(Y,Y')\right)$ and $R(X,Y) = O_p\left(L^2(X,Y)\right)$.

\end{proposition}
Henceforth we will drop the subscripts $X$ and $Y$ from $L_X, L_Y, R_X$ and $R_Y$ for notational convenience. Theorem \ref{th homo decomp} and Lemma \ref{lemma 1} below provide insights into the behavior of $\cal{E}(X,Y)$ in the high-dimensional framework.

\begin{assumption}\label{ass0.6}
Assume that $L(X,Y) = O_p(a_{p})$, $L(X,X') = O_p(b_{p})$ and $L(Y,Y') = O_p(c_{p})$, where $a_{p}, b_{p}, c_{p}$ are positive real sequences satisfying $a_{p}=o(1)$, $b_{p}=o(1)$, $c_{p}=o(1)$ and $\tau a^{2}_{p} + \tau_X b^{2}_{p} + \tau_Y c^{2}_{p} = o(1)$. 
\end{assumption}

\begin{remark}\label{assumption justification}
To illustrate Assumption \ref{ass0.6}, 
we observe that under assumption \ref{ass2 : ED} we can write 
\begin{align*}
\var \left(L(X,X')\right) \;&=\; O\Big(\frac{1}{p^2}\Big) \dis \sum_{i,j=1}^p \text{cov}\left( \rho_i( X_{(i)}, X_{(i)}')\,, \rho_j( X_{(j)}, X_{(j)}') \right) \;=\; O\Big(\frac{1}{p^2}\Big) \dis \sum_{i,j=1}^p \text{cov}\left( Z_i, Z_j\right)\,,
\end{align*}
where $Z_i := \rho_i( X_{(i)}, X_{(i)}')$ for $1\leq i \leq p$. Assume that $ \sup_{1\leq i \leq p} \E\,\rho_i^2( X_{(i)}, 0_{d_i}) < \infty$, which implies $ \sup_{1\leq i \leq p} \E\,Z_i^2 < \infty$. Under certain strong mixing conditions or in general certain weak dependence assumptions, it is not hard to see that $\sum_{i,j=1}^p \text{cov}\left( Z_i, Z_j\right) = O(p)$ as $p\to \infty$ (see for example Theorem 1.2 in Rio\,(1993) or Theorem 1 in Doukhan et al.\,(1999)). Therefore we have $\var \left(L(X,X')\right) = O(\frac{1}{p})$ and hence by Chebyshev's inequality, we have $L(X,X') = O_p(\frac{1}{\sqrt{p}})$. We refer the reader to Remark 2.1.1 in Zhu et al.\,(2019) for illustrations when each $\rho_i$ is the squared Euclidean distance.
\end{remark}

\begin{theorem}\label{th homo decomp}
Suppose Assumptions \ref{ass2 : ED} and \ref{ass0.6} hold. Further assume that the following three sequences
$$\left\{\frac{\sqrt{p}L^2(X,Y)}{1+L(X,Y)}\right\},~\left\{\frac{\sqrt{p} L^2(X,X')}{1+L(X,X')}\right\},~\left\{\frac{\sqrt{p} L^2(Y,Y')}{1+L(Y,Y')}\right\}$$
indexed by $p$ are all uniformly integrable.
Then we have
\begin{equation}\label{th homo decomp eqn}
\cal{E}(X,Y) \;=\; 2\tau - \tau_X -\tau_Y \,+\,o(1).
\end{equation}
\end{theorem}

\begin{remark}\label{refer to rem_ui}
Remark \ref{rem_ui} in the supplementary materials provides some illustrations on certain sufficient conditions under which $\{\sqrt{p}L^2(X,Y)/(1+L(X,Y))\}$, $\{\sqrt{p}L^2(X,X')/(1+L(X,X'))\}$ and\\ $\{\sqrt{p}L^2(Y,Y')/(1+L(Y,Y'))\}$ are uniformly integrable.
\end{remark}

\begin{remark}\label{pop approx E}
To illustrate that the leading term in equation (\ref{th homo decomp eqn}) indeed gives a close approximation of the population $\cal{E} (X,Y)$, we consider the special case when $K$ is the Euclidean distance. Suppose $X \sim N_p(0,I_p)$ and $Y=X+N$ where $N \sim N_p(0,I_p)$ with $N \bigCI X$. Clearly from (\ref{tau for Eucl}) we have $\tau_X^2=2p$, $\tau_Y^2=4p$ and $\tau^2=3p$. We simulate large samples of sizes $m=n=5000$ from the distributions of $X$ and $ Y$ for $p=20, 40, 60, 80$ and $100$. The large sample sizes are to ensure that the U-statistic type estimator of $\cal{E} (X,Y)$ gives a very close approximation of the population $\cal{E} (X,Y)$. In Table \ref{table:pop approx E} we list the ratio between $\cal{E} (X,Y)$ and the leading term in (\ref{th homo decomp eqn}) for the different values of $p$, which turn out to be very close to $1$, demonstrating that the leading term in (\ref{th homo decomp eqn}) indeed approximates $\cal{E} (X,Y)$ reasonably well.

\begin{table}[!h]\footnotesize
\centering
\caption{Ratio of $\cal{E} (X,Y)$ and the leading term in (\ref{th homo decomp eqn}) for different values of $p$.}
\label{table:pop approx E}
\begin{tabular}{ccccc}
\toprule
$p=20$ & $p=40$ & $p=60$ & $p=80$ & $p=100$  \\
\hline
0.995 & 0.987 & 0.992 & 0.997 & 0.983 \\
\toprule
\end{tabular}
\end{table}

\end{remark}

\begin{lemma}\label{lemma 1}
Assume $\tau,\tau_X, \tau_Y < \infty$. We have
\begin{enumerate}
\item In Case S1, $2\tau -\tau_X - \tau_Y = 0$ if and only if $X_{(i)} \overset{d}{=} Y_{(i)}$ for $i \in \{1, \dots, p\}$;
\item In Case S2, $2\tau -\tau_X - \tau_Y = 0$ if and only if $\mu_X = \mu_Y$ and $\textrm{tr}\, \Sigma_X = \textrm{tr}\, \Sigma_Y$.
\end{enumerate}
\end{lemma}

It is to be noted that assuming $\tau,\tau_X, \tau_Y < \infty$ does not contradict with the growth rate $\tau, \tau_X, \tau_Y=O(p^{1/2})$. Clearly under $H_0$, $2\tau -\tau_X - \tau_Y = 0$ irrespective of the choice of $K$. In view of Lemma \ref{lemma 1} and Theorem \ref{th homo decomp}, in Case S2, the leading term of $\cal{E}(X,Y)$ becomes zero if and only if $\mu_X = \mu_Y$ and $\textrm{tr}\, \Sigma_X = \textrm{tr}\, \Sigma_Y$. In other words, when dimension grows high, the Euclidean energy distance can only capture the equality of the means and the first spectral means, whereas our proposed metric captures the pairwise homogeneity of the low dimensional marginal distributions of $X_{(i)}$ and $Y_{(i)}$. Clearly $X_{(i)} \overset{d}{=} Y_{(i)}$ for $1\leq i\leq p$ implies $\mu_X = \mu_Y$ and $\textrm{tr}\, \Sigma_X = \textrm{tr}\, \Sigma_Y$. 
Thus the proposed metric can capture a wider range of inhomogeneity of distributions than the Euclidean energy distance.

Define
\begin{align*}
d_{kl}(i):=\rho_i(X_{k(i)}, Y_{l(i)}) \,-\, \E\,\left[\rho_i(X_{k(i)}, Y_{l(i)})|X_{k(i)}\right] \,-\, \E\,\left[\rho_i(X_{k(i)}, Y_{l(i)})|Y_{l(i)}\right] \,+\, \E\,\left[\rho_i(X_{k(i)}, Y_{l(i)})\right],
\end{align*}
as the double-centered distance between $X_{k(i)}$ and $Y_{l(i)}$
for $1\leq i \leq p$, $1\leq k\leq n$ and $1\leq l\leq m$. Similarly define  $d^X_{kl}(i)$ and $d^Y_{kl}(i)$ as the double-centered distances between $X_{k(i)}$ and $X_{l(i)}$ for $1\leq k \neq l\leq n$, and, $Y_{k(i)}$ and $Y_{l(i)}$ for $1\leq k \neq l\leq m$, respectively. Further define $H(X_k, Y_l) := \frac{1}{\tau}  \sum_{i=1}^p d_{kl}(i)$ for $1\leq k\leq n\,,\, 1\leq l \leq m$, $H(X_k, X_l) := \frac{1}{\tau_X}  \sum_{i=1}^p d^X_{kl}(i)$ for $1\leq k \neq l \leq n$ and $H(Y_k, Y_l)$ in a similar way.


We impose the following conditions to study the asymptotic behavior of the (unbiased) U-statistic type estimator of $\cal{E}(X,Y)$ in the HDLSS setup.

\begin{assumption}\label{ass6}
For fixed $n$ and $m$, as $p \to \infty$, $$\begin{pmatrix}
H(X_k, Y_l) \\
H(X_s, X_t) \\
H(Y_u, Y_v)
\end{pmatrix}_{k,l,\, s<t,\, u<v} \overset{d}{\longrightarrow} \;\;\begin{pmatrix}
a_{kl} \\  b_{st} \\ c_{uv}
\end{pmatrix}_{k,l,\, s<t,\, u<v} \, ,$$ where $\{a_{kl}, b_{st}, c_{uv}\}_{k,l,\, s<t,\, u<v}$ are jointly  Gaussian with zero mean. Further we assume that
\begin{align*}
var(a_{kl}) \; &:= \; \sigma^2 \; =\; \lim_{p \to \infty}\, \E \left[H^2(X_k, Y_l)\right],\\
var(b_{st}) \; &:= \; \sigma^2_X \; =\; \lim_{p \to \infty}\, \E \left[H^2(X_s, X_t)\right],\\
var(c_{uv}) \; &:= \; \sigma^2_Y \; =\; \lim_{p \to \infty}\, \E \left[H^2(Y_u, Y_v)\right].
\end{align*}
$\{a_{kl}, b_{st}, c_{uv}\}_{k,l,\, s<t,\, u<v}$ are all independent with each other.
\end{assumption}
Due to the double-centering property and the independence between the two samples, it is straightforward to verify that $\{H(X_k,Y_l), H(X_s,X_t), H(Y_u,Y_v)\}_{k,l,s<t,u<t}$
are uncorrelated with each other. So it is natural to expect that the limit $\{a_{kl}, b_{st}, c_{uv}\}_{k,l,\, s<t,\, u<v}$ are all independent with each other.
\begin{remark}
The above multi-dimensional central limit theorem is classic and can be derived under suitable moment and weak dependence assumptions on the components of $X$ and $Y$, such as mixing or near epoch dependent conditions. We refer the reader to Doukhan and Neumann\,(2008) for a review on central limit theorem results under weak dependence assumptions.
\end{remark}


We describe a new two-sample t-test for testing the null hypothesis $H_0 : X \overset{d}{=} Y.$ The t statistic can be constructed based on either the Euclidean energy distance or the new homogeneity metrics. We show that the t-tests based on different metrics can have strikingly different power behaviors under the HDLSS setup. The major difficulty here is to introduce a consistent and computationally efficient variance estimator.
Towards this end, we define a quantity called Cross Distance Covariance (cdCov) between $X$ and $Y$, which plays an important role in the construction of the t-test statistic:
\begin{align*}
cdCov^2_{n,m}(X,Y):=\frac{1}{(n-1)(m-1)}\sum^{n}_{k=1}\sum_{l=1}^{m}\widehat{K}(X_k,Y_l)^2,
\end{align*} where \begin{align*}
\widehat{K}(X_k, Y_l)\;=\;K(X_{k},Y_{l})-\frac{1}{n}\sum^{n}_{i=1}K(X_{i},Y_{l})-\frac{1}{m}\sum^{m}_{j=1}K(X_{k},Y_{j})+\frac{1}{nm}\sum_{i=1}^{n}\sum^{m}_{j=1}K(X_{i},Y_{j}).
\end{align*}
Let $v_s := s(s-3)/2$ for $s = m, n$. We introduce the following quantities
\begin{align}\label{quantities}
\begin{split}
m_0\; &:=\; \frac{\sigma^2\,(n-1)(m-1) + \sigma_X^2 \,v_n + \sigma_Y^2 \,v_m}{(n-1)(m-1)+v_n+v_m}\,,\\
\sigma_{nm} \;&:=\; \sqrt{\frac{\sigma^2}{nm} + \frac{\sigma^2_X}{2n(n-1)} + \frac{\sigma^2_Y}{2m(m-1)}}\, ,
\\ a_{nm} \;&:=\; \sqrt{\frac{1}{nm} + \frac{1}{2n(n-1)} + \frac{1}{2m(m-1)}} \,,\\
\Delta\;&:=\; \lim_{p \to \infty} 2\tau -\tau_X - \tau_Y,
\end{split}
\end{align} 
where $\sigma^2,\sigma_X^2$ and $\sigma_Y^2$ are defined in Assumption \ref{ass6}. Under Assumption \ref{ass0.5}, further define
\begin{align*}
&m_0^*:= \dis \lim_{m,n \to \infty} m_0 \; = \; \frac{2\alpha_0 \, \sigma^2 + \sigma_X^2 + \sigma_Y^2 \, \alpha_0^2}{2\alpha_0  + 1 + \alpha_0^2},\\
&a^*_0 := \dis \lim_{m,n \to \infty}\, \frac{a_{nm}}{\sigma_{nm}} \;=\; \Big(\frac{2\alpha_0 + \alpha_0^2 + 1}{2\alpha_0 \, \sigma^2 + \alpha_0^2 \,\sigma_X^2 + \sigma_Y^2 }\Big)^{1/2} \,.
\end{align*}
We are now ready to introduce the two-sample t-test $$T_{n,m}\;:=\;\frac{\cal{E}_{n,m}(X,Y)}{a_{nm}\,\sqrt{S_{n,m}}},$$
where $$S_{n,m}\;:=\;\frac{4(n-1)(m-1)\,cdCov^2_{n,m}(X,Y)\,+\,4v_n\, \widetilde{\cal{D}_n^2}(X,X)\,+\,4v_m\, \widetilde{\cal{D}_n^2}(Y,Y)}{(n-1)(m-1)+v_n+v_m}$$
is the pool variance estimator with  $\widetilde{\cal{D}^2_n}(X,X)$ and $\widetilde{\cal{D}^2_m}(Y,Y)$ being the unbiased estimators of the (squared) distance variances defined in equation (\ref{ustat dcov}). It is interesting to note that the variability of the sample generalized energy distance depends on the distance variances as well as the cdCov.
It is also worth mentioning that the computational complexity of the pool variance estimator and thus the t-statistic is linear in $p$.

To study the asymptotic behavior of the test, we consider the following class of distributions on $(X,Y)$:
\begin{align*}
\mathcal{P}=&\Big\{(P_X,P_Y):~X\sim P_X,~Y\sim P_Y,~E[\tau L(X,Y)-\tau_X L(X,X')|X]=o_p(1),
\\&E[\tau L(X,Y)-\tau_Y L(Y,Y')|Y]=o_p(1)\Big\}.
\end{align*}
If $P_X=P_Y$ (i.e., under the $H_0$), it is clear that $(P_X,P_Y)\in \mathcal{P}$ irrespective of the metrics in the definition of $L$. Suppose $\|X-\mu_X\|^2-\text{tr}(\Sigma_X)=O_p(\sqrt{p})$ and $\|Y-\mu_Y\|^2-\text{tr}(\Sigma_Y)=O_p(\sqrt{p})$, which hold under weak dependence assumptions on the components of $X$ and $Y$. Then in Case S2 (i.e., $K$ is the Euclidean distance), a set of sufficient conditions for $(P_X,P_Y)\in\mathcal{P}$ is given by
\begin{align}
&(\mu_X-\mu_Y)^\top (\Sigma_X+\Sigma_Y)(\mu_X-\mu_Y)=o(p),\quad \tau-\tau_X=o(\sqrt{p}), \quad \tau-\tau_Y=o(\sqrt{p}),
\end{align}
which suggests that the first two moments of $P_X$ and $P_Y$ are not too far away from each other. In this sense, $\mathcal{P}$ defines a class of local alternative distributions (with respect to the null $H_0: P_X=P_Y$).
We now state the main result of this subsection.

\begin{theorem}\label{th KED HDLSS}
In both Cases S1 and S2, under Assumptions \ref{ass2 : ED}, \ref{ass0.6} and \ref{ass6} as $p \to \infty$ with $n$ and $m$ remaining fixed, and further assuming that $(P_X,P_Y)\in\mathcal{P}$, we have
\begin{align*}
&\frac{\cal{E}_{n,m}(X,Y) - (2\tau -\tau_X - \tau_Y)}{a_{nm}\,\sqrt{S_{n,m}}} \; \overset{d}{\longrightarrow} \; \frac{\sigma_{nm} \, Z}{a_{nm}\,\sqrt{M}}\,,
\end{align*}
where $$ M \overset{d}{=} \frac{\sigma^2\,\chi^2_{(n-1)(m-1)} + \sigma_X^2 \chi^2_{v_n} + \sigma_Y^2 \chi^2_{v_m}}{(n-1)(m-1)+v_n+v_m}\,,$$ $\chi^2_{(n-1)(m-1)},\, \chi^2_{v_n},\, \chi^2_{v_m}$ are independent chi-squared random variables, and $Z \sim N(0,1)$. In other words, $$T_{n,m} \;\overset{d}{\longrightarrow} \; \frac{\sigma_{nm} \, N(\Delta/\sigma_{nm}, 1)}{a_{nm}\,\sqrt{M}}\,,$$ where $\sigma_{nm}$ and $a_{nm}$ are defined in equation (\ref{quantities}). In particular, under  $H_0$, we have
\begin{align*}
&T_{n,m}\overset{d}{\longrightarrow} \; t_{(n-1)(m-1)+v_n+v_m}.
\end{align*}
\end{theorem}

Based on the asymptotic behavior of $T_{n,m}$ for growing dimensions, we propose a test for $H_0$ as follows: at level $\alpha \in (0,1)$, reject $H_0$ if $T_{n,m} > q_{\alpha,(n-1)(m-1) + v_n + v_m}$ and fail to reject $H_0$ otherwise, where $P(t_{(n-1)(m-1) + v_n + v_m}> q_{\alpha,(n-1)(m-1) + v_n + v_m})=\alpha.$ For a fixed real number $t$, define
\begin{align}\label{eqn exact power}
\begin{split}
\phi_{n,m}(t) \;:=& \; \dis \lim_{p \to \infty} P(T_{n,m} \leq t) \;=\; \E\,\left[ P \left(\frac{\sigma_{nm} \, N(\Delta/\sigma_{nm}, 1)}{a_{nm}\,\sqrt{M}} \leq t \,\,\Big|\,\, M  \right)\right]\\ =& \;\E\,\,\left[ \Phi\left(\frac{a_{nm}\, \sqrt{M}\, t-\Delta}{\sigma_{nm}} \right)\right]\,.
\end{split}
\end{align}
The asymptotic power curve for testing $H_0$ based on $T_{n,m}$ is given by $1-\phi_{m,n}(t)$. The following proposition gives a large sample approximation of the power curve. 

\begin{assumption}\label{ass0.5}
As $m, n \to \infty$, $m/n \to \alpha_0$\; where \;$\alpha_0 > 0$.
\end{assumption}

\begin{proposition}\label{power}
Suppose $\Delta=\Delta_0/\sqrt{nm}$ where $\Delta_0$ is a constant with respect to $n,m$. Then for any bounded real number $t$ as $n, m \to \infty$ and under Assumption \ref{ass0.5}, we have $$\dis \lim_{m, n \to \infty} \phi_{n,m}(t) \;=\; \Phi \left(a^*_0 \sqrt{m_0^*}\,\, t\, -\, \Delta^*_0 \right) \;, $$ where
$$\Delta^*_0 = \Delta_0 \lim_{m, n \to \infty}\frac{1}{\sigma_{nm} \sqrt{nm}}= \Delta_0 \, \Big(\frac{2\alpha_0}{2\sigma^2\, \alpha_0 + \sigma_X^2\, \alpha_0^2 + \sigma_Y^2}\Big)^{1/2}.$$
\end{proposition}
Under the alternative, if $\Delta_0 \to \infty$ as $n, m \to \infty$, we have $$\dis \lim_{m, n \to \infty} \left\{ 1 - \phi_{n,m} ( q_{\alpha,(n-1)(m-1) + v_n + v_m}) \right\}\; = \; 1,$$ thereby justifying the consistency of the test.

\begin{remark}\label{rm:power}
We first derive the power function $1 - \phi_{n,m}(t)$ under the assumption that $n$ and $m$ are fixed. The main idea behind Proposition \ref{power} where we let $n,m\rightarrow\infty$ is to see whether we get a reasonably good approximation of power when $n,m$ are large. In a sense we are doing sequential asymptotics, first letting $p\rightarrow\infty$ and deriving the power function, and then deriving the leading term by letting $n,m\rightarrow\infty$. This is a quite common practice in Econometrics (see for example Phillips and Moon\,(1999)). The aim is to derive a leading term for the power when $n,m$ are fixed but large. Consider $\Delta = s/\sqrt{nm}$ (as in Proposition \ref{power}) and set $\sigma^2=\sigma_X^2=\sigma_Y^2=1$. In Figure \ref{fig_power} below, we plot the exact power (computed from (\ref{eqn exact power}) with $50,000$ Monte Carlo samples from the distribution of $M$) with $n=m=5$ and $10$, $t=q_{\alpha, (n-1)(m-1) + v_n + v_m}$ and $\alpha=0.05$, over different values of $s$. We overlay the large sample approximation of the power function (given in Proposition \ref{power}) and observe that the approximation works reasonably well even for small sample sizes.  Clearly larger $s$ results in better power and $s=0$ corresponds to trivial power. 
\begin{figure}[!h]
  \centering
  \subfloat[Power comparison when $m=n=5$]{\includegraphics[width=0.45\textwidth]{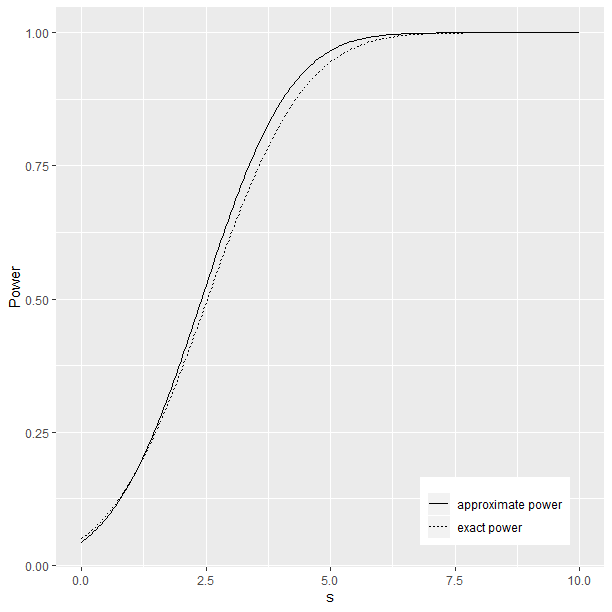}}
  \qquad
  \subfloat[Power comparison when $m=n=10$]{\includegraphics[width=0.45\textwidth]{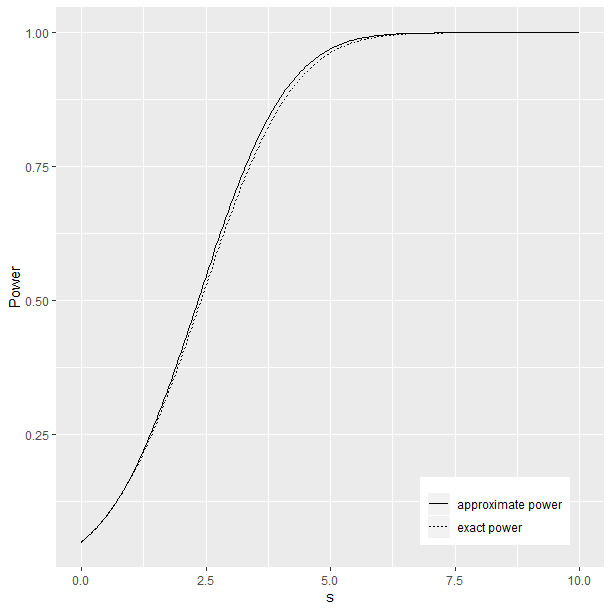}}
  \caption{Comparison of exact and approximate power.}\label{fig_power}
\end{figure}
\end{remark}

We now discuss the power behavior of $T_{n,m}$ based on the Euclidean energy distance. In Case S2, it can be seen that
\begin{align}\label{sigma_x^2}
\sigma_X^2=\lim_{p \to \infty}\frac{1}{\tau^2_X}\dis\sum_{i,i'=1}^{p}4\, \textrm{tr}\,\Sigma_X^2(i,i'),
\end{align}
where $\Sigma_X^2(i,i')$ is the covariance matrix between $X_{(i)}$ and $X_{(i')}$, and similar expressions for $\sigma_Y^2$. In case S2 (i.e., when $K$ is the Euclidean distance), if we further assume $\mu_X = \mu_Y$, it can be verified that
\begin{align}\label{sigma^2}
\sigma^2 \; = \; \dis\lim_{p \to \infty}\,\frac{1}{\tau^2}\dis\sum_{i,i'=1}^{p}4\,\textrm{tr} \big( \Sigma_X(i,i')\, \Sigma_Y(i,i')\big)\,.
\end{align}
Hence in Case S2, under the assumptions that $\mu_X = \mu_Y$, $\textrm{tr}\, \Sigma_X = \textrm{tr}\, \Sigma_Y$ and $\textrm{tr}\, \Sigma_X^2 = \textrm{tr}\, \Sigma_Y^2 = \textrm{tr}\, \Sigma_X \Sigma_Y$, it can be easily seen from equations (\ref{tau for Eucl}), (\ref{sigma_x^2}) and (\ref{sigma^2}) that
\begin{align}\label{Eucl cases}
\tau_X^2 = \tau_Y^2 = \tau^2,\quad \sigma_X^2 = \sigma_Y^2 = \sigma^2,
\end{align}
which implies that $\Delta^*_0=0$ in Proposition \ref{power}. Consider the following class of alternative distributions
$$H_A = \{(P_X,P_Y): P_X\neq P_Y,~\mu_X = \mu_Y,~\textrm{tr}\,\Sigma_X = \textrm{tr}\,\Sigma_Y,\, \textrm{tr}\, \Sigma_X^2 = \textrm{tr}\, \Sigma_Y^2 = \textrm{tr}\, \Sigma_X \Sigma_Y\}.$$
According to Theorem \ref{th KED HDLSS}, the t-test $T_{n,m}$ based on Euclidean energy distance has trivial power against $H_A.$ In contrast, the t-test based on the proposed metrics has non-trivial power against $H_A$ as long as $
\Delta^*_0>0.$\\

To summarize our contributions :
\begin{itemize}
\item We show that the Euclidean energy distance can only detect the equality of means and the traces of covariance matrices in the high-dimensional setup. To the best of our knowledge, such a limitation of the Euclidean energy distance has not been pointed out in the literature before.
\item We propose a new class of homogeneity metrics which completely characterizes homogeneity of two distributions in the low-dimensional setup and has nontrivial power against a broader range of alternatives, or in other words, can detect a wider range of inhomogeneity of two distributions in the high-dimensional setup. 
\item Grouping allows us to detect homogeneity beyond univariate marginal distributions, as the difference between two univariate marginal distributions is automatically captured by the difference between the marginal distributions of the groups that contain these two univariate components. 
\item Consequently we construct a high-dimensional two-sample t-test whose computational cost is linear in $p$. Owing to the pivotal nature of the limiting distribution of the test statistic, no resampling-based inference is needed.
\end{itemize}

\begin{remark}\label{discussion on power}
Although the test based on our proposed statistic is asymptotically powerful against the alternative $H_A$ unlike the Euclidean energy distance, it can be verified that it has trivial power against the alternative $H_{A'} = \{(X,Y) : X_{(i)} \overset{d}{=} Y_{(i)}, 1\leq i \leq p\}$. Thus although it can detect differences between two high-dimensional distributions beyond the first two moments (as a significant improvement to the Euclidean energy distance), it cannot capture differences beyond the equality of the low-dimensional marginal distributions. We conjecture that there might be some intrinsic difficulties for
distance and kernel-based metrics to completely characterize the discrepancy between two high-dimensional distributions.
\end{remark}

\section{Dependence metrics}\label{sec:ACdcov}
In this section, we focus on dependence testing of two random vectors $X \in \bb{R}^{\tilde{p}}$ and $Y \in \bb{R}^{\tilde{q}}$. Suppose $X$ and $Y$ can be partitioned into $p$ and $q$ groups, viz. $X = \left(X_{(1)}, X_{(2)}, \dots, X_{(p)} \right)$ and $Y = \left(Y_{(1)}, Y_{(2)}, \dots, Y_{(q)} \right)$, where the components $X_{(i)}$ and $Y_{(j)}$ are $d_i$ and $g_j$ dimensional, respectively, for $1\leq i\leq p, 1\leq j\leq q$. Here $p, q$ might be fixed or growing.
We assume that $X_{(i)}$ and $Y_{(j)}$'s are finite (low) dimensional vectors, i.e., $\{d_i\}_{i=1}^p$ and $\{g_j\}_{j=1}^q$ are bounded sequences. Clearly, $\tilde{p} = \sum_{i=1}^p d_i = O(p)$ and $\tilde{q} = \sum_{j=1}^q g_j = O(q)$.
We define a class of dependence metrics $\cal{D}$ between $X$ and $Y$ as the positive square root of
\begin{equation}\label{ACdcov def}
\cal{D}^2(X,Y) \;:=\; \E \,K_{\bf{d}}(X,X')\, K_{\bf{g}}(Y,Y') \, + \, \E \,K_{\bf{d}}(X,X') \, \E \,K_{\bf{g}}(Y,Y') \, - \, 2\, \E \,K_{\bf{d}}(X,X') \, K_{\bf{g}}(Y,Y'') \,,
\end{equation}
where ${\bf{d}}=(d_1,\dots,d_p)$ and ${\bf{g}}=(g_1,\dots,g_q)$. We drop the subscripts ${\bf{d}}, {\bf{g}}$ of $K$ for notational convenience.

To ensure the existence of $\cal{D}$, we make the following assumption.
\begin{assumption}\label{ass1}
Assume that $\sup_{1\leq i\leq p}\E \rho_i^{1/2}(X_{(i)},0_{d_i})< \infty$ and $\sup_{1\leq i\leq q}\E \rho_i^{1/2}(Y_{(i)},0_{g_i})< \infty$.
\end{assumption}

In Section \ref{ld D} of the supplement we demonstrate that in the low-dimensional setting, $\cal{D}(X,Y)$ completely characterizes independence between $X$ and $Y$. For an observed random sample $(X_k,Y_k)^{n}_{k=1}$ from the joint distribution of $X$ and $Y$, define $D^X:=(d^X_{kl}) \in \bb{R}^{n\times n}$ with $d^X_{kl} := K(X_k,X_l)$ and $k,l \in \{1, \dots, n\}$. Define $d^Y_{kl}$ and $D^Y$ in a similar way. With some abuse of notation, we consider the U-statistic type estimator $\widetilde{\cal{D}^2_n}(X,Y)$ of $\cal{D}^2$ as defined in (\ref{ustat dcov}) with $d_{\mathcal{X}}$ and $d_{\mathcal{Y}}$ being $K_\mathbf{d}$ and $K_\mathbf{g}$ respectively. In Section \ref{ld D} of the supplement, we illustrate that $\widetilde{\cal{D}^2_n}(X,Y)$ essentially inherits all the nice properties of the U-statistic type estimator of generalized dCov and HSIC.

In the subsequent discussion we study the asymptotic behavior of $\cal{D}$ in the high-dimensional framework, i.e., when $p$ and $q$ grow to $\infty$ with fixed $n$ (discussed in Subsection \ref{sec:ACdcov-HDLSS}) and when $n$ grows to $\infty$ as well (discussed in Subsection \ref{sec:D HDMSS} in the supplement).

\subsection{High dimension low sample size (HDLSS)} \label{sec:ACdcov-HDLSS}
In this subsection, our goal is to explore the behavior of $\cal{D}^2(X,Y)$ and its unbiased U-statistic type estimator in the HDLSS setting where $p$ and $q$ grow to $\infty$ while the sample size $n$ is held fixed. Denote
$\tau^2_{XY} = \tau^2_X \tau^2_Y = \E\, K^2(X,X')\, \E \,K^2(Y,Y').$
We impose the following conditions. 

\begin{assumption}\label{ass D pop taylor}
$\E\, [L^2(X, X')] = O(a_p'^2)$ and $\E\, [L^2(Y, Y')] = O(b_q'^2)$, where $a_p'$ and $b_q'$ are positive real sequences satisfying $a_p' = o(1)$, $b_q' = o(1)$, $\tau_{XY}\, a_p'^2 b_q' = o(1)$ and $\tau_{XY}\, a_p' b_q'^2 = o(1)$. Further assume that $\E\,[R^2(X,X')] = O(a_p'^4)$ and $\E\,[R^2(Y,Y')] = O(b_q'^4)$.
\end{assumption}

\begin{remark}
We refer the reader to Remark \ref{assumption justification} in Section \ref{sec:new-homo} for illustrations about some sufficient conditions under which we have $\var \left(L(X, X')\right)\,=\,\E\,L^2(X, X')\,=\,O(\frac{1}{p})$, and similarly for $L(Y, Y')$. Remark \ref{rem_ui} in the supplement illustrates certain sufficient conditions under which $\E\,[R^2(X,X')] = O(\frac{1}{p^2})$, and similarly for $R(Y,Y')$.
\end{remark}

\begin{theorem}\label{D pop taylor}
Under Assumptions \ref{ass2 : ED} and \ref{ass D pop taylor}, we have
\begin{equation}\label{eq D pop taylor}
\cal{D}^2 (X,Y) = \frac{1}{4\tau_{XY}} \dis \sum_{i=1}^p \sum_{j=1}^q D^2_{\rho_i, \rho_j}(X_{(i)},Y_{(j)}) \, + \, \mathcal{R} \; ,
\end{equation}
where $\mathcal{R}$ is the remainder term such that $\mathcal{R} = O(\tau_{XY}\, a_p'^2 b_q' + \tau_{XY}\, a_p' b_q'^2) = o(1)$.
\end{theorem}

Theorem \ref{D pop taylor} shows that when dimensions grow high, the population $\cal{D}^2 (X,Y)$  behaves as an aggregation of group-wise generalized dCov and thus essentially captures group-wise non-linear dependencies between $X$ and $Y$. 

\begin{remark}\label{rem to th5.3}
Consider a special case where $d_i =1$ and $g_j=1$, and $\rho_i$ and $\rho_j$ are Euclidean distances for all $1\leq i\leq p$ and $1\leq j\leq q$. Then Theorem \ref{D pop taylor} essentially boils down to
\begin{equation}\label{eq D pop taylor sp case 1}
\cal{D}^2 (X,Y) = \frac{1}{4\tau_{XY}} \dis \sum_{i=1}^p \sum_{j=1}^q dCov^2(X_{i},Y_{j}) \, + \, \mathcal{R} \; ,
\end{equation}
where $\mathcal{R} = o(1)$. This shows that in a special case (when we have unit group sizes), $\cal{D}^2 (X,Y)$ essentially  behaves as an aggregation of cross-component dCov between $X$ and $Y$. If $K_{\textbf{d}}$ and $K_{\textbf{g}}$ are Euclidean distances, or in other words if each $\rho_i$ and $\rho_j$ are squared Euclidean distances, then using equation (\ref{alt dcov}) it is straightforward to verify that $D^2_{\rho_i, \rho_j}(X_{i},Y_{j}) = 4\, cov^2(X_i, Y_j)$ for all $1\leq i\leq p$ and $1\leq j\leq q$. Consequently we have 
\begin{equation}\label{eq D pop taylor sp case 2}
\cal{D}^2 (X,Y) = dCov^2(X,Y) = \frac{1}{\tau_{XY}} \dis \sum_{i=1}^p \sum_{j=1}^q cov^2(X_{i},Y_{j}) \, + \, \mathcal{R}_1 \; ,
\end{equation}
where $\mathcal{R}_1 = o(1)$, which essentially presents a population version of Theorem 2.1.1 in Zhu et al.\,(2019) as a special case of Theorem \ref{D pop taylor}.

\end{remark}

\begin{remark}\label{pop approx}
To illustrate that the leading term in equation (\ref{eq D pop taylor}) indeed gives a close approximation of the population $\cal{D}^2 (X,Y)$, we consider the special case when $K_{\textbf{d}}$ and $K_{\textbf{g}}$ are Euclidean distances and $p=q$. Suppose $X \sim N_p(0,I_p)$ and $Y=X+N$ where $N \sim N_p(0,I_p)$ with $N \bigCI X$. Clearly we have $\tau_X^2=2p$, $\tau_Y^2=4p$,  $D^2_{\rho_i, \rho_j}(X_{i},Y_{j}) = 4\, cov^2(X_i, Y_j) = 4$\, for all\, $1\leq i=j\leq p$ and $D^2_{\rho_i, \rho_j}(X_{i},Y_{j}) = 0$ for all \,$1\leq i\neq j\leq p$. From Remark \ref{rem to th5.3}, it is clear that in this case we essentially have $\cal{D}^2 (X,Y) = dCov^2(X,Y)$. We simulate a large sample of size $n=5000$ from the distribution of $(X, Y)$ for $p=20, 40, 60, 80$ and $100$. The large sample size is to ensure that the U-statistic type estimator of $\cal{D}^2 (X,Y)$ (given in (\ref{ustat dcov})) gives a very close approximation of the population $\cal{D}^2 (X,Y)$. We list the ratio between $\cal{D}^2 (X,Y)$ and the leading term in (\ref{eq D pop taylor}) for the different values of $p$, which turn out to be very close to $1$, demonstrating that the leading term in (\ref{eq D pop taylor}) indeed approximates $\cal{D}^2 (X,Y)$ reasonably well.

\begin{table}[!h]\footnotesize
\centering
\caption{Ratio of $\cal{D}^2 (X,Y)$ and the leading term in  (\ref{eq D pop taylor}) for different values of $p$.}
\label{table:pop approx D}
\begin{tabular}{ccccc}
\toprule
$p=20$ & $p=40$ & $p=60$ & $p=80$ & $p=100$  \\
\hline
0.980 & 0.993 & 0.994 & 0.989 & 0.997 \\
\toprule
\end{tabular}
\end{table}

\end{remark}

The following theorem explores the behavior of the population $\cal{D}^2 (X,Y)$ when $p$ is fixed and $q$ grows to infinity, while the sample size is held fixed. As far as we know, this asymptotic regime has not been previously considered in the literature.
In this case, the Euclidean distance covariance behaves as an aggregation of martingale difference divergences proposed in Shao and Zhang (2014) which measures conditional mean dependence. Figure \ref{fig1} below summarizes the curse of dimensionality for the Euclidean distance covariance under different asymptotic regimes.

\begin{theorem}\label{th MDD}
Under Assumption \ref{ass2 : ED} and the assumption that \,$\E\,[R^2(Y,Y')] = O(b_q'^4)$ with $\tau_Y\,b_q'^2 = o(1)$, as $q \to \infty$ with $p$ and $n$ remaining fixed, we have
\begin{align*}
\cal{D}^2 (X,Y) \;&=\; \frac{1}{2\tau_{Y}} \dis \sum_{j=1}^q D^2_{K_{\textbf{d}}\,, \rho_j}(X,Y_{(j)}) \, + \, \mathcal{R},
\end{align*}
where $\mathcal{R}$ is the remainder term such that $\mathcal{R} = O(\tau_{Y}\, b_q'^2) = o(1)$.
\end{theorem}

\begin{remark}\label{rem for th MDD}
In particular, when both $K_{\textbf{d}}$ and $K_{\textbf{g}}$ are Euclidean distances, we have 
\begin{align*}
\cal{D}^2 (X,Y) \;&=\; dCov^2(X,Y)\;=\; \frac{1}{\tau_{Y}} \dis \sum_{j=1}^{\tilde{q}} MDD^2(Y_{j}|X) \, + \, \mathcal{R},
\end{align*}
where $MDD^2(Y_{j}|X)=-\E[(Y_j-\E Y_j)(Y_j'-\E Y_j)\|X-X'\|]$ is the martingale difference divergence which completely characterizes the conditional mean dependence of $Y_j$ given $X$ in the sense that $E[Y_j|X]=E[Y_j]$ almost surely if and only if $MDD^2(Y_{j}|X)=0.$
\end{remark}

\begin{figure}[!h]
\caption{Curse of dimensionality for the Euclidean distance covariance under different asymptotic regimes}\label{fig1}
\centering
\includegraphics[scale=0.22]{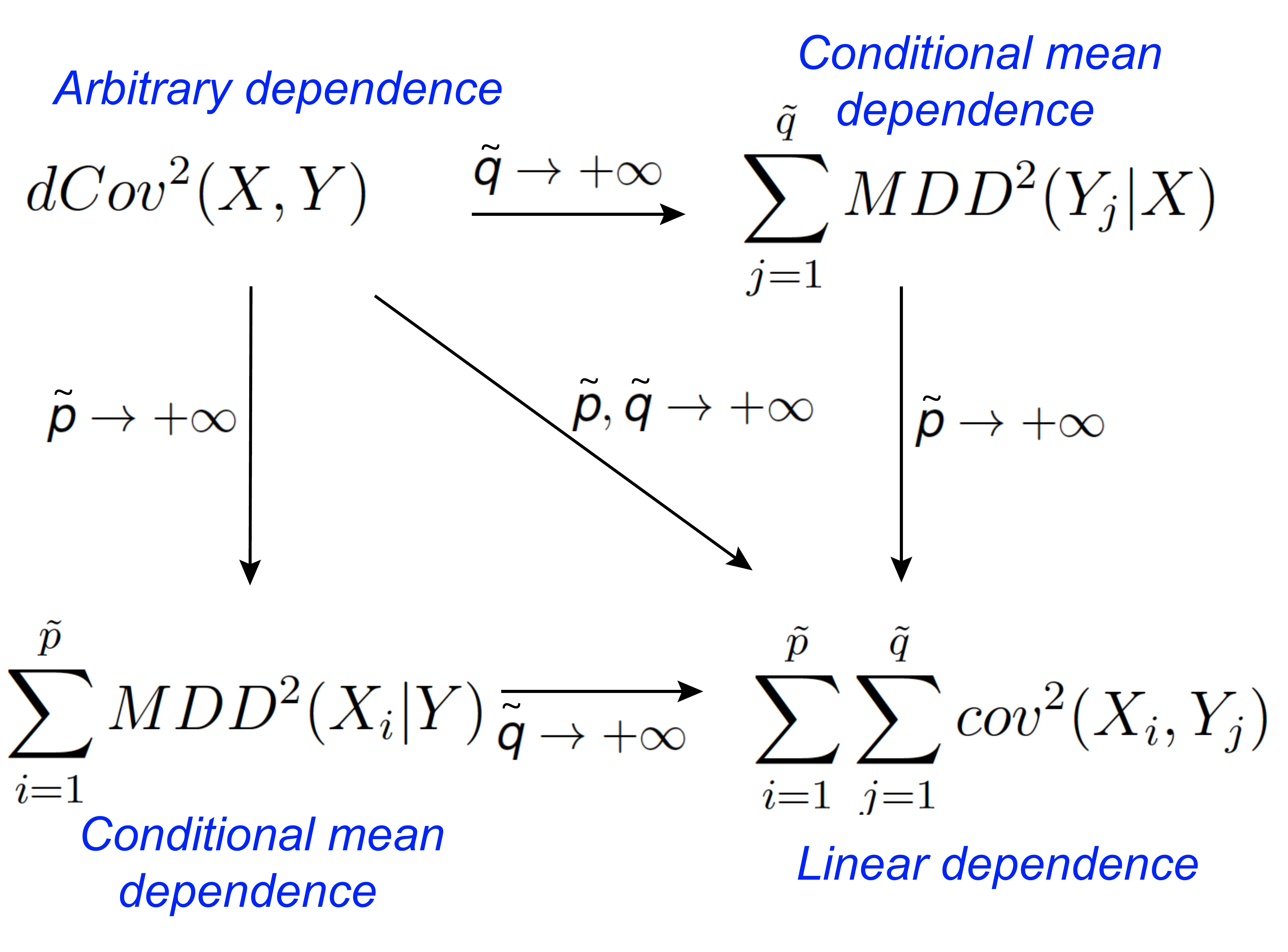}
\end{figure}


Next we study the asymptotic behavior of the sample version $\widetilde{\cal{D}^2_n} (X,Y)$.

\begin{assumption}\label{ass2.1}
Assume that $L(X,X') = O_p(a_p)$ and $L(Y,Y') = O_p(b_q)$, where $a_p$ and $b_q$ are positive real sequences satisfying $a_p = o(1)$, $b_q = o(1)$, $\tau_{XY}\, a_p^2 b_q = o(1)$ and $\tau_{XY}\, a_p b_q^2 = o(1)$.
\end{assumption}

\begin{remark}
We refer the reader to Remark \ref{assumption justification} in Section \ref{sec:new-homo} for illustrations about Assumption \ref{ass2.1}.
\end{remark}

\begin{theorem}\label{ACdCov taylor thm}
Under Assumptions \ref{ass2 : ED} and \ref{ass2.1}, it can be shown that
\begin{equation}\label{ACdCov taylor}
\widetilde{\cal{D}^2_n} (X,Y) = \frac{1}{4\tau_{XY}} \dis \sum_{i=1}^p \sum_{j=1}^q \widetilde{D^2_n}_{\,;\,\rho_i, \rho_j}(X_{(i)},Y_{(j)}) \, + \, \mathcal{R}_n \; ,
\end{equation}
where $X_{(i)}, Y_{(j)}$ are the $i^{th}$ and $j^{th}$ groups of $X$ and $Y$, respectively, $1 \leq i \leq p$, $1 \leq j \leq q$ , and $\mathcal{R}_n$ is the remainder term. Moreover $\mathcal{R}_n = O_p(\tau_{XY}\, a_p^2 b_q + \tau_{XY}\, a_p b_q^2) = o_p(1)$, i.e., $\mathcal{R}_n$ is of smaller order compared to the leading term and hence is asymptotically negligible.
\end{theorem}

The above theorem generalizes Theorem 2.1.1 in Zhu et al.\,(2019) 
by showing that the leading term of $\widetilde{\cal{D}^2_n} (X,Y)$ is the sum of all the group-wise (unbiased) squared sample generalized dCov scaled by $\tau_{XY}\,$. In other words, in the HDLSS setting, $\widetilde{\cal{D}^2_n} (X,Y)$ is asymptotically equivalent to the aggregation of group-wise squared sample generalized dCov. Thus $\widetilde{\cal{D}^2_n} (X,Y)$ can quantify group-wise non-linear dependencies between $X$ and $Y$, going beyond the scope of the usual Euclidean dCov. 

\begin{remark}\label{rem to th5.5}
Consider a special case where $d_i =1$ and $g_j=1$, and $\rho_i$ and $\rho_j$ are Euclidean distances for all $1\leq i\leq p$ and $1\leq j\leq q$. Then Theorem \ref{ACdCov taylor thm} essentially states that
\begin{equation}\label{ACdCov taylor sp case 1}
\widetilde{\cal{D}^2_n} (X,Y) = \frac{1}{4\tau_{XY}} \dis \sum_{i=1}^p \sum_{j=1}^q dCov^2_n(X_{i},Y_{j}) \, + \, \mathcal{R}_n \; ,
\end{equation}
where $\mathcal{R}_n = o_p(1)$. This demonstrates that in a special case (when we have unit group sizes), $\widetilde{\cal{D}^2_n} (X,Y)$ is asymptotically equivalent to the marginal aggregation of cross-component distance covariances proposed by Zhu et al.\,(2019) as dimensions grow high. If $K_{\textbf{d}}$ and $K_{\textbf{g}}$ are Euclidean distances, then Theorem \ref{ACdCov taylor thm} essentially boils down to Theorem 2.1.1 in Zhu et al.\,(2019) as a special case.
\end{remark}

\begin{remark}\label{approximation}
To illustrate the approximation of $\widetilde{\cal{D}^2_n} (X,Y)$ by the aggregation of group-wise squared sample generalized dCov given by Theorem \ref{ACdCov taylor thm}, we simulated the datasets in Examples \ref{eg1}.1, \ref{eg1}.2, \ref{eg2}.1 and \ref{eg2}.2\; $100$ times each with $n=50$ and $p=q=50$. For each of the datasets, the difference between $\widetilde{\cal{D}^2_n} (X,Y)$ and the leading term in the RHS of equation (\ref{ACdCov taylor}) is smaller than $0.01$ $100\%$ of the times, which illustrates that the approximation works reasonably well.
\end{remark}

The following theorem illustrates the asymptotic behavior of $\widetilde{\cal{D}^2_n} (X,Y)$ when $p$ is fixed and $q$ grows to infinity while the sample size is held fixed. Under this setup, if both $K_{\textbf{d}}$ and $K_{\textbf{g}}$ are Euclidean distances, the leading term of $\widetilde{\cal{D}^2_n} (X,Y)$ is the sum of the group-wise unbiased U-statistic type estimators of $MDD^2(Y_j | X)$ for $1\leq j \leq q$, scaled by $\tau_{Y}\,$. In other words, the sample Euclidean distance covariance behaves as an aggregation of sample martingale difference divergences.

\begin{theorem}\label{th MDD sample}
Under Assumption \ref{ass2 : ED} and the assumption that \,$L(Y,Y') = O_p(b_q)$ with $b_q=o(1)$ and\, $\tau_Y\,b_q^2 = o(1)$, as\, $q \to \infty$ with $p$ and $n$ remaining fixed, we have
\begin{align*}
\widetilde{\cal{D}^2_n}(X,Y) \;&=\; \frac{1}{2\tau_{Y}} \dis \sum_{j=1}^q \widetilde{\cal{D}^2_n}_{\,;\,K_{\textbf{d}} \,, \rho_j}(X,Y_{(j)}) \, + \, \mathcal{R}_n,
\end{align*}
where $\mathcal{R}_n$ is the remainder term such that $\mathcal{R}_n = O_p(\tau_{Y}\, b_q^2) = o_p(1)$. 
\end{theorem}

\begin{remark}\label{rem for th MDD sample}
In particular, when both $K_{\textbf{d}}$ and $K_{\textbf{g}}$ are Euclidean distances, we have 
\begin{align*}
\widetilde{\cal{D}^2_n} (X,Y) \;&=\; dCov^2_n(X,Y) \;=\; \frac{1}{\tau_{Y}} \dis \sum_{j=1}^{\tilde{q}} MDD^2_n(Y_{j}|X) \, + \, \mathcal{R}_n,
\end{align*}
where $MDD^2_n(Y_{j}|X)$ is the unbiased U-statistic type estimator of $MDD^2(Y_{j}|X)$ defined as in (\ref{ustat dcov}) with $d_{\cal{X}}(x,x')=\|x-x'\|$ for $x,x'\in\mathbb{R}^{\tilde{p}}$ and $d_{\cal{Y}}(y,y')=|y-y'|^2/2$ for $y,y'\in\mathbb{R}$, respectively.
\end{remark}

Now denote $X_k = (X_{k(1)}, \dots, X_{k(p)})$ and $Y_k=(Y_{k(1)},\dots,Y_{k(q)})$ for $1\leq k\leq n$. Define the leading term of $\widetilde{\cal{D}^2_n} (X,Y)$ in equation (\ref{ACdCov taylor}) as $$L := \,\frac{1}{4\tau_{XY}\,} \sum_{i=1}^p \sum_{j=1}^q \widetilde{D^2_n}_{\,;\,\rho_i, \rho_j}(X_{(i)},Y_{(j)})\,.$$
It can be verified that $$L\, = \,\frac{1}{4\tau_{XY}\,} \sum_{i=1}^p \sum_{j=1}^q \left(\tilde{D}^X(i) \cdot \tilde{D}^Y(j)\right) \,,$$ where $\tilde{D}^X(i), \tilde{D}^Y(j)$ are the $\cal{U}$-centered versions of $D^X(i) = \left(d^X_{kl}(i)\right)_{k,l=1}^n$ and $D^Y(j) = \left(d^Y_{kl}(j)\right)_{k,l=1}^n$, respectively.
As an advantage of using the double-centered distances, we have for all $1\leq i,i' \leq p$, $1\leq j, j' \leq q$ and $\{k,l\} \neq \{u, v\},$
\begin{align}\label{double}
\E \left[d^X_{kl}(i)\, d^X_{uv}(i')\right] \; = \; \E\left[d^Y_{kl}(j)\,d^Y_{uv}(j')\right] \; = \;\E \left[d^X_{kl}(i)\, d^Y_{uv}(j)\right]\;= \;0.
\end{align}
See for example the proof of Proposition 2.2.1 in Zhu et al.\,(2019) for a detailed explanation.

\begin{assumption}\label{ass3}
For fixed $n$, as $p, q \to \infty$, $$\begin{pmatrix}
\frac{1}{2\,\tau_X} \dis \sum_{i=1}^p d^X_{kl}(i) \\
\frac{1}{2\,\tau_Y} \dis \sum_{j=1}^q d^Y_{uv}(j)
\end{pmatrix}_{k<l,\, u<v} \overset{d}{\longrightarrow} \;\;\begin{pmatrix}
d_{kl}^1 \\ \\ d_{uv}^2
\end{pmatrix}_{k<l,\, u<v} \, ,$$ where $\{d_{kl}^1,\, d_{uv}^2\}_{k<l,\, u<v}$ are jointly  Gaussian. Further we assume that
\begin{align*}
&var(d_{kl}^1):= \sigma_X^2=  \lim_{p \to \infty}\, \frac{1}{4\tau_X^2} \dis \sum_{i,i'=1}^p D^2_{\rho_i, \rho_{i'}} \left(X_{(i)}, X_{(i')} \right),\\
&var(d_{kl}^2):= \sigma_Y^2 =  \lim_{q \to \infty}\, \frac{1}{4\tau_Y^2} \dis \sum_{j,j'=1}^q D^2_{\rho_j, \rho_{j'}} \left(Y_{(j)}, Y_{(j')} \right),\\
&\cov\,(d_{kl}^1, d_{kl}^2) := \sigma_{XY}^2  =  \lim_{p, q \to \infty}\, \frac{1}{4\tau_{XY}} \dis \sum_{i=1}^p \sum_{j=1}^q D^2_{\rho_i, \rho_j} \left(X_{(i)}, Y_{(j)} \right).
\end{align*}
\end{assumption}
In view of (\ref{double}), we have $\cov\,(d_{kl}^1, d_{uv}^1) = \cov\,(d_{kl}^2, d_{uv}^2) = \cov\,(d_{kl}^1, d_{uv}^2) = 0$ for $\{k,l\} \neq \{u, v\}$. Theorem \ref{ACdCov taylor thm} states that for growing $p$ and $q$ and fixed $n$, $\widetilde{\cal{D}^2_n} (X,Y)$ and $L$ are asymptotically equivalent. By studying the leading term, we obtain the limiting distribution of $\widetilde{\cal{D}^2_n} (X,Y)$ as follows. 

\begin{theorem}\label{ACdcov:dist_conv}
Under Assumptions \ref{ass2 : ED}, \ref{ass2.1} and \ref{ass3}, for fixed $n$ and $p, q \to \infty$,
\begin{align*}
&\widetilde{\cal{D}^2_n} (X,Y)\; \overset{d}{\longrightarrow} \; \frac{1}{\nu}  d^{1\top} M d^2 \,,\\
&\widetilde{\cal{D}^2_n} (X,X)\; \overset{d}{\longrightarrow} \; \frac{1}{\nu}  d^{1\top} M d^1\; \overset{d}{=} \; \frac{\sigma_X^2}{\nu} \chi^2_{\nu}\,,\\
&\widetilde{\cal{D}^2_n} (Y,Y)\; \overset{d}{\longrightarrow} \; \frac{1}{\nu} d^{2\top} M d^2\; \overset{d}{=} \; \frac{\sigma_Y^2}{\nu} \chi^2_{\nu}\,,
\end{align*}
where $M$ is a projection matrix of rank \,$\nu=\frac{n(n-3)}{2}$, and\, $$ \begin{pmatrix}
d^1 \\ d^2
\end{pmatrix} \; \sim \; N \left( 0\,, \begin{pmatrix}
\sigma_X^2 \,I_{\frac{n(n-1)}{2}} \;\; \sigma_{XY}^2 \,I_{\frac{n(n-1)}{2}} \\ \\
\sigma_{XY}^2 \,I_{\frac{n(n-1)}{2}} \;\; \sigma_Y^2 \,I_{\frac{n(n-1)}{2}}
\end{pmatrix} \right) \,.$$
\end{theorem}

To perform independence testing, in the spirit of Sz\'{e}kely and Rizzo\,(2014), we define the studentized test statistic
\begin{equation}\label{student t}
\cal{T}_n \; := \; \sqrt{\nu-1}\, \frac{\widetilde{\cal{DC}^2_n} (X,Y)}{\sqrt{1 - \left(\widetilde{\cal{DC}^2_n} (X,Y)\right)^2}}\; ,
\end{equation}
where $$ \widetilde{\cal{DC}^2_n} (X,Y) \; = \; \frac{\widetilde{\cal{D}^2_n} (X,Y)}{\sqrt{\widetilde{\cal{D}^2_n} (X,X) \, \widetilde{\cal{D}^2_n} (Y,Y)}}\,.$$

Define $\psi = \sigma_{XY}^2 / \sqrt{\sigma_X^2 \sigma_Y^2}$.  The following theorem states the asymptotic distributions of the test statistic $\cal{T}_n$ under the null hypothesis $\tilde{H}_0: X \bigCI Y$ and the alternative hypothesis $\tilde{H}_A: X \nbigCI Y$.

\begin{theorem}\label{HDLSS dist conv}
Under Assumptions \ref{ass2 : ED}, \ref{ass2.1} and \ref{ass3}, for fixed $n$ and $p, q \to \infty$,
\begin{align*}
& P_{\tilde{H}_0} \left(\cal{T}_n \leq t \right) \; \longrightarrow \; P\left(t_{\nu -1} \leq t \right),\\
& P_{\tilde{H}_A} \left(\cal{T}_n \leq t \right) \; \longrightarrow \; \E \left[P\left(t_{\nu -1, W} \leq t |W\right) \right],
\end{align*}
where $t$ is any fixed real number and $W \sim \sqrt{\frac{\psi^2}{1-\psi^2}\, \chi^2_{\nu}}$.
\end{theorem}
For an explicit form of $\E \left[P\left(t_{\nu -1, W} \leq t |W \right)\right]$, we refer the reader to Lemma 3 in the appendix of Zhu et al.\,(2019).
Now consider the local alternative hypothesis $\tilde{H}_{A}^*$: $X \nbigCI Y$ with $\psi=\psi_0/\sqrt{\nu}$, where $\psi_0$ is a constant with respect to $n$. The following proposition gives an approximation of $\E \left[P\left(t_{\nu -1, W} \leq t |W \right)\right]$ under the local alternative hypothesis $\tilde{H}_{A}^*$ when $n$ is allowed to grow.
\begin{proposition}\label{local alt power}
Under $\tilde{H}_{A}^*$, as $n \to \infty$ and \,$t = O(1)$, $$\E \left[P\left(t_{\nu -1, W} \leq t |W\right)\right] \; = \; P\left(t_{\nu -1, \,\psi_0} \leq t \right) \; + \; O\Big(\frac{1}{\nu}\Big)\,.$$
\end{proposition}

The following summarizes our key findings in this section.
\begin{itemize}
\item \textbf{Advantages of our proposed metrics over the Euclidean dCov and HSIC :} 
\begin{enumerate}[i)]
\item Our proposed dependence metrics completely characterize independence between $X$ and $Y$ in the low-dimensional setup, and can detect group-wise non-linear dependencies between $X$ and $Y$ in the high-dimensional setup as opposed to merely detecting component-wise linear dependencies by the Euclidean dCov and HSIC (in light of Theorem 2.1.1 in Zhu et al.\,(2019)). 

\item We also showed that with $p$ remaining fixed and $q$ growing high, the Euclidean dCov can only quantify conditional mean independence of the components of $Y$ given $X$ (which is weaker than independence). To the best of our knowledge, this has not been pointed out in the literature before.
\end{enumerate}

\item \textbf{Advantages over the marginal aggregation approach by Zhu et al.\,(2019) :} 
\begin{enumerate}[i)]
\item In the low-dimensional setup, our proposed dependence metrics can completely characterize independence between $X$ and $Y$, whereas the metric proposed by Zhu et al.\,(2019) can only capture pairwise dependencies between the components of $X$ and $Y$.  
\item We provide a neater way of generalizing dCov and HSIC between $X$ and $Y$ which is shown to be asymptotically equivalent to the marginal aggregation of cross-component distance covariances proposed by Zhu et al.\,(2019) as dimensions grow high. Also grouping or partitioning the two high-dimensional random vectors (which again may be problem specific) allows us to detect a wider range of alternatives compared to only detecting component-wise non-linear dependencies, as independence of two univariate marginals is implied from independence of two higher dimensional marginals containing the two univariate marginals.
\item 
The computational complexity of the (unbiased) squared sample $\cal{D}(X,Y)$ is $O(n^2(p+q))$. Thus the computational cost of our proposed two-sample t-test only grows linearly with the dimension 
and therefore is scalable to very high-dimensional data. Although a naive aggregation of marginal distance covariances has a computational complexity of $O(n^2 pq)$, the approach of Zhu et al.\,(2019) essentially corresponds to the use of an additive kernel and the computational cost of their proposed estimator can also be made linear in the dimensions if properly implemented.
\end{enumerate}
\end{itemize}

\begin{table}[!h]\small
\caption{Summary of the behaviors of the proposed homogeneity/dependence metrics for different choices of $\rho_i(x,x')$ in high dimension.}
    \centering
\begin{tabularx}{\linewidth}{|L|L|L|}
    \hline
Choice of $\rho_i(x,x')$ & Asymptotic behavior of the proposed homogeneity metric & Asymptotic behavior of the proposed dependence metric \\
    \hline
the semi-metric $\Vert x-x' \Vert^2$ & Behaves as a sum of squared Euclidean distances & Behaves as a sum of squared Pearson correlations \\
    \hline
metric of strong negative type on $\mathbb{R}^{d_i}$ &  Behaves as a sum of groupwise energy distances with the metric $\rho_i$ & Behaves as a sum of groupwise dCov with the metric $\rho_i$\\
    \hline
$k_i(x,x) + k_i(x',x') - 2k_i(x,x')$, where $k_i$ is a characteristic kernel on $\mathbb{R}^{d_i}\times \mathbb{R}^{d_i}$ & Behaves as a sum of groupwise MMD with the kernel $k_i$ & Behaves as a sum of groupwise HSIC with the kernel $k_i$\\
\hline
\end{tabularx}
\end{table}


\section{Numerical studies} \label{sec:num}
\subsection{Testing for homogeneity of distributions}\label{subsec homo num}
We investigate the empirical size and power of the tests for homogeneity of two high dimensional distributions. For comparison, we consider the t-tests based on the following metrics:
\begin{enumerate}[I.]
    \item $\cal{E}$ with $\rho_i$ as the Euclidean distance for $1\leq i\leq p$;
    \item $\cal{E}$ with $\rho_i$ as the distance induced by the Laplace kernel for $1\leq i\leq p$;
    \item $\cal{E}$ with $\rho_i$ as the distance induced by the Gaussian kernel for $1\leq i\leq p$;
    \item the usual Euclidean energy distance;
    \item MMD with the Laplace kernel;
    \item MMD with the Gaussian kernel.
\end{enumerate}

We set $d_i =1$ in Examples \ref{eg1:ed} and \ref{eg2:ed}, and $d_i =2$ in Example \ref{eg3:ed} for $1\leq i\leq p$.

\begin{example}\label{eg1:ed}
Consider $X_k = (X_{k1}, \dots, X_{kp})$\, and \,$Y_l = (Y_{l1}, \dots, Y_{lp})$ with $k=1,\dots, n$ and $l=1,\dots, m$. We generate i.i.d. samples from the following models:
\begin{enumerate}
\item $X_k \sim N(0, I_p)$ \,and \,$Y_l \sim N(0, I_p)$.
\item $X_k \sim N(0, \Sigma)$ \,and\, $Y_l \sim N(0, \Sigma)$, where $\Sigma = (\sigma_{ij})_{i,j=1}^p$ with $\sigma_{ii}=1$ for\, $i=1, \dots, p$, $\sigma_{ij} = 0.25$ if $1 \leq |i-j| \leq 2$ and $\sigma_{ij} = 0$ otherwise.
\item $X_k \sim N(0, \Sigma)$ \,and\, $Y_l \sim N(0, \Sigma)$, where $\Sigma = (\sigma_{ij})_{i,j=1}^p$ with $\sigma_{ij} = 0.7^{|i-j|}$.
\end{enumerate}
\end{example}

\begin{example}\label{eg2:ed}
Consider $X_k = (X_{k1}, \dots, X_{kp})$\, and \,$Y_l = (Y_{l1}, \dots, Y_{lp})$ with $k=1,\dots, n$ and $l=1,\dots, m$. We generate i.i.d. samples from the following models:
\begin{enumerate}
\item $X_k \sim N(\mu, I_p)$ with $\mu = (1, \dots, 1)\in\mathbb{R}^p$\, and \,$Y_{li} \overset{ind}{\sim}$ Poisson\,$(1)$ for $i=1,\dots,p$.
\item $X_k \sim N(\mu, I_p)$ with $\mu = (1, \dots, 1)\in \mathbb{R}^p$\, and\, $Y_{li} \overset{ind}{\sim}$ Exponential\,$(1)$ for $i=1,\dots,p$.
\item $X_k \sim N(0, I_p)$\, and\, $Y_l = (Y_{l1}, \dots, Y_{l\lfloor\beta p\rfloor}, Y_{l(\lfloor\beta p\rfloor +1)}, \dots ,  Y_{lp})$, where $Y_{l1}, \dots, Y_{l \lfloor\beta p\rfloor} \overset{i.i.d.}{\sim}$ Rademacher\,$(0.5)$\, and \,$Y_{l(\lfloor\beta p\rfloor +1)}, \dots ,  Y_{lp} \overset{i.i.d.}{\sim} N(0,1)$.
\item $X_k \sim N(0, I_p)$\, and\, $Y_l = (Y_{l1}, \dots, Y_{l \lfloor\beta p\rfloor}, Y_{l(\lfloor\beta p\rfloor +1)}, \dots ,  Y_{lp})$, where $Y_{l1}, \dots, Y_{l\lfloor \beta p\rfloor} \overset{i.i.d.}{\sim}$ Uniform\,$(-\sqrt{3}, \sqrt{3})$\, and \,$Y_{l(\lfloor\beta p\rfloor +1)}, \dots ,  Y_{lp} \overset{i.i.d.}{\sim} N(0,1)$.
\item $X_k = R^{1/2} Z_{1k}$ and $Y_l = R^{1/2} Z_{2l}$, where $R = (r_{ij})_{i,j=1}^p$ with $r_{ii}=1$ for\, $i=1, \dots, p$, $r_{ij} = 0.25$ if $1 \leq |i-j| \leq 2$ and $r_{ij} = 0$ otherwise, $Z_{1k} \sim N(0, I_p)$ and  $Z_{2l} = \underbrace{(Z_{2l1}, \dots, Z_{2lp})}_{ \overset{i.i.d.}{\sim} Exponential(1)} -\, 1.
$
\end{enumerate}
\end{example}

\begin{example}\label{eg3:ed}
Consider $X_k = (X_{k(1)}, \dots, X_{k(p)})$\, and \,$Y_l = (Y_{l(1)}, \dots, Y_{l(p)})$ with $k=1,\dots, n$ and $l=1,\dots, m$\, and\, $d_i =2$ for $1\leq i \leq p$. We generate i.i.d. samples from the following models:
\begin{enumerate}
\item $X_{k(i)} \sim N(\mu, \Sigma_1)$\, and\, $Y_{l(i)} \sim N(\mu, \Sigma_2)$ with $X_{k(i)} \bigCI X_{k(j)}$ and $Y_{l(i)} \bigCI Y_{l(j)}$ for $1\leq i \neq j \leq p$, where $\mu=(1,1)^\top $, $\Sigma_1 = \begin{pmatrix}
1 & 0.9\\
0.9 & 1
\end{pmatrix}$\,and\, $\Sigma_2 = \begin{pmatrix}
1 & 0.1\\
0.1 & 1
\end{pmatrix}$.
\item $X_{k(i)} \sim N(\mu, \Sigma)$ with $X_{k(i)} \bigCI X_{k(j)}$ for $1\leq i \neq j \leq p$, where $\mu=(1,1)^\top $, $\Sigma = \begin{pmatrix}
1 & 0.7\\
0.7 & 1
\end{pmatrix}$. The components of \,$Y_l$ are i.i.d. Exponential\,$(1)$.
\end{enumerate}
\end{example}

Note that for Examples \ref{eg1:ed} and \ref{eg2:ed}, the metric defined in equation (\ref{Kdef}) essentially boils down to the special case in equation (\ref{Kdef_sp}). We try small sample sizes $n=m=50$, dimensions $p=q=50, 100$ and $200$, and $\beta = 1/2$. Table \ref{table1:ed} reports the proportion of rejections out of $1000$ simulation runs for the different tests. For the tests V and VI, we chose the bandwidth parameter heuristically as the median distance between the aggregated sample observations. For tests II and III, the bandwidth parameters are chosen using the median heuristic separately for each group.

In Example \ref{eg1:ed}, the data generating scheme suggests that the variables $X$ and $Y$ are identically distributed. The results in Table \ref{table1:ed} show that the tests based on both the proposed homogeneity metrics and the usual Euclidean energy distance and MMD perform more or less equally good, and the rejection probabilities are quite close to the $10\%$ or $5\%$ nominal level. In Example \ref{eg2:ed}, clearly $X$ and $Y$ have different distributions but $\mu_X = \mu_Y$ and $\Sigma_X = \Sigma_Y$. The results in Table \ref{table1:ed} indicate that the tests based on the proposed homogeneity metrics are able to detect the differences between the two high-dimensional distributions beyond the first two moments unlike the tests based on the usual Euclidean energy distance and MMD, and thereby outperform the latter in terms of empirical power. In Example \ref{eg3:ed}, clearly $\mu_X = \mu_Y$ and $\textrm{tr}\, \Sigma_X = \textrm{tr}\, \Sigma_Y$ and the results show that the tests based on the proposed homogeneity metrics are able to detect the in-homogeneity of the low-dimensional marginal distributions unlike the tests based on the usual Euclidean energy distance and MMD.

\begin{table}[H]\scriptsize
\centering
\caption{Empirical size and power for the different tests of homogeneity of distributions.}
\label{table1:ed}
\begin{tabular}{c cccc cc cc cc cc cc cc}
\toprule
&&&&&\multicolumn{2}{c}{I}&\multicolumn{2}{c}{II}&\multicolumn{2}{c}{III}&\multicolumn{2}{c}{IV}&\multicolumn{2}{c}{V}&\multicolumn{2}{c}{VI}
\\ \cmidrule(r){6-7}\cmidrule(r){8-9}\cmidrule(r){10-11}\cmidrule(r){12-13}\cmidrule(r){14-15}\cmidrule(r){16-17}
& & $n$ & $m$ & $p$ &
10\% & 5\% & 10\% & 5\% & 10\% & 5\% & 10\% & 5\% & 10\% & 5\% & 10\% & 5\%  \\
\hline
\multirow{9}{*}{Ex \ref{eg1:ed}} & (1) & 50 & 50 & 50 & 0.109 & 0.062 & 0.109 & 0.058 & 0.106 & 0.063 & 0.109 & 0.068 & 0.110 & 0.069 & 0.109 & 0.070   \\
& (1) & 50 & 50 & 100 &  0.124 & 0.073 & 0.119 & 0.053 & 0.121 & 0.063 & 0.116 & 0.067 & 0.114 & 0.068 & 0.117 & 0.068   \\
    & (1) & 50 & 50 & 200 &  0.086 & 0.043 & 0.099 & 0.048 & 0.088 & 0.035 & 0.090 & 0.045 & 0.086 & 0.043 & 0.090 & 0.045  \\
    & (2) & 50 & 50 & 50 & 0.114 & 0.069 & 0.108 & 0.054 & 0.118 & 0.068 & 0.116 & 0.077 & 0.115 & 0.073 & 0.116 & 0.078    \\
    & (2) & 50 & 50 & 100 &  0.130 & 0.069 & 0.133 & 0.073 & 0.124 & 0.070 & 0.126 & 0.067 & 0.123 & 0.068 & 0.124 & 0.067  \\
     & (2) & 50 & 50 & 200 & 0.099 & 0.048 & 0.103 & 0.041 & 0.092 & 0.047 & 0.097 & 0.040 & 0.095 & 0.039 & 0.097 & 0.040   \\
     & (3) & 50 & 50 & 50 &  0.100 & 0.064 & 0.107 & 0.057 & 0.099 & 0.060 & 0.112 & 0.072 & 0.105 & 0.067 & 0.110 & 0.073 \\
     & (3) & 50 & 50 & 100 &  0.103 & 0.062 & 0.113 & 0.061 & 0.113 & 0.063 & 0.097 & 0.060 & 0.100 & 0.057 & 0.098 & 0.059  \\
     & (3) & 50 & 50 & 200 & 0.108 & 0.062 & 0.115 & 0.062 & 0.117 & 0.064 & 0.091 & 0.055 & 0.093 & 0.056 & 0.090 & 0.055  \\
\hline
\multirow{12}{*}{Ex \ref{eg2:ed}} 
    & (1) & 50 & 50 & 50 & 1& 1& 1& 1& 0.995 & 0.994 & 0.102 & 0.067 & 0.111 & 0.069 & 0.103 & 0.066 \\
    & (1) & 50 & 50 & 100 & 1& 1& 1& 1& 1& 1& 0.120 & 0.066 & 0.120 & 0.071 & 0.119 & 0.066  \\
     & (1) & 50 & 50 & 200 & 1& 1& 1& 1& 1& 1& 0.111 & 0.057 & 0.111 & 0.057 & 0.111 & 0.057   \\
     & (2) & 50 & 50 & 50 &  1& 1& 1& 1& 1& 1& 0.126 & 0.085 & 0.154 & 0.105 & 0.119 & 0.073   \\
     & (2) & 50 & 50 & 100 & 1& 1& 1& 1& 1& 1& 0.098 & 0.058 & 0.108 & 0.066 & 0.094 & 0.055  \\
     & (2) & 50 & 50 & 200 & 1& 1& 1& 1& 1& 1& 0.111 & 0.055 & 0.114 & 0.056 & 0.108 & 0.054 \\
     & (3) & 50 & 50 & 50 &  1& 1& 1& 1& 1& 0.999& 0.118 & 0.069 & 0.117 & 0.072 & 0.120 & 0.070    \\
     & (3) & 50 & 50 & 100 & 1& 1& 1& 1& 1& 1& 0.102 & 0.067 & 0.106 & 0.065 & 0.103 & 0.067  \\
     & (3) & 50 & 50 & 200 & 1& 1& 1& 1& 1& 1& 0.103 & 0.046 & 0.103 & 0.049 & 0.102 & 0.046 \\
     & (4) & 50 & 50 & 50 &  0.452 & 0.328 & 0.863 & 0.771 & 0.552 & 0.421 & 0.114 & 0.061 & 0.111 & 0.061 & 0.114 & 0.061    \\
     & (4) & 50 & 50 & 100 & 0.640 & 0.491 & 0.990 & 0.967 & 0.761 & 0.637 & 0.098 & 0.063 & 0.104 & 0.063 & 0.098 & 0.062  \\
     & (4) & 50 & 50 & 200 & 0.840 & 0.733 & 1& 0.999 & 0.933 & 0.876 & 0.105 & 0.042 & 0.108 & 0.042 & 0.105 & 0.043 \\
     & (5) & 50 & 50 & 50 &  1& 1& 1& 1& 1& 1& 0.128 & 0.078 & 0.163 & 0.098 & 0.115 & 0.077   \\
     & (5) & 50 & 50 & 100 & 1& 1& 1& 1& 1& 1& 0.098 & 0.053 & 0.115 & 0.063 & 0.091 & 0.051  \\
     & (5) & 50 & 50 & 200 & 1& 1& 1& 1& 1& 1& 0.100 & 0.050 & 0.103 & 0.054 & 0.098 & 0.050  \\

\hline
\multirow{6}{*}{Ex \ref{eg3:ed}} & (1) & 50 & 50 & 50 & 1& 1& 1& 1& 1 & 1 & 0.157 & 0.098 & 0.223 & 0.137 & 0.156 & 0.098 \\
    & (1) & 50 & 50 & 100 & 1& 1& 1& 1& 1& 1& 0.158 & 0.089 & 0.188 & 0.124 & 0.157 & 0.090 \\
     & (1) & 50 & 50 & 200 & 1& 1& 1& 1& 1& 1& 0.122 & 0.074 & 0.161 & 0.091 & 0.121 & 0.074   \\
     & (2) & 50 & 50 & 50 &  1& 1& 1& 1& 1& 1& 0.140 & 0.078 & 0.190 & 0.118 & 0.137 & 0.075   \\
     & (2) & 50 & 50 & 100 & 1& 1& 1& 1& 1& 1& 0.139 & 0.080 & 0.171 & 0.105 & 0.136 & 0.080   \\
     & (2) & 50 & 50 & 200 & 1& 1& 1& 1& 1& 1& 0.109 & 0.053 & 0.127 & 0.069 & 0.108 & 0.053 \\

\toprule
\end{tabular}
\\
\end{table}

\begin{remark}\label{rem on choice of d_i}
In Example \ref{eg3:ed}.1, marginally the $p$-many two-dimensional groups of $X$ and $Y$ are not identically distributed, but each of the $2p$ unidimensional components of $X$ and $Y$ have identical distributions. Consequently, choosing $d_i = 1$ for $1\leq i \leq p$ leads to trivial power of even our proposed tests, as is evident from Table \ref{table1:ed rem} below. This demonstrates that grouping allows us to detect a wider range of alternatives.

\begin{table}[H]\scriptsize
\centering
\caption{Empirical power in Example \ref{eg3:ed}.1 if we choose $d_i = 1$ for $1\leq i \leq p$.}
\label{table1:ed rem}
\begin{tabular}{c cccc cc cc cc cc cc cc}
\toprule
&&&&&\multicolumn{2}{c}{I}&\multicolumn{2}{c}{II}&\multicolumn{2}{c}{III}&\multicolumn{2}{c}{IV}&\multicolumn{2}{c}{V}&\multicolumn{2}{c}{VI}
\\ \cmidrule(r){6-7}\cmidrule(r){8-9}\cmidrule(r){10-11}\cmidrule(r){12-13}\cmidrule(r){14-15}\cmidrule(r){16-17}
& & $n$ & $m$ & $p$ &
10\% & 5\% & 10\% & 5\% & 10\% & 5\% & 10\% & 5\% & 10\% & 5\% & 10\% & 5\%  \\
\hline
\multirow{3}{*}{Ex \ref{eg3:ed}} & (1) & 50 & 50 & 50 & 0.144 & 0.087 & 0.133 & 0.076 & 0.143 & 0.086 & 0.174 & 0.107 & 0.266 & 0.170 & 0.175 & 0.105 \\
    & (1) & 50 & 50 & 100 & 0.145 & 0.085 & 0.134 & 0.070 & 0.142 & 0.085 & 0.157 & 0.098 & 0.223 & 0.137 & 0.156 & 0.098\\
     & (1) & 50 & 50 & 200 & 0.126 & 0.063 & 0.101 & 0.058 & 0.111 & 0.065 & 0.158 & 0.089 & 0.188 & 0.124 & 0.157 & 0.090   \\

\toprule
\end{tabular}
\\
\end{table} 

\end{remark}

\subsection{Testing for independence}
 We study the empirical size and power of tests for independence between two high dimensional random vectors.
 We consider the t-tests based on the following metrics:
\begin{enumerate}[I.]
    \item $\cal{D}$ with $d_i=1$ and $\rho_i$ be the Euclidean distance for $1\leq i\leq p$;
    \item $\cal{D}$ with $d_i=1$ and $\rho_i$ be the distance induced by the Laplace kernel for $1\leq i\leq p$;
    \item $\cal{D}$ with $d_i=1$ and $\rho_i$ be the distance induced by the Gaussian kernel for $1\leq i\leq p$;
    \item the usual Euclidean distance covariance;
    \item HSIC with the Laplace kernel;
    \item HSIC with the Gaussian kernel.
\end{enumerate}
 The numerical examples we consider are motivated from Zhu et al. (2019).

\begin{example}\label{eg1}
Consider $X_k = (X_{k1}, \dots, X_{kp})$\, and \,$Y_k = (Y_{k1}, \dots, Y_{kp})$ for $k=1,\dots, n$. We generate i.i.d. samples from the following models :
\begin{enumerate}
\item $X_k \sim N(0, I_p)$ \,and \,$Y_k \sim N(0, I_p)$.
\item $X_k \sim AR(1), \phi = 0.5$, $Y_k \sim AR(1), \phi = -0.5$, where $AR(1)$ denotes the autoregressive model of order $1$ with parameter $\phi$.
\item $X_k \sim N(0, \Sigma)$ and $Y_k \sim N(0, \Sigma)$, where $\Sigma = (\sigma_{ij})_{i,j=1}^p$ with $\sigma_{ij} = 0.7^{|i-j|}$.
\end{enumerate}
\end{example}

\begin{example}\label{eg2}
Consider $X_k = (X_{k1}, \dots, X_{kp})$\, and \,$Y_k = (Y_{k1}, \dots, Y_{kp})$, $k=1,\dots, n$. We generate i.i.d. samples from the following models :
\begin{enumerate}
\item $X_k \sim N(0, I_p)$ \,and \,$Y_{kj} = X_{kj}^2$ for $j=1,\dots, p$.
\item $X_k \sim N(0, I_p)$ \,and \,$Y_{kj} = \log|X_{kj}|$ for $j=1,\dots, p$.
\item $X_k \sim N(0, \Sigma)$ and $Y_{kj} = X_{kj}^2$ for $j=1,\dots, p$, where $\Sigma = (\sigma_{ij})_{i,j=1}^p$ with $\sigma_{ij} = 0.7^{|i-j|}$.
\end{enumerate}
\end{example}

\begin{example}\label{eg3}
Consider $X_k = (X_{k1}, \dots, X_{kp})$\, and \,$Y_k = (Y_{k1}, \dots, Y_{kp})$, $k=1,\dots, n$. Let $\circ$ denote the Hadamard product of matrices. We generate i.i.d. samples from the following models:
\begin{enumerate}
\item $X_{kj} \sim U(-1,1)$ for $j=1,\dots, p$, and $Y_k = X_k \circ X_k$.
\item $X_{kj} \sim U(0,1)$ for $j=1,\dots, p$, and $Y_{k} = 4 X_k \circ X_k - 4X_k +2$.
\item $X_{kj} = \sin(Z_{kj})$ and $Y_{kj} = \cos(Z_{kj})$ with $Z_{kj} \sim U(0,2\pi)$ and $j=1,\dots,p$.
\end{enumerate}
\end{example}

For each example, we draw $1000$ simulated datasets and perform tests for independence between the two variables based on the proposed dependence metrics, and the usual Euclidean dCov and HSIC. 
We try a small sample size $n=50$ and dimensions $p=50, 100$ and $200$. For the tests II, III, V and VI, we chose the bandwidth parameter heuristically as the median distance between the sample observations. Table \ref{table1}  reports the proportion of rejections out of the $1000$ simulation runs for the different tests.

In Example \ref{eg1}, the data generating scheme suggests that the variables $X$ and $Y$ are independent. The results in Table \ref{table1} show that the tests based on the proposed dependence metrics perform almost equally good as the other competitors, and the rejection probabilities are quite close to the $10\%$ or $5\%$ nominal level. In Examples \ref{eg2} and \ref{eg3}, the variables are clearly (componentwise non-linearly) dependent by virtue of the data generating scheme. The results indicate that the tests based on the proposed dependence metrics are able to detect the componentwise non-linear dependence between the two high-dimensional random vectors unlike the tests based on the usual Euclidean dCov and HSIC, and thereby outperform the latter in terms of empirical power.

\begin{table}\scriptsize
\centering
\caption{Empirical size and power for the different tests of independence.}
\label{table1}
\begin{tabular}{c ccc cc cc cc cc cc cc}
\toprule
&&&&\multicolumn{2}{c}{I}&\multicolumn{2}{c}{II}&\multicolumn{2}{c}{III}&\multicolumn{2}{c}{IV} &\multicolumn{2}{c}{V} &\multicolumn{2}{c}{VI}
\\ \cmidrule(r){5-6}\cmidrule(r){7-8}\cmidrule(r){9-10}\cmidrule(r){11-12}\cmidrule(r){13-14}\cmidrule(r){15-16}
& & $n$ & $p$ &
10\% & 5\% & 10\% & 5\% & 10\% & 5\% & 10\% & 5\% & 10\% & 5\% & 10\% & 5\% \\
\hline
\multirow{9}{*}{Ex \ref{eg1}} & (1) & 50 & 50 & 0.115 & 0.053 & 0.109 & 0.055 & 0.106 & 0.053 & 0.112 & 0.060 & 0.112 & 0.053 & 0.111 & 0.061   \\
& (1) & 50 & 100 &  0.106 & 0.057 & 0.090 & 0.046 & 0.095 & 0.048 & 0.111 & 0.060 & 0.112 & 0.059 & 0.113 & 0.060  \\
    & (1) & 50 & 200 & 0.076 & 0.031 & 0.084 & 0.046 & 0.084 & 0.042 & 0.096 & 0.035 & 0.090 & 0.038 & 0.095 & 0.035  \\
    & (2) & 50 & 50 & 0.101 & 0.052 & 0.096 & 0.061 & 0.094 & 0.053 & 0.096 & 0.050 & 0.103 & 0.054 & 0.096 & 0.052   \\
    & (2) & 50 & 100 & 0.080 & 0.036 & 0.083 & 0.035 & 0.086 & 0.042 & 0.081 & 0.041 & 0.088 & 0.044 & 0.083 & 0.041  \\
     & (2) & 50 & 200 &  0.117 & 0.051 & 0.098 & 0.056 & 0.103 & 0.052 & 0.104 & 0.048 & 0.103 & 0.052 & 0.106 & 0.048  \\
     & (3) & 50 & 50 & 0.093 & 0.056 & 0.098 & 0.052 & 0.097 & 0.056 & 0.091 & 0.052 & 0.080 & 0.050 & 0.087 & 0.052   \\
     & (3) & 50 & 100 & 0.104 & 0.052 & 0.085 & 0.046 & 0.091 & 0.054 & 0.104 & 0.048 & 0.105 & 0.051 & 0.102 & 0.048 \\
     & (3) & 50 & 200 & 0.105 & 0.059 & 0.110 & 0.057 & 0.103 & 0.051 & 0.106 & 0.055 & 0.099 & 0.052 & 0.105 & 0.056 \\
\hline
\multirow{9}{*}{Ex \ref{eg2}} & (1) & 50 & 50 &  1 & 1 & 1 & 1& 1& 1& 0.267 & 0.172 & 0.534 & 0.398 & 0.277 & 0.182  \\
    & (1) & 50 & 100 &  1& 1& 1& 1& 1& 1& 0.171 & 0.102 & 0.284 & 0.180 & 0.167 & 0.102 \\
    & (1) & 50 & 200 & 1& 1& 1& 1& 1& 1& 0.130 & 0.075 & 0.194 & 0.108 & 0.128 & 0.073  \\
    & (2) & 50 & 50 & 1& 1& 1& 1& 1& 1& 0.154 & 0.092 & 0.199 & 0.130 & 0.154 & 0.091  \\
    & (2) & 50 & 100 & 1& 1& 1& 1& 1& 1& 0.109 & 0.050 & 0.128 & 0.064 & 0.108 & 0.049  \\
     & (2) & 50 & 200 & 1& 1& 1& 1& 1& 1& 0.099 & 0.057 & 0.107 & 0.060 & 0.097 & 0.057  \\
     & (3) & 50 & 50 & 1& 1& 1& 1& 1& 1& 0.654 & 0.546 & 0.981 & 0.959 & 0.708 & 0.631   \\
     & (3) & 50 & 100 &  1& 1& 1& 1& 1& 1& 0.418 & 0.309 & 0.790 & 0.700 & 0.455 & 0.343  \\
     & (3) & 50 & 200 & 1& 1& 1& 1& 1& 1& 0.277 & 0.188 & 0.504 & 0.391 & 0.284 & 0.193  \\
     \hline
\multirow{9}{*}{Ex \ref{eg3}} & (1) & 50 & 50 & 1& 1& 1& 1& 1& 1& 0.129 & 0.072 & 0.193 & 0.105 & 0.130 & 0.071  \\
    & (1) & 50 & 100 & 1& 1& 1& 1& 1& 1& 0.145 & 0.069 & 0.158 & 0.091 & 0.145 & 0.069   \\
    & (1) & 50 & 200 & 1& 1& 1& 1& 1& 1& 0.113 & 0.065 & 0.123 & 0.067 & 0.113 & 0.065 \\
    & (2) & 50 & 50 & 1& 1& 1& 1& 1& 1& 0.129 & 0.072 & 0.193 & 0.105 & 0.130 & 0.071 \\

    & (2) & 50 & 100 & 1& 1& 1& 1& 1& 1& 0.145 & 0.069 & 0.158 & 0.091 & 0.145 & 0.069  \\
     & (2) & 50 & 200 & 1& 1& 1& 1& 1& 1& 0.113 & 0.065 & 0.123 & 0.067 & 0.113 & 0.065  \\
     & (3) & 50 & 50 & 0.540 & 0.388 & 1& 1& 0.859 & 0.760 & 0.110 & 0.057 & 0.108 & 0.063 & 0.111 & 0.056  \\

     & (3) & 50 & 100 & 0.550 & 0.416 & 1 & 1 & 0.857 & 0.761 & 0.108 & 0.063 & 0.112 & 0.063 & 0.108 & 0.062   \\
     & (3) & 50 & 200 & 0.542 & 0.388 & 1& 1& 0.872 & 0.765 & 0.106 & 0.049 & 0.111 & 0.051 & 0.106 & 0.050  \\

\toprule
\end{tabular}
\\
\vspace{0.02in}
\end{table}

\subsection{Real data analysis}\label{sub:real}
\vspace{0.1in}

\subsubsection{Testing for homogeneity of distributions}

We consider the two sample testing problem of homogeneity of two high-dimensional distributions on Earthquakes data. The dataset has been downloaded from UCR Time Series Classification Archive (\url{https://www.cs.ucr.edu/~eamonn/time_series_data_2018/}). The data are taken from Northern California Earthquake Data Center. There are 368 negative and 93 positive earthquake events and each data point is of length 512.

Table \ref{table:real} shows the p-values corresponding to the different tests for the homogeneity of distributions between the two classes. Here we set $d_i=1$ for tests I-III.
Clearly the tests based on the proposed homogeneity metrics reject the null hypothesis of equality of distributions at $5\%$ level. However the tests based on the usual Euclidean energy distance and MMD fail to reject the null at $5\%$ level, thereby indicating no significant difference between the distributions of the two classes.

\begin{table}[!h]\footnotesize
\centering
\caption{p-values corresponding to the different tests for homogeneity of distributions for Earthquakes data.}
\label{table:real}
\begin{tabular}{cccccc}
\toprule
I & II & III & IV & V & VI \\
\hline
$2.27\times 10^{-93}$ & $3.19\times 10^{-86}$ & $9.74\times 10^{-110}$ & $0.070$ & $0.068$ & $0.070$\\
\toprule
\end{tabular}
\end{table}

\subsubsection{Testing for independence}

We consider the daily closed stock prices of $p=127$ companies under the finance sector and $q=125$ companies under the healthcare sector on the first dates of each month during the time period between January 1, 2017 and December 31, 2018. The data has been downloaded from Yahoo Finance via the R package `quantmod'. At each time $t$, denote the closed stock prices of these companies from the two different sectors by $X_t = (X_{1 t}, \dots, X_{p t})$ and $Y_t = (Y_{1 t}, \dots, Y_{q t})$ for $1\leq t \leq 24$. We consider the stock returns $S^X_t = (S^X_{1 t}, \dots, S^X_{p t})$ and $S^Y_t = (S^Y_{1 t}, \dots, S^Y_{q t})$ for $1\leq t \leq 23$, where $S^X_{i t} = \log \frac{X_{i, t+1}}{X_{i t}}$ and $S^Y_{j t} = \log \frac{Y_{j, t+1}}{Y_{j t}}$ for $1\leq i \leq p$ and $1\leq j \leq q$. It seems intuitive that the stock returns for the companies under two different sectors are not totally independent, especially when a large number of companies are being considered. Table \ref{table:real independence} shows the p-values corresponding to the different tests for independence between $\{S^X_t\}_{t=1}^{23}$ and $\{S^Y_t\}_{t=1}^{23}$, where we set $d_i=g_i=1$ for the proposed tests. The tests based on the proposed dependence metrics deliver much smaller p-values compared to the tests based on traditional metrics. We note that the tests based on the usual dCov and HSIC with the Laplace kernel fail to reject the null at $5\%$ level, thereby indicating cross-sector independence of stock return values. These results are consistent with the fact that the dependence among financial asset returns is usually nonlinear and thus cannot be fully characterized by traditional metrics in the high dimensional setup.

\begin{table}[!h]\footnotesize
\centering
\caption{p-values corresponding to the different tests for cross-sector independence of stock returns data.}
\label{table:real independence}
\begin{tabular}{cccccc}
\toprule
I & II & III & IV & V & VI \\
\hline
$5.70\times 10^{-13}$ & $2.36\times 10^{-10}$ & $7.99\times 10^{-11}$ & $0.120$ & $0.093$ & $0.040$\\
\toprule
\end{tabular}
\end{table}

We present an additional real data example on testing for independence in high dimensions in Section \ref{addl data ex} of the supplement. There the data admits a natural grouping, and our results indicate that our proposed tests for independence exhibit better power when we consider the natural grouping than when we consider unit group sizes. It is to be noted that considering unit group sizes makes our proposed statistics essentially equivalent to the marginal aggregation approach proposed by Zhu et al.\,(2019). This indicates that grouping or clustering might improve the power of testing as they are capable of detecting a wider range of dependencies.

\section{Discussions}
In this paper, we introduce a family of distances for high dimensional Euclidean spaces. Built on the new distances, we propose a class of distance and kernel-based metrics for high-dimensional two-sample and independence testing. The proposed metrics overcome certain limitations of the traditional metrics constructed based on the Euclidean distance. The new distance we introduce corresponds to a semi-norm given by
$$B(x)=\sqrt{\rho_1(x_{(1)})+\dots,\rho_p(x_{(p)})},$$
where $\rho_i(x_{(i)})=\rho_i(x_{(i)},0_{d_i})$ and $x = (x_{(1)},\dots,x_{(p)})\in\mathbb{R}^{\tilde{p}}$ with $x_{(i)}=(x_{i,1},\dots,x_{i,d_i}).$ Such a semi-norm has an interpretation based on a tree as illustrated by Figure \ref{fig}.

\begin{figure}[!h]
\caption{An interpretation of the semi-norm $B(\cdot)$ based on a tree}\label{fig}
\centering
\includegraphics[height=6cm,width=10cm]{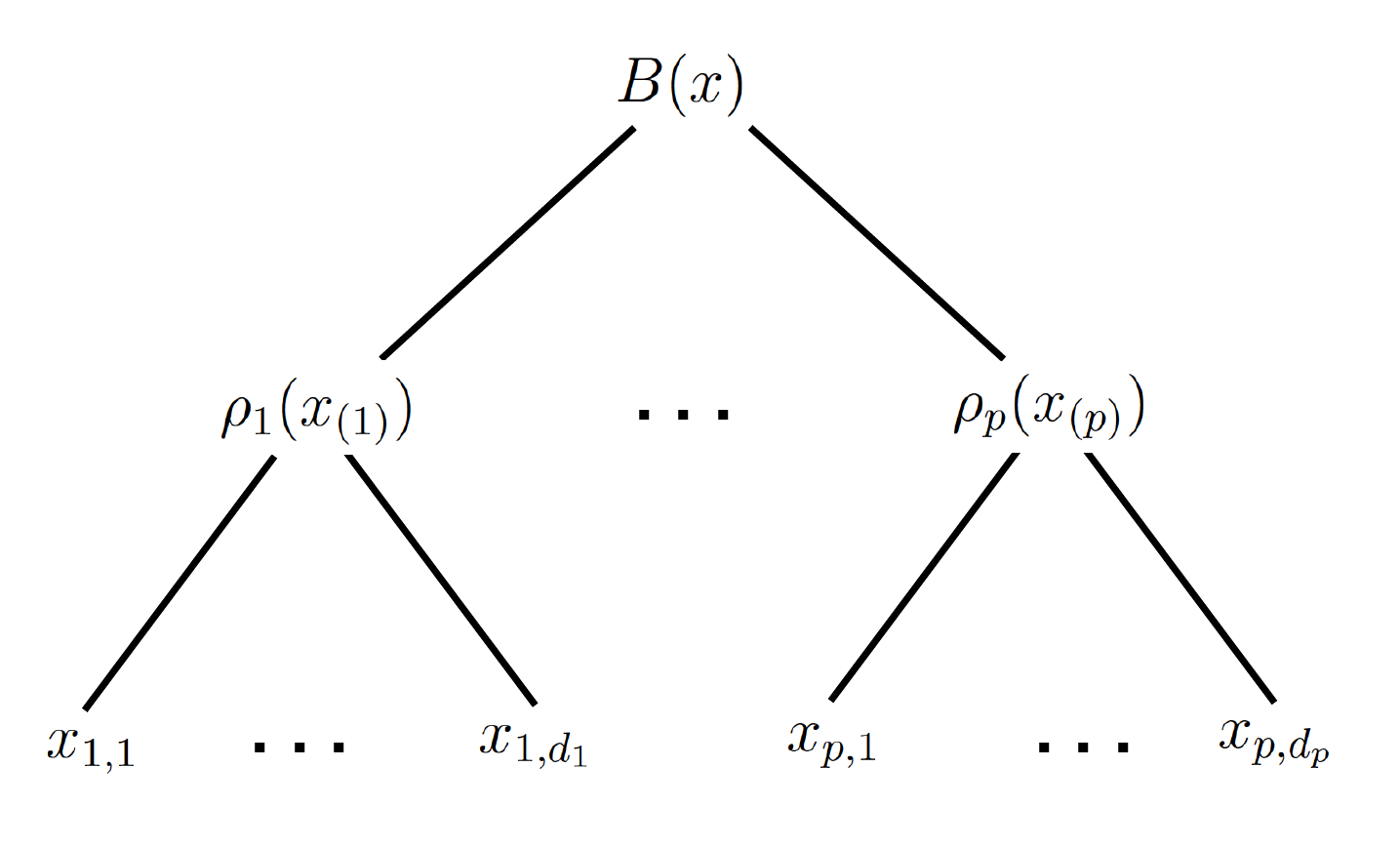}
\end{figure}

Tree structure provides useful information for doing grouping at different levels/depths. Theoretically, grouping allows us to detect a wider range of alternatives. For example, in two-sample testing, the difference between two one-dimensional marginals is always captured by the difference between two higher dimensional marginals that contain the two one-dimensional marginals. The same thing is true for dependence testing. Generally, one would like to find blocks which are nearly independent, but the variables inside a block have significant dependence among themselves. It is interesting to develop an algorithm for finding the optimal groups using the data and perhaps some auxiliary information. Another interesting direction is to study the semi-norm and distance constructed based on a more sophisticated tree structure. For example, in microbiome-wide association studies, phylogenetic tree or evolutionary tree which is a branching diagram or ``tree'' showing the evolutionary relationships among various biological species. Distance and kernel-based metrics constructed based on the distance utilizing the phylogenetic tree information is expected to be more powerful in signal detection. We leave these topics for future investigation.


\vspace{0.1in}

\begin{alphasection}

\begin{center}
{\bf \Large Supplement to ``A New Framework for Distance and Kernel-based Metrics in High Dimensions"}\\
\vspace{0.1in}
\large Shubhadeep Chakraborty\\ \vspace{-0.1in} Department of Statistics, Texas A\&M University\\ \vspace{-0.1in} and \\ \vspace{-0.1in} Xianyang Zhang\\ \vspace{-0.1in} Department of Statistics, Texas A\&M University\\
\end{center}
\spacingset{1.45} 

\vspace{0.3cm}
The supplement is organized as follows. In Section \ref{ld app} we 
explore our proposed homogeneity and dependence metrics in the low-dimensional setup. In Section \ref{HDMSS app} we study the asymptotic behavior of our proposed homogeneity and dependence metrics in the high dimension medium sample size (HDMSS) framework where both the dimension(s) and the sample size(s) grow. Section \ref{addl data ex} illustrates an additional real data example for testing for independence in the high-dimensional framework. Finally, Section \ref{technical} contains additional proofs of the main results in the paper and Sections \ref{ld app} and \ref{HDMSS app} in the supplement.

\section{Low-dimensional setup}\label{ld app}

In this section we illustrate that the new class of homogeneity metrics proposed in this paper inherits all the nice properties of generalized energy distance and MMD in the low-dimensional setting. Likewise, the proposed dependence metrics inherit all the desirable properties of generalized dCov and HSIC in the low-dimensional framework.

\subsection{Homogeneity metrics}\label{ld E}
Note that in either Case S1 or S2, the Euclidean space equipped with distance $K$ is of strong negative type.
As a consequence, we have the following result.
\begin{theorem}\label{ED HD homo}
$\cal{E}(X,Y) = 0$ if and only if $X \overset{d}{=} Y$, in other words $\cal{E}(X,Y)$ completely characterizes the homogeneity of the distributions of $X$ and $Y$.
\end{theorem}

The following proposition shows that $\cal{E}_{n,m}(X,Y)$ is a two-sample U-statistic and an unbiased estimator of $\cal{E} (X,Y)$.

\begin{proposition}\label{prop ED U stat}
The U-statistic type estimator enjoys the following properties:
\begin{enumerate}
\item $\cal{E}_{n,m}$ is an unbiased estimator of the population $\cal{E}$.
\item $\cal{E}_{n,m}$ admits the following form : $$\cal{E}_{n,m}(X,Y) \;=\; \frac{1}{\binom{n}{2} \,\binom{m}{2}} \dis \sum_{1 \leq i < j \leq n}\, \sum_{1 \leq k < l \leq m} h(X_i, X_j ; Y_k, Y_l)\,,$$ where $$h(X_i, X_j ; Y_k, Y_l) \;=\; \frac{1}{2}\Big( K(X_i, Y_k) \,+\,K(X_i, Y_l)\,+\,K(X_j, Y_k)\,+\,K(X_j, Y_l)\Big) \,-\,  K(X_i, X_j) \,-\, K(Y_k, Y_l)\,.$$
\end{enumerate}
\end{proposition}

The following theorem shows the asymptotic behavior of the U-statistic type estimator of $\cal{E}$ for fixed $p$ and growing $n$.


\begin{theorem}\label{th ED U stat}
Under Assumption \ref{ass0.5} and the assumption that $\sup_{1\leq i\leq p}\E \rho_i(X_{(i)},0_{d_i})< \infty$ and $\sup_{1\leq i\leq p}\E \rho_i(Y_{(i)},0_{d_i})< \infty$, as $m, n \to \infty$ with $p$ remaining fixed, we have the following:
\begin{enumerate}
\item $\cal{E}_{n,m} (X,Y) \,\overset{a.s.}{\longrightarrow}\, \cal{E} (X,Y)$.
\item When $X \overset{d}{=} Y$, $\cal{E}_{n,m}$ has degeneracy of order $(1,1)$, and 
$$\frac{(m-1)(n-1)}{n+m} \, \cal{E}_{n,m} (X,Y) \, \overset{d}{\longrightarrow} \,  \sum_{k=1}^{\infty} \lambda_k^2 \left(Z_k^2 \,-\, 1 \right)\,,$$ where $\{Z_k\}$ is a sequence of independent $N(0,1)$ random variables and $\lambda_k$'s depend on the distribution of $(X, Y)$.
\end{enumerate}
\end{theorem}

Proposition \ref{prop ED U stat}, Theorem \ref{ED HD homo} and Theorem \ref{th ED U stat} demonstrate that $\cal{E}$ inherits all the nice properties of generalized energy distance and MMD in the low-dimensional setting.

\subsection{Dependence metrics}\label{ld D}

Note that Proposition \ref{metric} in Section \ref{new distance} and Proposition 3.7 in Lyons\,(2013) ensure that $\cal{D}(X,Y)$ completely characterizes independence between $X$ and $Y$, which leads to the following result.

\begin{theorem}\label{characterization}
Under Assumption \ref{ass1}, $\cal{D}(X,Y) = 0$ if and only if $X \bigCI Y$.
\end{theorem}

The following proposition shows that $\widetilde{\cal{D}^2_n}(X,Y)$ is an unbiased estimator of $\cal{D}^2 (X,Y)$ and is a U-statistic of order four.

\begin{proposition}\label{unbiased and ustat form} The U-statistic type estimator $\widetilde{\cal{D}^2_n}$ (defined in (\ref{ustat dcov}) in the main paper) has the following properties:
\begin{enumerate}
\item $\widetilde{\cal{D}^2_n}$ is an unbiased estimator of the squared population $\cal{D}^2$.
\item $\widetilde{\cal{D}^2_n}$ is a fourth-order U-statistic which admits the following form: $$ \widetilde{\cal{D}^2_n} \;=\; \frac{1}{\binom{n}{4}} \dis \sum_{i<j<k<l} h_{i,j,k,l}\,,$$ where \begin{align*}
h_{i,j,k,l} \,&= \, \frac{1}{4!} \dis \sum_{(s,t,u,v)}^{(i,j,k,l)} (d^X_{st} d^Y_{st} + d^X_{st} d^Y_{uv} - 2 d^X_{st} d^Y_{su}) \\ &= \, \frac{1}{6} \dis \sum_{s<t, u<v}^{(i,j,k,l)} (d^X_{st} d^Y_{st} + d^X_{st} d^Y_{uv}) - \frac{1}{12} \dis \sum_{(s,t,u)}^{(i,j,k,l)} d^X_{st} d^Y_{su}\,,
\end{align*}
the summation is over all possible permutations of the $4$-tuple of indices $(i,j,k,l)$.
For example, when $(i,j,k,l)=(1,2,3,4)$, there exist 24 permutations, including $(1,2,3,4),\dots,(4,3,2,1)$. Furthermore, $\widetilde{\cal{D}^2_n}$ has degeneracy of order 1 when $X$ and $Y$ are independent.
\end{enumerate}
\end{proposition}

The following theorem shows the asymptotic behavior of the U-statistic type estimator of $\cal{D}^2$ for fixed $p, q$ and growing $n$.
\begin{theorem}\label{ACdCov U-stat prop}
Under Assumption \ref{ass1}, with fixed $p, q$ and $n \to \infty$, we have the following as $n \to \infty$:
\begin{enumerate}
\item $\widetilde{\cal{D}^2_n}(X,Y) \overset{a.s.}{\longrightarrow} \cal{D}^2(X,Y)$;
\item When $\cal{D}^2(X,Y)=0$ (i.e., $X \bigCI Y$), $n\,\widetilde{\cal{D}^2_n}(X,Y)\, \overset{d}{\longrightarrow} \dis\sum_{i=1}^{\infty} \tilde{\lambda}^2_i (Z_i^2 -1)$,\, where $Z_i's$ are i.i.d. standard normal random variables and \,$ \tilde{\lambda}_i$'s depend on the distribution of $(X, Y)$;
\item When $\cal{D}^2(X,Y)>0$, $n\,\widetilde{\cal{D}^2_n}(X,Y) \, \overset{a.s.}{\longrightarrow}\infty$.
\end{enumerate}
\end{theorem}
Proposition \ref{unbiased and ustat form}, Theorem \ref{characterization} and Theorem \ref{ACdCov U-stat prop} \, demonstrate that in the low-dimensional setting, $\cal{D}$ inherits all the nice properties of generalized dCov and HSIC.

\section{High dimension medium sample size (HDMSS)}\label{HDMSS app}
\subsection{Homogeneity metrics}\label{subsec:ED HDMSS}
In this subsection, we consider the HDMSS setting where $p \to \infty$ and  $n, m \to \infty$ at a slower rate than $p$. Under $H_0$, we impose the following conditions to obtain the asymptotic null distribution of the statistic $T_{n,m}$ under the HDMSS setup.

\begin{assumption}\label{assED_HDMSS}
As $n, m$ and $p \to \infty$,
\begin{align*}
& \frac{1}{n^2}\, \frac{\E\,\left[H^4(X,X')\right]}{\left(\E\,\left[H^2(X, X')\right]  \right)^2} \;=\; o(1),\quad \frac{1}{n}\, \frac{\E\,\left[H^2(X,X'')\, H^2(X', X'')\right]}{\left(\E\,\left[H^2(X, X')\right]  \right)^2} \;=\; o(1),\\ & \frac{\E\,\left[H(X,X'')\, H(X', X'')\,H(X, X''')\,H(X', X''')\right]}{\left(\E\,\left[H^2(X, X')\right]  \right)^2} \;=\; o(1).
\end{align*}
\end{assumption}

\begin{remark}\label{rem_ass4.6}
We refer the reader to Section 2.2 in Zhang et al.\,(2018) and Remark A.2.2 in Zhu et al.\,(2019) for illustrations of Assumption \ref{assED_HDMSS} where  $\rho_i$ has been considered to be the Euclidean distance or the squared Euclidean distance, respectively, for $1\leq i\leq p$.  
\end{remark}




\begin{assumption}\label{assED_HDMSS4}
Suppose $\E\, [L^2(X, X')] = O(\alpha_p^2)$ where $\alpha_p$ is a positive real sequence such that $\tau_X  \alpha_p^2 = o(1)$ as $p \to \infty$. Further assume that as\, $n, p \to \infty$, $$\frac{n^4 \,\tau_X^4\, \E\,\left[ R^4(X,X')\right]}{\left( \E\,\left[ H^2(X,X')\right] \right)^2} = o(1)\,.$$
\end{assumption}

\begin{remark}\label{expl for assED_HDMSS4}
We refer the reader to Remark \ref{assumption justification} in the main paper which illustrates some sufficient conditions under which \,$\alpha_p = O(\frac{1}{\sqrt{p}})$\, and consequently\, $\tau_X  \alpha_p^2  = o(1)$ holds, as \,$\tau_X \asymp p^{1/2}$. In similar lines of Remark \ref{rem_ui} in Section \ref{technical} of the supplementary material, it can be argued that\, $\E\,\left[ R^4(X,X')\right] = O\left(\frac{1}{p^4}\right)$. If we further assume that Assumption \ref{ass6} holds, then we have $\E\,\left[ H^2(X,X')\right] \asymp 1$. Combining all the above, it is easy to verify that\, $\frac{n^4 \,\tau_X^4\, \E\,\left[ R^4(X,X')\right]}{\left( \E\,\left[ H^2(X,X')\right] \right)^2} = o(1)$\, holds provided \,$n = o(p^{1/2})$.
\end{remark}

The following theorem illustrates the limiting null distribution of $T_{n,m}$ under the HDMSS setup. We refer the reader to Section \ref{technical} of the supplement for a detailed proof.

\begin{theorem}\label{th_ED_HDMSS}
Under $H_0$ and Assumptions \ref{ass0.5}, \ref{assED_HDMSS} and \ref{assED_HDMSS4}, as $n, m$ and $p \to \infty$, we have
$$T_{n,m} \; \overset{d}{\longrightarrow} \; N(0,1).$$
\end{theorem}

\subsection{Dependence metrics}\label{sec:D HDMSS}
In this subsection, we consider the HDMSS setting where $p, q \to \infty$ and $n \to \infty$ at a slower rate than $p, q$. The following theorem shows that similar to the HDLSS setting, under the HDMSS setup, $\widetilde{\cal{D}^2_n}$ is asymptotically equivalent to the aggregation of group-wise generalized dCov. In other words $\widetilde{\cal{D}^2_n} (X,Y)$ can quantify group-wise nonlinear dependence between $X$ and $Y$ in the HDMSS setup as well.

\begin{assumption}\label{ass4}
$\E[L_X(X,X')^2] = \alpha_p^2$, $\E[L_X(X,X')^4] = \gamma_p^2$, $\E[L_Y(Y,Y')^2] = \beta_q^2$ and $\E[L_Y(Y,Y')^4] = \lambda_q^2$, where $\alpha_p, \gamma_p, \beta_q, \lambda_q$ are positive real sequences satisfying $n \alpha_p = o(1)$, $n \beta_q = o(1)$, $\tau_X^2 (\alpha_p \gamma_p + \gamma_p^2) = o(1)$, $\tau_Y^2 (\beta_q \lambda_q + \lambda_q^2) = o(1)$, and\, $\tau_{XY}\, (\alpha_p \lambda_q + \gamma_p \beta_q + \gamma_p \lambda_q ) = o(1)$.
\end{assumption}

\begin{remark}\label{assumption justification for rem5.5}
Following Remark \ref{assumption justification} in the main paper, we can write \,
$L(X,X') = O(\frac{1}{p})  \sum_{i=1}^p \left(Z_i - \E\, Z_i \right)$, where $Z_i = \rho_i( X_{(i)}, X_{(i)}')$ for $1\leq i \leq p$. Assume that $ \sup_{1\leq i \leq p} \E\,\rho_i^4( X_{(i)}, 0_{d_i}) < \infty$, which implies \,$ \sup_{1\leq i \leq p} \E\,Z_i^4 < \infty$. Under certain weak dependence assumptions, it can be shown that \,$\E\,\big( \sum_{i=1}^p ( Z_i - \E\, Z_i)\big)^4 = O(p^2)$ as $p\to \infty$ (see for example Theorem 1 in Doukhan et al.\,(1999)). Therefore we have \,$\E [L(X,X')^4] = O(\frac{1}{p^2})$. It follows from H{\"o}lder's inequality that \,$\E [L(X,X')^2 ] = O(\frac{1}{p})$. Similar arguments can be made about \,$\E [L(Y,Y')^4]$ and $\E [L(Y,Y')^2]$ as well. 
\end{remark}

\begin{theorem}\label{ACdCov hd}
Under Assumptions \ref{ass2 : ED} and \ref{ass4}, we can show that
\begin{equation}\label{ACdCov taylor HDMSS}
\widetilde{\cal{D}^2_n} (X,Y) = \frac{1}{4\tau_{XY}} \dis \sum_{i=1}^p \sum_{j=1}^q \widetilde{D^2_n}_{\,;\,\rho_i, \rho_j}(X_{(i)},Y_{(j)}) \, + \, \mathcal{R}_n \; ,
\end{equation}
where $\mathcal{R}_n$ is the remainder term satisfying that $\mathcal{R}_n = O_p(\tau_{XY}\, (\alpha_p \lambda_q + \gamma_p \beta_q + \gamma_p \lambda_q )) = o_p(1)$, i.e., $\mathcal{R}_n$ is of smaller order compared to the leading term and hence is asymptotically negligible.
\end{theorem}

The following theorem states the asymptotic null distribution of the studentized test statistic $\mathcal{T}_n$ (given in equation (\ref{student t}) in the main paper) under the HDMSS setup. Define $$ U(X_k, X_l) := \frac{1}{\tau_X} \dis \sum_{i=1}^p d^X_{kl}(i), \quad \text{and} \quad V(Y_k, Y_l) := \frac{1}{\tau_Y} \dis \sum_{i=1}^q d^Y_{kl}(i).$$

\begin{assumption}\label{ass5}
Assume that
\begin{align*}
&\frac{\E \left[U(X, X')\right]^4}{\sqrt{n}\left(\E [U(X, X')]^2 \right)^2} \,=\, o(1),
\\ &\frac{\E \left[ U(X, X')\, U(X', X'')\, U(X'', X''')\,U(X''', X) \right]}{\left(\E [U(X, X')]^2 \right)^2}\, =\, o(1),
\end{align*}
and the same conditions hold for $Y$ in terms of $V(Y,Y')$.
\end{assumption}

\begin{remark}\label{rem HDdCov HDMSS}
We refer the reader to Section 2.2 in Zhang et al.\,(2018) and Remark A.2.2 in Zhu et al.\,(2019) for illustrations of Assumption \ref{assED_HDMSS} where $\rho_i$ has been considered to be the Euclidean distance or the squared Euclidean distance, respectively.

\end{remark}

We can show that under $H_0$, the studentized test $\mathcal{T}_n$ converge to the standard
normal distribution under the HDMSS setup.
\begin{theorem}\label{HDMSS dist conv}
Under $H_0$ and Assumptions \ref{ass4}-\ref{ass5}, as $n, p, q \to \infty$, we have\; $\mathcal{T}_n \, \overset{d}{\longrightarrow} \, N(0,1)\,.$
\end{theorem}

\section{Additional real data example}\label{addl data ex}

We consider the monthly closed stock prices of $\tilde{p}=33$ companies under the oil and gas sector and $\tilde{q}=34$ companies under the transport sector between January 1, 2017 and December 31, 2018. The companies under both the sectors are clustered or grouped according to their countries. The data has been downloaded from Yahoo Finance via the R package `quantmod'. Under the oil and gas sector, we have $p=18$ countries or groups, viz. USA, Australia, UK, Canada, China, Singapore, Hong Kong, Netherlands, Colombia, Italy, Norway, Bermuda, Switzerland, Brazil, South Africa, France, Turkey and Argentina, with $d = (5,1,2,5,4,1,1,2,1,1,1,1,2,1,1,2,1,1)$. And under the transport sector, we have $q=14$ countries or groups, viz. USA, Brazil, Canada, Greece, China, Panama, Belgium, Bermuda, UK, Mexico, Chile, Monaco, Ireland and Hong Kong, with $g = (5,1,2,6,4,1,1,3,1,3,1,4,1,1)$. At each time $t$, denote the closed stock prices of these companies from the two different sectors by $X_t = (X_{1 t}, \dots, X_{p t})$ and $Y_t = (Y_{1 t}, \dots, Y_{q t})$ for $1\leq t \leq 24$. We consider the stock returns $S^X_t = (S^X_{1 t}, \dots, S^X_{p t})$ and $S^Y_t = (S^Y_{1 t}, \dots, S^Y_{q t})$ for $1\leq t \leq 23$, where $S^X_{i t l} = \log \frac{X_{i, t+1, l}}{X_{i t l}}$ and $S^Y_{j t l'} = \log \frac{Y_{j, t+1, l'}}{Y_{j t l'}}$ for $1\leq l\leq d_i$, $1\leq i \leq p$, $1\leq l'\leq g_j$ and $1\leq j \leq q$. 

The intuitive idea is, stock returns of oil and gas companies should affect the stock returns of companies under the transport sector, and here both the random vectors admit a natural grouping based on the countries. Table \ref{table:real independence 2} shows the p-values corresponding to the different tests for independence between $\{S^X_t\}_{t=1}^{23}$ and $\{S^Y_t\}_{t=1}^{23}$. The tests based on the proposed dependence metrics considering the natural grouping deliver much smaller p-values compared to the tests based on the usual dCov and HSIC which fail to reject the null hypothesis of independence between $\{S^X_t\}_{t=1}^{23}$ and $\{S^Y_t\}_{t=1}^{23}$. This makes intuitive sense as the dependence among financial asset returns is usually nonlinear in nature and thus cannot be fully characterized by the usual dCov and HSIC in the high dimensional setup.

\begin{table}[!h]\footnotesize
\centering
\caption{p-values corresponding to the different tests for cross-sector independence of stock returns data considering the natural grouping based on countries.}
\label{table:real independence 2}
\begin{tabular}{cccccc}
\toprule
I & II & III & IV & V & VI \\
\hline
$0.036$ & $0.048$ & $0.087$ & $0.247$ & $0.136$ & $0.281$\\
\toprule
\end{tabular}
\end{table}

Table \ref{table:real independence 2.5} shows the p-values corresponding to the different tests for independence when we disregard the natural grouping and consider $d_i =1$ and $g_j =1$ for all $1\leq i \leq p$ and $1\leq j \leq q$. Considering unit group sizes makes our proposed statistics essentially equivalent to the marginal aggregation approach proposed by Zhu et al.\,(2019). In this case the proposed tests have higher p-values than when we consider the natural grouping, indicating that grouping or clustering might improve the power of testing as they are capable of detecting a wider range of dependencies.

\begin{table}[!h]\footnotesize
\centering
\caption{p-values corresponding to the different tests for cross-sector independence of stock returns data considering unit group sizes.}
\label{table:real independence 2.5}
\begin{tabular}{cccccc}
\toprule
I & II & III & IV & V & VI \\
\hline
$0.092$ & $0.209$ & $0.226$ & $0.247$ & $0.136$ & $0.281$\\
\toprule
\end{tabular}
\end{table}

\section{Technical Appendix}\label{technical}

\begin{proof}[Proof of Proposition \ref{metric}]
To prove (1), note that if $d$ is a metric on a space $\cal{X}$, then so is $d^{1/2}$. It is easy to see that $K^2$ is a metric on $\bb{R}^{\tilde{p}}$. To prove (2), note that $(\bb{R}^{d_i}, \rho_i)$ has strong negative type for $1 \leq i \leq p$. The rest follows from Corollary 3.20 in Lyons\,(2013).
\end{proof}

\begin{proof}[Proof of Proposition \ref{prop ED U stat}]
It is easy to verify that $\cal{E}_{n,m}$ is an unbiased estimator of $\cal{E}$ and is a two-sample U-statistic with the kernel $h$.
\end{proof}

\begin{proof}[Proof of Theorem \ref{th ED U stat}]
The first part of the proof follows from Theorem 1 in Sen\,(1977) and the observation that $\E\,\left[|h| \log^+ |h|\right] \leq \E[h^2]$. The power mean inequality says that for \,$a_i \in \mathbb{R},\, 1\leq i \leq n, \, n\geq 2$ and $r>1$,
\begin{align}\label{power mean ineq}
 \left|\dis \sum_{i=1}^n a_i \right|^r \; \leq \; n^{r-1} \, \dis \sum_{i=1}^n |a_i|^r \,.
\end{align} Using the power mean inequality, it is easy to see that the assumptions $\sup_{1\leq i\leq p}\E \rho_i(X_{(i)},0_{d_i})< \infty$ and $\sup_{1\leq i\leq p}\E \rho_i(Y_{(i)},0_{d_i})< \infty$ ensure that $\E[h^2] < \infty$. For proving the second part, define
$h_{1,0}(X)=\; \E\, \left[ h(X, X'; Y, Y') | X \right]$ and
$h_{0,1}(Y)=\; \E\, \left[ h(X, X'; Y, Y') | Y \right]$
Clearly, when $X\overset{d}{=}Y$, $h_{1,0}(X)$ and $h_{0,1}(Y)$ are degenerate at $0$ almost surely. Following Theorem 1.1 in Neuhaus\,(1977), we have $$\frac{(m-1)(n-1)}{n+m} \, \cal{E}_{n m} (X,Y) \, \overset{d}{\longrightarrow} \,  \sum_{k=1}^{\infty} \sigma_k^2 \left[(a_k U_k + b_k V_k )^2 \,-\, (a_k^2 + b_k^2) \right]\,,$$ where $\{U_k\}, \{V_k\}$ are two sequences of independent $N(0,1)$ variables, independent of each other, and $(\sigma_k, a_k, b_k)$'s depend on the distribution of $(X,Y)$. The proof can be completed by some simple rearrangement of terms.
\end{proof}

\begin{proof}[Proof of Proposition \ref{K taylor : ED}]
The proof is essentially similar to the proof of Proposition 2.1.1 in Zhu et al. (2019), replacing the Euclidean distance between, for example, $X$ and $X'$, viz. $\Vert X-X' \Vert_{\tilde{p}}$\,, by the new distance metric $K(X,X')$. To show that $R(X,X')=O_p(L^2(X,X'))$ if $L(X,X')=o_p(1)$, we define $f(x) = \sqrt{1+x}$. By the definition of the Lagrange's form of the remainder term from Taylor's expansion, we have $$R(X,X') = \int_0^{L(X,X')} f''(t) \left(\,L(X,X') - t\,\right)\, dt\,.$$ Using $R$ and $L$ interchangeably with $R(X,X')$ and $L(X,X')$ respectively, we can write
\begin{equation}\label{eq_taylor}
\begin{split}
|R| \; & \leq \; |L| \, \left[ \int_0^L f''(t) \, \mathbbm{1}_{L>0} \, dt \;+\; \int_L^0 f''(t) \, \mathbbm{1}_{L<0} \, dt \,\right]\\
&= \;\frac{|L|}{2}\, \big|1 - \frac{1}{\sqrt{1+L}}\big|\\
&= \; \frac{|L|}{2}\, \frac{|L|}{1 + L + \sqrt{1+L}}\\
&\leq \; \frac{L^2}{2(1+L)}\,.
\end{split}
\end{equation}
It is clear that $R(X,X')=O_p(L^2(X,X'))$ provided that $L(X,X')=o_p(1)$.
\end{proof}

\begin{proof}[Proof of Theorem \ref{th homo decomp}]
Observe that $\E\,L(X,Y) = \E\,L(X,X') = \E\,L(Y,Y') = 0$. By Proposition \ref{K taylor : ED},
\begin{align*}
\cal{E}(X,Y) \;&=\; 2\,\E\,\left[\tau + \tau\,R(X,Y)\right] \,-\, \E\,\left[\tau_X + \tau_X\,R(X,X')\right] \,-\, \E\,\left[\tau_Y + \tau_Y\,R(Y,Y')\right] \\
&=\; 2\tau -\tau_X - \tau_Y \,+\, \cal{R}_{\cal{E}}\,.
\end{align*}
Clearly $\vert \cal{R}_{\cal{E}} \vert\; \leq \; 2\,\tau\,\E\,\left[\,\vert R(X,Y)\vert \,\right] \,+\, \tau_X\,\E\,\left[\,\vert R(X,X')\vert \,\right] \,+\, \tau_Y\,\E\,\left[\,\vert R(Y,Y')\vert \,\right].$ By (\ref{eq_taylor}) and Assumption \ref{ass0.6}, we have
$$\tau |R(X,Y)|\leq \frac{\tau L^2(X,Y)}{2(1+L(X,Y))}=O(\tau a^{2}_p)=o_p(1).$$
As $\{\sqrt{p} L^2(X,Y)/(1+L(X,Y))\}$ is uniformly integrable and $\tau\asymp \sqrt{p}$, we must have
$\tau \E[|R(X,Y)|]=o(1)$. The other terms can be handled in a similar fashion.
\end{proof}

\begin{remark}\label{rem_ui}
Write $L(X,Y) = \frac{1}{\tau^2} (A_p - \E\,A_p) = \frac{1}{\tau^2} \sum_{i=1}^p (Z_i - \E Z_i)$, where $A_p := \sum_{i=1}^p Z_i$ and $Z_i := \rho_i (X_i, Y_i)$ for $1\leq i \leq p$. Assume $\sup_i \E \rho_i^8 (X_i, 0_{d_i}) < \infty$ and\, $\sup_i \E \rho_i^8 (X_i, 0_{d_i}) < \infty$, which imply $\sup_i \E Z_i^8 < \infty$. Denote $L(X,Y)$ by $L$ and $R(X,Y)$ by $R$ for notational simplicities. Further assume that $E \exp(t A_p)=O((1-\theta_1 t)^{-\theta_2 p})$ for $\theta_1, \theta_2 >0$ and $\theta_2 \,p > 4$ uniformly over\, $t<0$ 
(which is clearly satisfied when $Z_i$'s are independent and 
$\E\exp(t Z_i)\leq a_1(1-a_2 t)^{-a_3}$ 
uniformly over $t<0$ and \, $1\leq i\leq p$ for some $a_1,a_2,a_3>0$ with $a_3\, p > 4$). Under certain weak dependence assumptions, it can be shown that:
\begin{enumerate}
\item $\{\sqrt{p} L^2/(1+L)\}$ is uniformly integrable;
\item $\E\, R^2 = O(\frac{1}{p^2})$.
\end{enumerate}
Similar arguments hold for $L(X,X')$ and $R(X,X')$, and, $L(Y,Y')$ and $R(Y,Y')$ as well.
\end{remark}

\begin{proof}[Proof of Remark \ref{rem_ui}]
To prove the first part, define $L_p := \sqrt{p} L^2/(1+L)$. Following Chapter 6 of Resnick\,(1999), it suffices to show that $\sup_p \E\,L_p^2 < \infty$. Towards that end, using  H{\"o}lder's inequality we observe 
\begin{align}\label{eqn rem_ui}
\E\,L_p^2 \;& \leq \left(\E (p^2 L^8) \right)^{1/2} \, \left(\E \Big[\frac{1}{(1+L)^4}\Big]\right)^{1/2}\,.
\end{align}

With $\sup_i \E Z_i^8 < \infty$ and under certain weak dependence assumptions, it can be shown that $\E (A_p - \E A_p)^8 = O(p^4)$ (see for example Theorem 1 in Doukhan et al.\,(1999)). Consequently we have $\E\,L^8 = O(\frac{1}{p^4})$ , as $\tau\asymp \sqrt{p}$. Clearly this yields\, $\E\,(p^2 L^8) = O(\frac{1}{p^2})$.

Now note that 
\begin{align}\label{eqn rem_ui 2}
\E \Big[\frac{1}{(1+L)^4}\Big] \;&=\; \tau^8\, \E \left(\frac{1}{A_p^4}\right)\,.
\end{align}
Equation (3) in Cressie et al.\,(1981) states that for a non-negative random variable $U$ with moment-generating function $M_U(t) = \E \exp(tU)$, one can write 
\begin{align}\label{eqn rem_ui 2.5}
\E(U^{-k}) &= (\Gamma(k))^{-1} \dis \int_0^{\infty} t^{k-1} M_U(-t) \, dt\,,
\end{align}
for any positive integer $k$, provided both the integrals exist. Using equation (\ref{eqn rem_ui 2.5}), the assumptions stated in Remark \ref{rem_ui} and basic properties of beta integrals, some straightforward calculations yield
\begin{align}\label{eqn 3 calc}
\E \left(\frac{1}{A_p^4}\right) \;&\leq\; C_1 \, \int_0^{\infty} \frac{t^{4-1}}{(1+\theta_1 t)^{\theta_2 p}} \, dt \;=\; C_2 \, \frac{\Gamma(\theta_2 p -4)}{\Gamma(\theta_2 p)}\,,
\end{align}
where $C_1, C_2$ are positive constants, which clearly implies that $\E \left(\frac{1}{A_p^4}\right) = O(\frac{1}{p^4})$. This together with equation (\ref{eqn rem_ui 2}) implies that $\E \Big[\frac{1}{(1+L)^4}\Big] = O(1)$, as $\tau\asymp \sqrt{p}$. 

Combining all the above, we get from (\ref{eqn rem_ui}) that $\E\,L_p^2 = O(\frac{1}{p})$ and therefore $\sup_p \E\,L_p^2 < \infty$, which completes the proof of the first part.

To prove the second part, note that following the proof of Proposition \ref{K taylor : ED} and H{\"o}lder's inequality we can write 
\begin{align}\label{R_expr}
\E\,R^2 = O\left( \E\,\left[\frac{L^4}{(1+L)^2} \right]\right) = O\left( \left(\E (L^8) \right)^{1/2} \, \left(\E \Big[\frac{1}{(1+L)^4}\Big] \right)^{1/2} \right)\,.
\end{align}
Following the arguments as in the proof of the first part, clearly we have  $\E\,L^8 = O(\frac{1}{p^4})$ and $\E \Big[\frac{1}{(1+L)^4}\Big] = O(1)$. From this and equation (\ref{R_expr}), it is straightforward to verify that $\E\, R^2 = O(\frac{1}{p^2})$, which completes the proof of the second part.
\end{proof}

\begin{proof}[Proof of Lemma \ref{lemma 1}]
To see (2), first observe that the sufficient part is straightforward from equation (\ref{tau for Eucl}) in the main paper. For the necessary part, denote $a = \textrm{tr}\, \Sigma_X$, $b = \textrm{tr}\, \Sigma_Y$ and $c = \Vert \mu_X - \mu_Y \Vert^2$. Then we have $2\,\sqrt{a + b + c} = \sqrt{2 a} + \sqrt{2 b}$. Some straightforward calculations yield $(\sqrt{2 a} - \sqrt{2 b})^2 + 4\,c = 0$ which implies the rest.

To see (1), again the sufficient part is straightforward from equation (\ref{tau def original}) in the paper and the form of $K$ given in equation (\ref{Kdef}) in the paper. For the necessary part, first note that as $(\bb{R}^{d_i}, \rho_i)$ is a metric space of strong negative type for $1 \leq i \leq p$, there exists a Hilbert space $\cal{H}_i$ and an injective map $\phi_i : \bb{R}^{d_i} \to \cal{H}_i$ such that $\rho_i(z,z') = \Vert \phi_i(z) - \phi_i(z') \Vert^2_{\cal{H}_i}$\,, where $\langle\cdot,\cdot\rangle_{\cal{H}_i}$ is the inner product defined on $\cal{H}_i$ and $\Vert \cdot \Vert_{\cal{H}_i}$ is the norm induced by the inner product (see Proposition 3 in Sejdinovic et al.\,(2013) for detailed discussions). Further, if $k_i$ is a distance-induced kernel induced by the metric $\rho_i$, then by Proposition 14 in Sejdinovic et al.\,(2013), $\cal{H}_i$ is the RKHS with the reproducing kernel $k_i$ and $\phi_i(z)$ is essentially the canonical feature map for $\cal{H}_i$, viz. $\phi_i(z) : z \mapsto k_i(\cdot,z)$. It is easy to see that
\begin{align*}
\tau_X^2\,=&\,\E\,\sum^{p}_{i=1}\|\phi_i(X_{(i)})-\phi_i(X_{(i)}')\|^2_{\cal{H}_i}\,=\,2\,\E\sum^{p}_{i=1}\|\phi_i(X_{(i)})-\E\,\phi_i(X_{(i)})\|^2_{\cal{H}_i},\\
\tau_Y^2\,=&\,\E\,\sum^{p}_{i=1}\|\phi_i(Y_{(i)})-\phi_i(Y_{(i)}')\|^2_{\cal{H}_i}\,=\,2\,\E\sum^{p}_{i=1}\|\phi_i(Y_{(i)})-\E\,\phi_i(Y_{(i)})\|^2_{\cal{H}_i},\\
\tau^2\,=&\,\E\,\sum^{p}_{i=1}\|\phi_i(X_{(i)})-\phi_i(Y_{(i)})\|^2_{\cal{H}_i}\,=\,\tau_X^2/2+\tau_Y^2/2+\zeta^2,
\end{align*}
where $\zeta^2=\sum^{p}_{i=1}\|\E\,\phi(X_{(i)})-\E\,\phi(Y_{(i)})\|^2_{\cal{H}_i}$.
Thus $2\tau - \tau_X - \tau_Y = 0$\, is equivalent to
\begin{align*}
4(\tau_X^2/2+\tau_Y^2/2+\zeta^2)=(\tau_X+\tau_Y)^2=\tau_X^2+\tau_Y^2+2\tau_X\tau_Y.
\end{align*}
which implies that
$$4\zeta^2+(\tau_X-\tau_Y)^2=0.$$
Therefore, $2\tau - \tau_X - \tau_Y = 0$\, holds if and only if (1) $\zeta=0$, i.e., $\E\,\phi_i(X_{(i)})=\E\,\phi_i(Y_{(i)})$ for all $1\leq i\leq p$, and, (2) $\tau_X=\tau_Y$, i.e.,
$$\E\sum^{p}_{i=1}\|\phi_i(X_{(i)})-\E\,\phi_i(X_{(i)})\|_{\cal{H}_i}^2\,=\,\E\sum^{p}_{i=1}\|\phi_i(Y_{(i)})-\E\,\phi_i(Y_{(i)})\|_{\cal{H}_i}^2.$$ Now if $X \sim P$ and $Y \sim Q$, then note that $$\E\,\phi_i(X_{(i)}) \,=\, \dis\int_{\bb{R}^{d_i}} k_i(\cdot, z)\, dP_i (z) \,=\, \Pi_i(P_i)\;\;\; \textrm{and} \;\;\; \E\,\phi_i(Y_{(i)}) \,=\, \dis\int_{\bb{R}^{d_i}} k_i(\cdot, z)\, dQ_i (z) \,=\, \Pi_i(Q_i)\,, $$ where $\Pi_i$ is the mean embedding function (associated with the distance induced kernel $k_i$) defined in Section \ref{ed sec}, $P_i$ and $Q_i$ are the distributions of $X_{(i)}$ and $Y_{(i)}$, respectively. As $\rho_i$ is a metric of strong negative type on $\bb{R}^{d_i}$, the induced kernel $k_i$ is characteristic to $\cal{M}_1(\bb{R}^{d_i})$ and hence the mean embedding function $\Pi_i$ is injective. Therefore condition (1) above implies $X_{(i)} \overset{d}{=} Y_{(i)}$.
\end{proof}

Now we introduce some notation before presenting the proof of Theorem \ref{th KED HDLSS}. The key of our analysis is to study the variance of the leading term of $\cal{E}_{n,m}(X,Y)$ in the HDLSS setup, propose the variance estimator and study the asymptotic behavior of the variance estimator. It will be shown later (in the proof of Theorem \ref{th KED HDLSS}) that the leading term in the Taylor's expansion of $\cal{E}_{n,m}(X,Y) - (2\tau - \tau_X - \tau_Y)$ can be written as $L_1 + L_2$, where
\begin{align}\label{L decomp.1}
\begin{split}
L_1\;&:=\;\frac{1}{nm\tau}\sum_{k=1}^{n}\sum_{l=1}^{m}\sum_{i=1}^{p}d_{kl}(i)-\frac{1}{n(n-1)\tau_X}\sum_{k< l}\sum_{i=1}^{p}d^X_{kl}(i)
-\frac{1}{m(m-1)\tau_Y}\sum_{k< l}\sum_{i=1}^{p}d^Y_{kl}(i)\\
\;&:=\,L_1^1-L_1^2-L_1^3 \;,
\end{split}
\end{align}
where $L_1^i$'s are defined accordingly and
\begin{align}\label{L decomp.2}
\begin{split}
L_2 \;:=&\; \frac{1}{nm\tau}\sum_{k=1}^{n}\sum_{l=1}^{m}\sum_{i=1}^{p}\Big(\E\,[\rho_i(X_{k(i)},Y_{l(i)})|X_{k(i)}] + [\rho_i(X_{k(i)},Y_{l(i)})|Y_{l(i)}] - 2\,\E\,\rho_i(X_{k(i)},Y_{l(i)})\Big) \\ & \;-\frac{1}{n(n-1)\tau_X}\sum_{k< l}\sum_{i=1}^{p}\Big(\E\,[\rho_i(X_{k(i)},X_{l(i)})|X_{k(i)}] + [\rho_i(X_{k(i)},X_{l(i)})|X_{l(i)}] - 2\,\E\,\rho_i(X_{k(i)},X_{l(i)})\Big)
\\&-\frac{1}{m(m-1)\tau_Y}\sum_{k< l}\sum_{i=1}^{p}\Big(\E\,[\rho_i(Y_{k(i)},Y_{l(i)})|Y_{k(i)}] + [\rho_i(Y_{k(i)},Y_{l(i)})|Y_{l(i)}] - 2\,\E\,\rho_i(Y_{k(i)},Y_{l(i)})\Big)\,.
\end{split}
\end{align}
By the double-centering properties, it is easy to see that $L_1^i$ for $1\leq i\leq 3$ are uncorrelated. Define
\begin{align}\label{V def supp}
\begin{split}
V :=& \frac{1}{n m \tau^2}\sum_{i,i'=1}^{p}\E\,[d_{kl}(i)\,d_{kl}(i')] \;+\;\frac{1}{2n(n-1)\tau^2_X}\sum_{i,i'=1}^{p}\E\,[d_{kl}^X(i)\,d_{kl}^X(i')]\\ &  +\;\frac{1}{2m(m-1)\tau^2_Y}\sum_{i,i'=1}^{p}\E\,[d_{kl}^Y(i)\,d_{kl}^Y(i')]\\ :=& \;V_1\; +\; V_2 \;+\; V_3,
\end{split}
\end{align}
where $V_i$'s are defined accordingly.
Further let
\begin{align}\label{V tilde def supp 1}
\widetilde{V_1} \;:=\; nm V_1\;,\;  \widetilde{V_2} \;:=\; 2n(n-1) V_2 \;,\; \widetilde{V_3} \;:=\; 2m(m-1) V_3\,.\end{align} It can be verified that
\begin{align*}
\E\,[d^X_{kl}(i)\,d^X_{kl}(i')] \;=\; D^2_{\rho_i, \rho_{i'}} (X_{(i)}, X_{(i')})\,.
\end{align*}
Thus we have
\begin{align} \label{V tilde def supp 2}
\widetilde{V_2}\;=\; \frac{1}{\tau^2_X} \dis \sum_{i,i'=1}^p D^2_{\rho_i, \rho_{i'}} (X_{(i)}, X_{(i')}) \;\;\;\; \textrm{and} \;\;\;\; \widetilde{V_3}\;=\; \frac{1}{\tau^2_Y} \dis \sum_{i,i'=1}^p D^2_{\rho_i, \rho_{i'}} (Y_{(i)}, Y_{(i')})\,.
\end{align}

We study the variances of $L_1^i$ for $1\leq i\leq 3$ and propose some suitable estimators. The variance for $L_1^2$ is given by
\begin{align*}
var(L_1^2) \;=&\;\frac{1}{n^2(n-1)^2\tau_X^2}\sum_{i,i'=1}^{p}\sum_{k<l}\E\,[d^X_{kl}(i)\,d^X_{kl}(i')]
\;=\;V_2\,.
\end{align*}
Clearly
$$\frac{n(n-1)V_2}{2}\;=\;\frac{1}{4\tau^2_X}\dis\sum_{i,i'=1}^{p}D^2_{\rho_i, \rho_j}(X_{(i)},X_{(i')})\,.$$
From Theorem \ref{ACdCov taylor thm} \,in Section \ref{sec:ACdcov-HDLSS}, we know that for fixed $n$ and growing $p$, $\widetilde{\cal{D}_n^2}(X,X)$ is asymptotically equivalent to $\frac{1}{4\tau^2_X} \sum_{i,i'=1}^{p}\widetilde{D^2_n}_{\,;\,\rho_i, \rho_j}(X_{(i)},X_{(i')})$. Therefore an estimator of $\widetilde{V_2}$ is given by $4\,\widetilde{\cal{D}_n^2}(X,X)$. 
Note that the computational cost of $\widetilde{\cal{D}_n^2}(X,X)$ is linear in $p$ while direct calculation of its leading term $\frac{1}{4\tau^2_X} \sum_{i,i'=1}^{p}\widetilde{D^2_n}_{\,;\,\rho_i, \rho_j}(X_{(i)},X_{(i')})$ requires computation in the quadratic order of $p$. Similarly it can be shown that the variance of $L_1^3$ is $V_3$ and $\widetilde{V_3}$ can be estimated by $4\,\widetilde{\cal{D}_m^2}(Y,Y)$. Likewise some easy calculations show that the variance of $L_1^1$ is $V_1$. Define
\begin{equation}\label{eq rho hat}
\begin{split}
\hat{\rho_i}(X_{k(i)},Y_{l(i)})\;:=&\;\rho_i(X_{k(i)},Y_{l(i)})\,-\,\frac{1}{n}\sum^{n}_{a=1}\rho_i(X_{a(i)},Y_{l(i)})\,-\,\frac{1}{m}\sum^{m}_{b=1}\rho_i(X_{k(i)},Y_{b(i)})
\\&+\,\frac{1}{nm}\sum_{a=1}^{n}\sum^{m}_{b=1}\rho_i(X_{a(i)},Y_{b(i)})\;,
\end{split}
\end{equation}
and \begin{align}\label{expr for R hat}
\hat{R}(X_k, Y_l)\;:=\;R(X_{k},Y_{l})-\frac{1}{n}\sum^{n}_{a=1}R(X_{a},Y_{l})-\frac{1}{m}\sum^{m}_{b=1}R(X_{k},Y_{b})+\frac{1}{nm}\sum_{a=1}^{n}\sum^{m}_{b=1}R(X_{a},Y_{b})\,.
\end{align} It can be verified that
$$\hat{\rho_i}(X_{k(i)},Y_{l(i)})\; =\;d_{kl}(i)\,-\,\frac{1}{n}\sum^{n}_{a=1}d_{al}(i)\,-\,\frac{1}{m}\sum^{m}_{b=1}d_{kb}(i)\,+\,\frac{1}{nm}\sum_{a=1}^{n}\sum^{m}_{b=1}d_{ab}(i).$$ Observe that
\begin{align}\label{eq 0.8}
\E\,[\hat{\rho_i}(X_{k(i)},Y_{l(i)})\rho_{i'}(X_{k(i')},Y_{l(i')})]\;=\;(1-1/n)(1-1/m)\,\E\,[d_{kl}(i)\,d_{kl}(i')]\,.
\end{align}
Let $\hat{\bf A}_i=(\hat{\rho_i}(X_{k(i)},Y_{l(i)}))_{k,l},\;{\bf A}_i=(\rho_i(X_{k(i)},Y_{l(i)}))_{k,l}\,\in \mathbb{R}^{n\times m}$. Note that
\begin{align}\label{eq 0.9}
\begin{split}
&\;\frac{1}{(n-1)(m-1)}\,\E\,\sum_{k=1}^{n}\sum_{l=1}^{m}\hat{\rho_i}(X_{k(i)},Y_{l(i)})\hat{\rho_i}(X_{k(i')},Y_{l(i')})\\
 =&\;\frac{1}{(n-1)(m-1)}\,\E\,\text{tr}(\hat{{\bf A}}_i\hat{{\bf A}}_{i'}^\top)
\\=&\;\frac{1}{(n-1)(m-1)}\,\E\,\text{tr}(\hat{{\bf A}}_i {\bf A}_{i'}^\top)
\\=&\;\frac{1}{(n-1)(m-1)}\,\E\,\sum_{k=1}^{n}\sum_{l=1}^{m}\,\rho_i(X_{k(i')},Y_{l(i')})\,\hat{\rho_i}(X_{k(i)},Y_{l(i)})
\\=&\;\E\,[d_{kl}(i)\,d_{kl}(i')],
\end{split}
\end{align}
which suggests that
\begin{align*}
\breve{V}_1=\frac{1}{nm\tau^2}\sum_{i,i'=1}^{p}\frac{1}{(n-1)(m-1)}\sum_{k=1}^{n}\sum_{l=1}^{m}\hat{\rho_i}(X_{k(i)},Y_{l(i)})\,\hat{\rho_i}(X_{k(i')},Y_{l(i')})
\end{align*}
is an unbiased estimator for $V_1.$ However, the computational cost for $\breve{V}_1$ is linear in $p^2$ which is prohibitive for large $p.$
We aim to find a joint metric whose computational cost is linear in $p$ whose leading term is proportional to $\breve{V}_1.$ It can be verified that $cdCov^2_{n,m}(X,Y)$ is asymptotically equivalent to
\begin{align*}
\frac{1}{4\tau^2}\sum_{i,i'=1}^{p}\frac{1}{(n-1)(m-1)}\sum_{k=1}^{n}\sum_{l=1}^{m}\hat{\rho_i}(X_{k(i)},Y_{l(i)})\hat{\rho_i}(X_{k(i')},Y_{l(i')})\;.
\end{align*}
This can be seen from the observation that
\begin{align}
\begin{split}
4\,cdCov^2_{n,m}(X,Y)  \;&=\; \frac{1}{\tau^2} \dis \sum_{i,i'=1}^p \frac{1}{(n-1)(m-1)} \sum_{k=1}^n \sum_{l=1}^m  \hat{\rho}_i (X_{k(i)}, Y_{l(i)})\,\hat{\rho}_{i'} (X_{k(i')}, Y_{l(i')})   \\ & \qquad +\;  \frac{\tau^2}{(n-1)(m-1)} \dis \sum_{k=1}^n \sum_{l=1}^m \hat{R}^2(X_k, Y_l) \\
&\qquad + \; \frac{1}{(n-1)(m-1)} \dis \sum_{k=1}^n \sum_{l=1}^m \frac{1}{\tau} \dis \sum_{i=1}^p \hat{\rho}_i(X_{k(i)}, Y_{(li)}) \,\, \tau \hat{R}(X_k, Y_l).
\end{split}
\end{align}
Using the H{\"o}lder's inequality as well as the fact that $\tau^2 \,\hat{R}^2(X_k, Y_l)$ is $O_p(\tau^2 a^4_p) = o_p(1)$ under Assumption \ref{ass0.6}.
Therefore, we can estimate $\widetilde{V}_1$ by $4cdCov^2_{n,m}(X,Y)$. Thus the variance of $L_1$ is $V$ which can be estimated by
\begin{align}\label{V hat def}
\begin{split}
\hat{V} \;&:= \;\frac{1}{n m}\, 4\, cdCov_{n,m}^2(X,Y) \;+\;\frac{1}{2n(n-1)}\, 4\,\widetilde{\cal{D}_n^2}(X,X) \;+\; \frac{1}{2m(m-1)}\,4\, \widetilde{\cal{D}_m^2}(Y,Y) \,\\
&:= \;\hat{V}_1\; +\; \hat{V}_2 \;+\; \hat{V}_3\,.
\end{split}
\end{align}

\vspace{0.1in}

\begin{proof}[Proof of Theorem \ref{th KED HDLSS}]
Using Proposition \ref{K taylor : ED}, some algebraic calculations yield
\begin{align*}
&\cal{E}_{nm}(X,Y)-(2\tau-\tau_X-\tau_Y)
\\=&\;\frac{\tau}{nm}\sum_{k=1}^{n}\sum_{l=1}^{m}L(X_k,Y_l)-\frac{\tau_X}{2n(n-1)}\sum_{k\neq l}^n L(X_k,X_l)-\frac{\tau_Y}{2m(m-1)}\sum_{k\neq l}^m L(Y_k,Y_l)\;+\;R_{n,m}
\\=& \;\frac{1}{nm\tau}\sum_{k=1}^{n}\sum_{l=1}^{m}\sum_{i=1}^{p}\big(\rho_i(X_{k(i)},Y_{l(i)})-\E \,\rho_i(X_{k(i)},Y_{l(i)})\big)
\\&-\frac{1}{2n(n-1)\tau_X}\sum_{k\neq l}^n \sum_{i=1}^{p}\big(\rho_i(X_{k(i)},X_{l(i)})-\E \,\rho_i(X_{k(i)},X_{l(i)})\big)
\\&-\frac{1}{2m(m-1)\tau_Y}\sum_{k\neq l}^m \sum_{i=1}^{p} \big(\rho_i(Y_{k(i)},Y_{l(i)})-\E\, \rho_i(Y_{k(i)},Y_{l(i)})\big)\;+ \;R_{n,m},
\end{align*}
where
\begin{align}\label{remainder negl}
R_{n,m} \;=\; \frac{2\tau}{nm}\sum_{k=1}^{n}\sum_{l=1}^{m}R(X_k,Y_l)-\frac{\tau_X}{n(n-1)}\sum_{k\neq l}^n R(X_k,X_l)-\frac{\tau_Y}{m(m-1)}\sum_{k\neq l}^m R(Y_k,Y_l)\; .
\end{align}
By Assumption \ref{ass0.6},  $R_{n,m} = O_p(\tau a^{2}_{p} + \tau_X b^{2}_{p} + \tau_Y c^{2}_{p}) = o_p(1)$ as $p \to \infty$. Denote the leading term above by $L$. We can rewrite $L$ as $L_1+L_2$, where $L_1$ and $L_2$ are defined in equations (\ref{L decomp.1}) and (\ref{L decomp.2}), respectively. Some calculations yield that
\begin{align}\label{L2 expression}
\begin{split}
L_2 \; =& \; \frac{1}{n} \dis \sum_{k=1}^n \left[ \frac{1}{\tau} \sum_{i=1}^p \E\,[\rho_i(X_{k(i)},Y_{(i)})|X_{k(i)}] \; -\; \frac{1}{\tau_X} \sum_{i=1}^p \E\,[\rho_i(X_{k(i)},X_{(i)}')|X_{k(i)}]\,\right]\;-\; (\tau-\tau_X) \\
& \; + \frac{1}{m} \dis \sum_{l=1}^m \left[ \frac{1}{\tau} \sum_{i=1}^p \E\,[\rho_i(X_{(i)},Y_{l(i)})|Y_{l(i)}] \; -\; \frac{1}{\tau_Y} \sum_{i=1}^p \E\,[\rho_i(Y_{l(i)},Y_{(i)}')|Y_{l(i)}]\,\right] \; - \; (\tau - \tau_Y)\\
=& \; \frac{1}{n} \dis \sum_{k=1}^n \E\,\left[\tau L(X_k, Y) - \tau_X L(X_k, X') \,|\, X_k  \right] \;+\; \frac{1}{m} \dis \sum_{l=1}^m \E\,\left[\tau L(X, Y_l) - \tau_X L(Y_l, Y') \,|\, Y_l  \right]\,.
\end{split}
\end{align}
For $(P_X,P_Y)\in\mathcal{P}$, we have $L_2=o_p(1)$.

Under Assumption \ref{ass6}, the asymptotic distribution of $L_1$ as $p \to \infty$ is given by
\begin{align*}
L_1 \overset{d}{\longrightarrow} N \Big(0\,,\,\frac{\sigma^2}{nm} + \frac{\sigma_X^2}{2n(n-1)} + \frac{\sigma_Y^2}{2m(m-1)}\Big).
\end{align*}
Define the vector $d_{\textrm{vec}} := \left(\frac{1}{\tau} \sum_{i=1}^p d_{kl}(i) \right)_{1\leq k \leq n, \, 1\leq l \leq m}$. It can be verified that
\begin{align}\label{dvec expr}
4(n-1)(m-1)\,cdCov^2_{n,m}(X,Y) \;&= \; d_{\textrm{vec}}^{\top} \,A\, d_{\textrm{vec}}
\end{align}
where $A=A_1 + A_2 + A_3 + A_4$\, with $A_1 = I_n \otimes I_m$, $A_2 = - I_n \otimes \frac{1}{m}1_m 1_m^{\top}$, $A_3 = -\frac{1}{n}1_n 1_n^{\top} \otimes I_m$\, and \, $A_4 = \frac{1}{nm}1_{nm} 1_{nm}^{\top}$. Here $\otimes$ denotes the Kronecker product. It is not hard to see that $A^2 = A$ and $\textrm{rank}(A) = (n-1)(m-1)$. Therefore by Assumption \ref{ass6}, we have as $p \to \infty$,
\begin{align*}
4(n-1)(m-1)\,cdCov^2_{n,m}(X,Y)\;\overset{d}{\rightarrow }\;\sigma^2\chi^2_{(n-1)(m-1)}.
\end{align*}
By Theorem \ref{ACdcov:dist_conv}, we have as $p \to \infty$,
\begin{align*}
& 4\,\widetilde{\cal{D}_n^2}(X,X) \;\overset{d}{\rightarrow} \;  \frac{\sigma_X^2}{v_n} \chi^2_{v_n}\;,\;\; \textrm{i.e.}, \;\;  4\,v_n\,\widetilde{\cal{D}_n^2}(X,X) \;\overset{d}{\rightarrow} \;  \sigma_X^2 \, \chi^2_{v_n}\,,
\end{align*}
and similarly
\begin{align*}
4\,v_m\,\widetilde{\cal{D}_m^2}(Y,Y) \;\overset{d}{\rightarrow} \;  \sigma_Y^2 \, \chi^2_{v_m}\,.
\end{align*}
By Assumption \ref{ass6}, $\chi^2_{(n-1)(m-1)},\chi^2_{v_n}$ and $\chi^2_{v_m}$ are mutually independent. The proof can be completed by combining all the arguments above and using the continuous mapping theorem.
\end{proof}

\begin{proof}[Proof of Proposition \ref{power}]
Note that as $n, m \to \infty$,
\begin{align*}
\E\,[(M-m_0)^2] \; &= \; \frac{2(n-1)(m-1)\sigma^4 \, +  2v_n\sigma_X^4 \, +  2v_m\sigma_Y^4}{\left\{\,(n-1)(m-1) \,+ \,v_n \,+\, v_m \,\right\}^2}\; = \; o(1),
\end{align*}
where $m_0=\E[M]$. Therefore by Chebyshev's inequality, $M - m_0 = o_p(1)$ as $n, m \to \infty$. As a consequence, we have $M \overset{p}{\longrightarrow} m_0^*$\, as $n, m \to \infty$. Observing that $\Phi$ is a bounded function, the rest follows from Lebesgue's Dominated Convergence Theorem.
\end{proof}

\vspace{0.1in}
Under $H_0$, without any loss of generality define $U_1=X_1,\dots,U_n=X_n,U_{n +1}:= Y_1,\dots,
U_{n + m}:=Y_{m}$. Further define
\begin{align}\label{phi def supp}
\phi_{i_1 i_2}:=\phi(U_{i_1},U_{i_2})=\begin{cases} -\frac{1}{n(n -1)}\;H(U_{i_1}, U_{i_2}) \;\; & \; \textrm{if}\;\, i_1, i_2 \in \{1, \dots \,,n\} \,,\\ \qquad \frac{1}{n m}\;H(U_{i_1}, U_{i_2}) \;\; & \; \textrm{if}\;\, i_1 \in \{1, \dots \,,n\},i_2 \in \{n +1, \dots \,,n+m\}\,, \\ -\frac{1}{m(m -1)}\,H(U_{i_1}, U_{i_2}) \;\; & \; \textrm{if}\;\, i_1, i_2 \in \{n +1, \dots \,,n+m\}\,. \end{cases}
\end{align}
It can be verified that $\cov(\phi_{i_1 i_2},\, \phi_{i_1' i_2'})=0$\, if the cardinality of the set $\{i_1, i_2\} \cap \{i_1', i_2'\}$ is less than $2$. Define $$\breve{T}_{n,m} \;=\;\frac{\cal{E}_{n,m}(X,Y)}{\sqrt{V}}.$$

\begin{lemma}\label{lemma_supp_1}
Under $H_0$ and Assumptions \ref{ass0.5}, \ref{assED_HDMSS} and \ref{assED_HDMSS4}, as $n, m$ and $p \to \infty$, we have $$\breve{T}_{n, m} \; \overset{d}{\longrightarrow} \; N(0,1)\,.$$
\end{lemma}

\begin{proof}[Proof of Lemma \ref{lemma_supp_1}]
Set $N=n+m.$ Define\, $V_{Nj} :=  \sum_{i=1}^{j-1} \phi_{ij}$ for $2\leq j\leq N$\,, $S_{N r} :=  \sum_{j=2}^r V_{N j} =  \sum_{j=2}^r \sum_{i=1}^{j-1} \phi_{ij}$ for $2\leq r \leq N$, and \,$\mathcal{F}_{N,r} := \sigma(X_1, \dots\,, X_r)$. Then the leading term of $\cal{E}_{n m}(X,Y)$, viz., $L_1$ (see equation (\ref{L decomp.1})) can be expressed as $$L_1 \;=\; S_{NN} = \dis \sum_{j=2}^N V_{Nj} \;=\; \sum_{j=2}^N \sum_{i=1}^{j-1} \phi_{ij}\;=\;\dis \sum_{1\leq i_1<i_2\leq n} \phi_{i_1 i_2} \;+ \; \sum_{i_1=1}^{n} \sum_{i_2=n+1}^{N} \phi_{i_1 i_2} \;+\; \sum_{n+1\leq i_1<i_2\leq N} \phi_{i_1 i_2}\,.$$ By Corollary 3.1  of Hall and Heyde\,(1980), it suffices to show the following : \begin{enumerate}
\item For each $N$, $\{S_{Nr}, \mathcal{F}_{N,r}\}_{r=1}^N$ \;is a sequence of zero mean and square integrable martingales,
\vspace{-0.08in}
\item $\frac{1}{V} \dis \sum_{j=2}^N \E\,\left[V_{Nj}^2 \,|\, \mathcal{F}_{N,j-1}\right] \; \overset{P}{\longrightarrow} \; 1\,$,
\vspace{-0.1in}
\item $\frac{1}{V} \dis \sum_{j=2}^N \E\,\left[V_{Nj}^2 \,\mathbbm{1}(|V_{Nj}| > \epsilon \sqrt{V}) \,|\, \mathcal{F}_{N,j-1}\,\right] \; \overset{P}{\longrightarrow} \; 0\,,\;\;\; \forall \; \epsilon > 0$.
\end{enumerate}
\vspace{-0.2in}
To show (1), it is easy to see that $S_{Nr}$ is square integrable, $\E(S_{Nr}) = \dis \sum_{j=2}^r \sum_{i=1}^{j-1} \E (\phi_{ij}) = 0$, and, $\mathcal{F}_{N,1} \subseteq \mathcal{F}_{N,2} \subseteq\, \dots \, \subseteq\mathcal{F}_{N,N}$. We only need to show $\E(S_{Nq}\,|\,\mathcal{F}_{N,r}) = S_{Nr}$ for $q>r$. Now \;$\E(S_{Nq}\,|\,\mathcal{F}_{N,r}) = \dis \sum_{j=2}^q \sum_{i=1}^{j-1} \E (\phi_{ij}\,|\,\mathcal{F}_{N,r})$. If $j \leq r <q$ and $i<j$, then $\E(\phi_{ij}\,|\,\mathcal{F}_{N,r}) = \phi_{ij}$. If $r<j\leq q$, then :
\begin{enumerate}
\item[(i)] if \;$r<i<j\leq q$, then $\E(\phi_{ij}\,|\,\mathcal{F}_{N,r}) = \E(\phi_{ij}) = 0$,
\item[(ii)] if \;$i\leq r<j\leq q$, then $\E(\phi_{ij}\,|\,\mathcal{F}_{N,r}) = 0$ (due to $\mathcal{U}$-centering).
\end{enumerate}
Therefore $\E(S_{Nq}\,|\,\mathcal{F}_{N,r}) = S_{Nr}$ for $q>r$. This completes the proof of (1).\\

To show (2), define \,$L_j(i,k) := \E\,[\phi_{i j}\, \phi_{k j}\,|\,\mathcal{F}_{N,j-1}]$\, for $i, k < j \leq N$,\, and $$\eta_N := \dis \sum_{j=2}^N \E\,\left[V_{Nj}^2\,|\,\mathcal{F}_{N,j-1}\,\right] = \sum_{j=2}^N \sum_{i,k=1}^{j-1} \E[\phi_{i j}\, \phi_{k j}\,|\,\mathcal{F}_{N,j-1}] = \sum_{j=2}^N \sum_{i,k=1}^{j-1} L_j(i,k)\,.$$ Note that $\E\, [L_j(i,k)] = 0$ \,for $i \neq k$. Clearly
\begin{align}\label{eq 1}
\E[\eta_N] \;&=\; \dis\sum_{j=2}^N \E[V_{Nj}^2] \;= \;\dis\sum_{j=2}^N \sum_{i,k=1}^{j-1} \E[\phi_{i j}\, \phi_{k j}] \;=\; \sum_{j=2}^N \sum_{i=1}^{j-1} \E[\phi^2_{i j}] = V\,.
\end{align}
By virtue of Chebyshev's inequality, it will suffice to show $\text{var}(\frac{\eta_N}{V}) = o(1)$. Note that
\begin{align}\label{exp L}
\begin{split}
& \qquad \E\, [L_j(i,k) \, L_{j'}(i',k')] \\&= \begin{cases}
 \E\,\left[\phi^2(U_i,U_j)\phi^2(U_i,U_{j'}')\right] & \; \; i=k=i'=k'\,,\\
\E\,\left[\phi(U_i,U_j) \phi(U_k,U_j) \phi(U_i,U_{j'}') \phi(U_{k},U'_{j'})\right]  & \; \; i=i'\neq k=k' \;\; \textrm{or} \;\; i=k'\neq k=i'\,,\\
 \E\,\left[\phi^2(U_i,U_j)\right] \E\,\left[U^2(U_{i'},U_{j'})\right]  & \; \; i=k \neq i'=k'\,.
\end{cases}
\end{split}
\end{align}
In view of equation (\ref{phi def supp}), it can be verified that the above expression for $\E\, L_j(i,k) \, L_{j'}(i',k')$ holds true for $j=j'$ as well. Therefore
\begin{align*}
\text{var}\,(\eta_N^2) \;&=\; \dis \sum_{j,j'=2}^N \sum_{i,k=1}^{j-1} \sum_{i',k'=1}^{j'-1} \text{cov}\,(L_j(i,k) \,,L_{j'}(i',k'))\\
&=\; \dis \sum_{j=j'} \Bigg\{\sum_{i=1}^{j-1} \text{cov}\,\left(\phi^2(U_i,U_j),\phi^2(U_i,U_{j}')\right) \, +\, 2\sum_{i \neq k}^{j-1}\E\,\left[\phi(U_i,U_j) \phi(U_k,U_j) \phi(U_i,U_{j}') \phi(U_{k},U'_{j})\right]\Bigg\}\\
& \; + \; 2\sum_{2 \leq j < j' \leq N} \Bigg\{\sum_{i=1}^{j-1} \text{cov}\,\left(\phi^2(U_i,U_j),\phi^2(U_i,U_{j'}')\right) \\&\,+\, 2\sum_{i \neq k}^{j-1}\E\,\left[\phi(U_i,U_j) \phi(U_k,U_j) \phi(U_i,U_{j'}') \phi(U_{k},U'_{j'})\right]\Bigg\}\,.
\end{align*}
Under Assumption \ref{ass0.5} and $H_0$, it can be verified that
\begin{align}\label{var L 2nd}
\begin{split}
\text{var}(\eta_N)=O\Big(\frac{1}{N^5}\, \E\,\left[H^2(X, X'') H^2(X', X'') \right] +\frac{1}{N^4}\E\left[ H(X, X'') H(X', X'') H(X, X''') H(X', X''') \right] \Big),
\end{split}
\end{align}
and
\begin{align}\label{eqn V order}
V^2\;& \asymp \;  \,\frac{1}{N^4}\, \left(\E\, \left[H^2(X,X')\right]\right)^2  \,.
\end{align}
Therefore under Assumption \ref{assED_HDMSS} and $H_0$, we have $$\text{var}\left(\frac{\eta_N}{V}\right)=o(1),$$
which completes the proof of (2). To show (3), note that it suffices to show $$\frac{1}{V^2} \dis \sum_{j=2}^N \E \left[\,V^4_{Nj} \, | \, \mathcal{F}_{N,j-1}\,\right] \; \overset{P}{\longrightarrow} \; 0\;.$$ Observe that
\begin{align*}
\dis \sum_{j=2}^N \E \left[V^4_{Nj} \,\right] \; &= \; \sum_{j=2}^N \E \left(\, \sum_{i=1}^{j-1} \phi_{ij}\,\right)^4 \;
\\&= \; \sum_{j=2}^N \sum_{i=1}^{j-1} \E[\phi^4(U_i,U_j)] \; + \; 3 \,\sum_{j=2}^{N} \sum_{i_1 \neq i_2}^{j-1} \E[\phi^2(U_{i_1},U_j) \,\phi^2(U_{j_2},U_j)] \,.
\end{align*}
Under Assumption \ref{ass0.5}, we have
\begin{align*}
\dis \sum_{j=2}^N \E \left[\,V^4_{Nj} \,\right] \; &= \; O\,\Big(\frac{1}{N^6}\, \E\, \left[H^4(X,X') \right] \;+\;  \frac{1}{N^5}\, \E\,\left[H^2(X, X'') H^2(X', X'') \right] \Big) \,.
\end{align*}
This along with the observation from equation (\ref{var L 2nd}) and Assumption \ref{assED_HDMSS} complete the proof of (3).

Finally to see that $\frac{R_{n,m}}{\sqrt{V}} = o_p(1)$, note that from equation (\ref{remainder negl}) we can derive using power mean inequality that $\E\,R_{n,m}^2 \leq C\, \tau^2 \,\E\,\left[ R^2(X,X')\right]$ for some positive constant $C$. Using this, equation (\ref{eqn V order}), Chebyshev's inequality and H{\"o}lder's inequality, we have for any $\epsilon >0$
\begin{align}\label{rem/v neg}
\begin{split}
P\left(\Big\vert\frac{R_{n,m}}{\sqrt{V}}\Big\vert > \epsilon \right)\; & \leq \; \frac{\E\,R_{n,m}^2}{\epsilon^2 \; V} \;\leq \; C' \, \frac{N^2 \,\tau^2\, \E\,\left[ R^2(X,X')\right]}{\epsilon^2 \; \E\,\left[ H^2(X,X')\right]} \;\leq \; \frac{C'}{\epsilon^2} \, \left(\frac{N^4 \,\tau^4\, \E\,\left[ R^4(X,X')\right]}{\left( \E\,\left[ H^2(X,X')\right] \right)^2}  \right)^{1/2}\,,
\end{split}
\end{align}
for some positive constant $C'$. From this and Assumptions \ref{ass0.5} and \ref{assED_HDMSS4}, we get $\frac{R_{n,m}}{\sqrt{V}} = o_p(1)$, as $N \asymp n$. This completes the proof of the lemma.
\end{proof}

\begin{lemma}\label{lemma_supp_2}
Under $H_0$ and Assumptions \ref{ass0.5} and \ref{assED_HDMSS4}, as $n, m$ and $p \to \infty$, we have $$\frac{\left|\E\,[\hat{V}_i] - V_i\right|}{V_i} = o(1)\;\; , \;\; 1 \leq i \leq 3\,,$$ where \,$V_i$ and $\hat{V}_i$, $1\leq i \leq 3$ are defined in equations (\ref{V def supp}) and (\ref{V hat def}), respectively in the supplementary material.
\end{lemma}

\begin{proof}[Proof of Lemma \ref{lemma_supp_2}]
We first deal with $\hat{V}_2$. Note that $$\widetilde{\cal{D}^2_n} (X,X) \,=\, \frac{1}{n(n-3)} \dis \sum_{k\neq l} \left(\widetilde{D}^X_{kl}\right)^2\,,$$ where
\begin{align}\label{expr to be used later}
\begin{split}
\widetilde{D}^X_{kl} \;=&\; K(X_k, X_l) \;-\; \frac{1}{n-2}\dis \sum_{b=1}^n K(X_k, X_{b}) \;-\; \frac{1}{n-2}\dis \sum_{a =1}^n K( X_{a}, X_l)
\\& \;+\; \frac{1}{(n-1)(n-2)}\dis \sum_{a,b =1}^n K(X_{a}, X_{b})\\ \; \;&= \;\; \frac{1}{2\tau} \dis \sum_{i=1}^p \wrho_i(X_{k(i)}, X_{l(i)}) \;+\; \tau \widetilde{R}(X_k, X_l)\,,
\end{split}
\end{align}
using Proposition \ref{K taylor : ED}. As a consequence, we can write \begin{align}\label{expr for D_n(X,X)}
\begin{split}
\widetilde{\cal{D}^2_n} (X,X) \;&=\; \frac{1}{4\tau^2} \dis \sum_{i,i'=1}^p \widetilde{D^2_n}_{\,;\,\rho_i, \rho_{i'}} (X_{(i)}, X_{(i')}) \;+\; \frac{ \tau^2}{n(n-3)} \dis \sum_{k\neq l} \widetilde{R}^2(X_k, X_l) \\
&\qquad + \; \frac{1}{n(n-3)} \dis \sum_{k\neq l} \frac{1}{\tau} \dis \sum_{i=1}^p \wrho_i(X_{k(i)}, X_{(li)}) \,\, \tau \widetilde{R}(X_k, X_l) \,.
\end{split}
\end{align}
Note that following Step 3 in Section 1.6 in the supplementary material of Zhang et al.\,(2018), we can write
\begin{align*}
\widetilde{R}(X_k, X_l) \;&=\; \frac{n-3}{n-1} \bar{R}(X_k, X_l) \,-\, \frac{n-3}{(n-1)(n-2)} \dis \sum_{b\notin \{k,l\}} \bar{R}(X_k, X_{b}) \,-\, \frac{n-3}{(n-1)(n-2)} \dis \sum_{a \notin \{k,l\}} \bar{R}(X_{a}, X_l) \,\\ & \;\;+\, \frac{1}{(n-1)(n-2)}\dis \sum_{a,b \notin \{k,l\}} \bar{R}(X_{a}, X_{b})\,,
\end{align*}
where $\bar{R}(X,X')=R(X,X')-E[R(X,X')|X]-E[R(X,X')|X']+E[R(X,X')]$. Using the power mean inequality, it can be verified that $\E\,[\widetilde{R}^2(X_k, X_l)]\leq C\, \E\,[\bar{R}^2(X_k, X_l)]$ for some positive constant $C$. Using this and the H{\"o}lder's inequality, the expectation of the third term in the summation in equation (\ref{expr for D_n(X,X)}) can be bounded as follows
\begin{align*}
&\left|\E\, \left[ \frac{1}{n(n-3)} \dis \sum_{k\neq l} \frac{1}{\tau} \dis \sum_{i=1}^p \wrho_i(X_{k(i)}, X_{l(i)}) \,\, \tau \widetilde{R}(X_k, X_l) \right]\right| \\
\leq \; &\frac{1}{n(n-3)} \dis \sum_{k\neq l} \left( \E\, \left[\left( \frac{1}{\tau} \dis \sum_{i=1}^p \wrho_i(X_{k(i)}, X_{l(i)}) \right)^2\right] \, \tau^2\, \E\, \left[\bar{R}^2(X_k, X_l)\right]\,\right)^{1/2}\\
\leq \; & C' \, \left( \left(\frac{1}{\tau^2} \dis \sum_{i,i'=1}^p D^2_{\rho_i, \rho_{i'}} (X_{(i)}, X_{(i')}) \right) \,\tau^2\, \E\, \left[\bar{R}^2(X, X')\right]\,  \right)^{1/2}\,
\end{align*}
for some positive constant $C'$.
Combining all the above, we get
\begin{align*}
\vert \E\,(\hat{V}_2) - V_2 \vert \leq& \frac{C_1}{n(n-1)}\, \tau^2 \,\E\, \bar{R}^2(X, X')
\\&\,+\,  \frac{C_2}{n(n-1)}\, \left( \left(\frac{1}{\tau^2} \dis \sum_{i,i'=1}^p D^2_{\rho_i, \rho_{i'}} (X_{(i)}, X_{(i')}) \right) \tau^2\, \E \left[ \bar{R}^2(X, X')\right]  \right)^{1/2},
\end{align*}
for some positive constants $C_1$ and $C_2$. As $V_2=\frac{1}{2n(n-1)}E[H^2(X,X')]$,
$$\frac{\left|\E[\hat{V}_2] - V_2\right|}{V_2} = o(1) \quad \text{is satisfied if}\quad
\frac{\tau^2\, \E \left[\bar{R}^2(X,X')\right]}{\E [H^2(X,X')]} \;=\; o(1)\,.$$ Using power mean inequality and Jensen's inequality, it is not hard to verify that $\E\left[\bar{R}^4(X, X')\right] = O\left(\E\, \left[R^4(X, X')\right]\right)$. Using this and H{\"o}lder's inequality, we have $$\frac{\tau^2\, \E \left[\bar{R}^2(X,X')\right]}{\E [H^2(X,X')]} \;= \; O\left(\left(\frac{\tau^4\, \E\, [R^4(X,X')]}{\left( \E\,[ H^2(X,X')] \right)^2}  \right)^{1/2}\right)\,.$$ Clearly Assumption \ref{assED_HDMSS4} implies \,$\frac{\tau^4\, \E\, [R^4(X,X')]}{\left( \E\,[ H^2(X,X')] \right)^2}   = o(1)$, which in turn implies $$\frac{\tau^2\, \E \left[\bar{R}^2(X,X')\right]}{\E [H^2(X,X')]} \;=\; o(1)\,.$$ Similar expressions can be derived for $\hat{V}_3$ as well. For the term involving $\hat{V}_1$, in the similar fashion, we can write
\begin{align}\label{expr for first term}
\begin{split}
\E\, \left[4\, cdCov^2_{n,m}(X,Y) \right] \;&=\; \frac{1}{\tau^2} \dis \sum_{i,i'=1}^p \frac{1}{(n-1)(m-1)} \sum_{k=1}^n \sum_{l=1}^m \E \left[ \hat{\rho}_i (X_{k(i)}, Y_{l(i)})\,\hat{\rho}_{i'} (X_{k(i')}, Y_{l(i')})  \right] \\ & \qquad +\; \tau^2 \frac{1}{(n-1)(m-1)} \dis \sum_{k=1}^n \sum_{l=1}^m \E\left[\hat{R}^2(X_k, Y_l)\right] \\
&\qquad + \; \frac{1}{(n-1)(m-1)} \dis \sum_{k=1}^n \sum_{l=1}^m \frac{1}{\tau} \dis \sum_{i=1}^p \E\left[\hat{\rho}_i(X_{k(i)}, Y_{(li)}) \,\, \tau \hat{R}(X_k, Y_l)\right] \,,
\end{split}
\end{align}
where the expression for $\hat{R}(X_k, Y_l)$ is given in equation (\ref{expr for R hat}).
Following equation (\ref{eq 0.9}) we can write $$\frac{1}{\tau^2} \dis \sum_{i,i'=1}^p \frac{1}{(n-1)(m-1)} \sum_{k=1}^n \sum_{l=1}^m \E \left[ \hat{\rho}_i (X_{k(i)}, Y_{l(i)})\,\hat{\rho}_{i'} (X_{k(i')}, Y_{l(i')})  \right] \;=\; \E \left[H^2(X,Y)\right] \,.$$ Therefore in view of equations (\ref{V def supp}), (\ref{eq 0.8}) and (\ref{V hat def}), using the power mean inequality we can write
\begin{align*}
\vert \E\,(\hat{V}_1) - V_1 \vert &\leq \frac{C_1'}{nm}\, \tau^2 \,\E\, \bar{R}^2(X, Y) \,+\,  \frac{C_2'}{nm}\, \left( \left(\frac{1}{\tau^2} \dis \sum_{i,i'=1}^p \E \left[d_{kl}(i) d_{kl}(i')  \right] \right) \tau^2\, \E \left[ \bar{R}^2(X, Y)\right]  \right)^{1/2},
\end{align*}
for some positive constants $C_1'$ and $C_2'$. Then under $H_0$ and Assumptions \ref{ass0.5} and \ref{assED_HDMSS4}, we have
\begin{align*}
\frac{\left|\E\,(\hat{V}_1) - V_1\right|}{V_1} = o(1)\,.
\end{align*}
\end{proof}

\begin{lemma}\label{lemma_supp_3}
Under $H_0$ and Assumptions \ref{ass0.5}, \ref{assED_HDMSS} and \ref{assED_HDMSS4}, as $n, m$ and $p \to \infty$, we have $$\frac{\text{var}(\hat{V}_i)}{V_i^2} = o(1) , \;\; 1 \leq i \leq 3\,.$$
\end{lemma}

\begin{proof}[Proof of Lemma \ref{lemma_supp_3}]
Again we deal with $\hat{V}_2$ first. To simplify the notations, denote $A_{ij} = K(X_i, X_j)$ and $\widetilde{A}_{ij} = \widetilde{D}^X_{ij}$\, for \,$1 \leq i \neq j \leq n$. Observe that
\begin{align}\label{var decomp}
\begin{split}
\text{var}\left( \widetilde{\cal{D}^2_n}(X,X) \right) \;&=\; \text{var} \left( \frac{1}{n(n-3)} \dis \sum_{i \neq j} \widetilde{A}^2_{ij} \right)\\
& \asymp \; \frac{1}{n^4} \, \left[\dis \sum_{i<j} \text{var}(\widetilde{A}^2_{ij}) \;+\; \dis \sum_{i<j<j'} \text{cov}(\widetilde{A}^2_{ij}, \widetilde{A}^2_{jj'}) \;+\; \dis \sum_{\substack{i<j, i'<j'\\ \{i,j\} \cap \{i',j'\}= \phi}} \text{cov}(\widetilde{A}^2_{ij}, \widetilde{A}^2_{i'j'})  \right]\,.
\end{split}
\end{align}
As in the proof of Lemma \ref{lemma_supp_2}, we can write
\begin{equation}\label{zhang 2018}
\begin{split}
\widetilde{A}_{ij} \;=&\; \frac{n-3}{n-1} \bar{A}_{ij} \,-\, \frac{n-3}{(n-1)(n-2)} \dis \sum_{l \notin \{i,j\}} \bar{A}_{il} \,-\, \frac{n-3}{(n-1)(n-2)} \dis \sum_{k \notin \{i,j\}} \bar{A}_{kj}
\\&+\, \frac{1}{(n-1)(n-2)} \dis \sum_{k,l \notin \{i,j\}} \bar{A}_{kl}\,,
\end{split}
\end{equation}
where the four summands are uncorrelated with each other. Using the power mean inequality, it can be shown that
\begin{align*}
\E\,(\widetilde{A}^4_{ij}) \;\leq \; C\, \E\,(\bar{A}^4_{ij}) \;=\; C\, \E\,\left[ \bar{K}^4(X,X') \right],
\end{align*}
for some positive constant $C$, where $\bar{K}(X,X')=K(X,X') - E[K(X,X')|X] - E[K(X,X')|X'] + E[K(X,X')]$ (similarly define $\bar{L}(X,X')$). Therefore the first summand in equation (\ref{var decomp}) scaled by $\widetilde{V_2}^2$ is $o(1)$ as $n, p \to \infty$, provided $$ \frac{1}{n^2} \, \frac{\E\,\left[ \bar{K}^4(X,X') \right]}{\widetilde{V_2}^2} \;=\; o(1)\,,$$ where $\widetilde{V_2}$ is defined in equations (\ref{V tilde def supp 1}) and (\ref{V tilde def supp 2}). Note that $$ \bar{K}(X,X') \;=\; \frac{\tau_X}{2}\, \bar{L}(X,X') \;+\; \tau_X\,\bar{R}(X,X')\,.$$ 
Using the power mean inequality we can write $$ \frac{1}{n^2} \, \frac{\E\,\left[ \bar{K}^4(X,X') \right]}{\left(\E\,\left[H^2(X,X')\right]\right)^2} \,\leq \, C_0\, \frac{1}{n^2} \,\frac{\tau^4_X \,\E\,\left[ \bar{L}^4(X,X') \right]}{\left(\E\,\left[H^2(X,X')\right]\right)^2} \,+\, C_0'\,\frac{1}{n^2} \,\frac{\tau^4_X\,\E\,\left[ \bar{R}^4(X,X') \right]}{\left(\E\,\left[H^2(X,X')\right]\right)^2}\, $$ for some positive constants $C_0$ and $C_0'$. It is easy to see that  
\begin{align}\label{eqn used later}
\bar{L}(X_k,X_l) \,&=\, \frac{1}{\tau^2_X}\,\bar{K}^2(X_k,X_l) \,=\, \frac{1}{\tau^2_X}\, \dis \sum_{i=1}^p d^X_{kl}(i) \,=\, \frac{1}{\tau_X}\, H(X_k, X_l)\,.
\end{align}
From equation (\ref{eqn used later}) it is easy to see that the condition $$\frac{1}{n^2} \,\frac{\tau^4_X \,\E\,\left[ \bar{L}^4(X,X') \right]}{\left(\E\,\left[H^2(X,X')\right]\right)^2} = o(1)\quad \text{is equivalent to} \quad \frac{1}{n^2} \,\frac{\E\,\left[ H^4(X,X') \right]}{\left(\E\,\left[H^2(X,X')\right]\right)^2} = o(1).$$
For the third summand in equation (\ref{var decomp}), observe that
\begin{align}\label{eq 0.7}
\begin{split}
\widetilde{A}^2_{ij}=& O(1)\bar{A}^2_{ij} + O\left(\frac{1}{n^2}\right) \dis \sum_{l,l' \notin \{i,j\}} \bar{A}_{il} \bar{A}_{il'} + O\left(\frac{1}{n^2}\right) \dis \sum_{k,k' \notin \{i,j\}} \bar{A}_{kj} \bar{A}_{k'j} + O\left(\frac{1}{n^4}\right) \dis \sum_{k,k',l,l' \notin \{i,j\}} \bar{A}_{kl} \bar{A}_{k'l'} \\
&  +  O\left(\frac{1}{n}\right) \,\bar{A}_{ij} \dis \sum_{l \notin \{i,j\}} \bar{A}_{il} \,+\,  O\left(\frac{1}{n}\right) \,\bar{A}_{ij} \dis \sum_{k \notin \{i,j\}} \bar{A}_{kj} \,+\,  O\left(\frac{1}{n^2}\right) \,\bar{A}_{ij} \dis \sum_{k,l \notin \{i,j\}} \bar{A}_{kl}
\\&+O\left(\frac{1}{n^2}\right) \, \dis \sum_{k,l \notin \{i,j\}} \bar{A}_{il} \bar{A}_{kj} +  O\left(\frac{1}{n^3}\right) \, \dis \sum_{k,l,l' \notin \{i,j\}} \bar{A}_{il} \bar{A}_{kl'} \,+\,  O\left(\frac{1}{n^3}\right) \, \dis \sum_{k,k',l \notin \{i,j\}} \bar{A}_{kl} \bar{A}_{k'j}\,.
\end{split}
\end{align}
Likewise $\widetilde{A}^2_{i'j'}$ admits a similar expression as in equation (\ref{eq 0.7}). We claim that when $\{i,j\} \cap \{i',j'\}= \phi $, the leading term of $\text{cov}(\widetilde{A}^2_{ij}, \widetilde{A}^2_{i'j'})$ is $O \left(\frac{1}{n^2}\, \E\,(\bar{A}^4_{ij})\right)$. To see this first note  that $\bar{A}_{ij}$ is independent of $\bar{A}_{i'j'}$  when $\{i,j\} \cap \{i',j'\}= \phi $. 
Using the double-centering properties, it can be verified that $$\text{cov}\left(\bar{A}^2_{i'j'}, \bar{A}_{ij} \dis \sum_{l \notin \{i,j\}} \bar{A}_{il} \right)\; = \; \text{cov} \left(\bar{A}^2_{i'j'}, \bar{A}_{ij} \dis \sum_{k \notin \{i,j\}} \bar{A}_{kj} \right) \;=\; \text{cov} \left(\bar{A}^2_{i'j'}, \bar{A}_{ij} \dis \sum_{k,l \notin \{i,j\}} \bar{A}_{kl} \right) = 0.$$ 

To compute the quantity \,$\text{cov}\left( \bar{A}^2_{i'j'} \;,\, O\left(\frac{1}{n^2}\right)\dis \sum_{l,l' \notin \{i,j\}} \bar{A}_{il} \bar{A}_{il'} \right)$,
consider the following cases:
\begin{enumerate}
\item[Case 1\,.] When $l=l'=i'$\, or\, $l=l'=j'$\, or \,$l=i', l'=j'$, $\text{cov} \left(\bar{A}^2_{i'j'}, \bar{A}_{il} \bar{A}_{il'}\right)$ boils down to $\text{cov}(\bar{A}^2_{i'j'}, \bar{A}^2_{ii'})$\, or\, $\text{cov}(\bar{A}^2_{i'j'}, \bar{A}^2_{ij'})$ or $\text{cov}(\bar{A}^2_{i'j'}, \bar{A}_{ii'} \bar{A}_{ij'})$.
\item[Case 2\,.] When $l=i, l'\notin \{i,j,i',j'\}$\, or \,$l=j', l'\notin \{i,j,i',j'\}$, $\text{cov} \left(\bar{A}^2_{i'j'}, \bar{A}_{il} \bar{A}_{il'}\right)$ boils down to $\text{cov}(\bar{A}^2_{i'j'}, \bar{A}_{ii'} \bar{A}_{il'})$ or $\text{cov}(\bar{A}^2_{i'j'}, \bar{A}_{ij'} \bar{A}_{il'})$, which can be easily verified to be zero. 
\item[Case 3\,.] When $\{l, l'\} \cap \{i', j'\} = \phi$, $\text{cov} \left(\bar{A}^2_{i'j'}, \bar{A}_{il} \bar{A}_{il'}\right)$ is again zero.
\end{enumerate}
Similar arguments can be made about 
$$\text{cov}\left( \bar{A}^2_{i'j'} \;,\, O\left(\frac{1}{n^2}\right)\dis \sum_{k,k' \notin \{i,j\}} \bar{A}_{kj} \bar{A}_{k'j} \right)~~\text{ and }~~\text{cov}\left( \bar{A}^2_{i'j'} \;,\, O\left(\frac{1}{n^2}\right)\dis \sum_{k,l \notin \{i,j\}} \bar{A}_{il} \bar{A}_{kj}\right).$$ 
With this and using H{\"o}lder's inequality, it can be verified that when $\{i,j\} \cap \{i',j'\}= \phi $, the leading term of $\text{cov}(\widetilde{A}^2_{ij}, \widetilde{A}^2_{i'j'})$ is $O \left(\frac{1}{n^2}\, \E\,(\bar{A}^4_{ij})\right)$. Therefore the third summand in equation (\ref{var decomp}) scaled by $\widetilde{V_2}^2$ can be argued to be $o(1)$ in similar lines of the argument for the first summand in equation (\ref{var decomp}).

For the second summand in equation (\ref{var decomp}), in the similar line we can argue that the leading term of $\text{cov}(\widetilde{A}^2_{ij}, \widetilde{A}^2_{jj'})$ is $$
O \left(\frac{1}{n}\right) \E \left[\bar{A}_{ij}^4 \right] \;+\; O(1) \, \E \left[ \bar{A}_{ij}^2 \bar{A}_{jj'}^2 \right]\,.$$ Therefore the leading term of \,$\frac{1}{n^4}\dis \sum_{i<j<j'}\text{cov}(\widetilde{A}^2_{ij}, \widetilde{A}^2_{jj'})$ is $$
O \left(\frac{1}{n^2}\right) \E \left[\bar{A}_{ij}^4 \right] \;+\; O\left(\frac{1}{n}\right) \, \E \left[ \bar{A}_{ij}^2 \bar{A}_{jj'}^2 \right]\,.$$ For the second term above, using the power mean inequality we can write
\begin{align*}
\frac{1}{n} \, \frac{\E\,\left[ \bar{A}_{ij}^2 \,\bar{A}_{jj'}^2 \right]}{\left(\E\,\left[H^2(X,X')\right]\right)^2} \;&\leq \; C_3 \, \frac{1}{n} \,\frac{\tau^4 \,\E\,\left[ \bar{L}^2(X,X')\,\bar{L}^2(X',X'') \right]}{\left(\E\,\left[H^2(X,X')\right]\right)^2} \,+\, C_3'\frac{1}{n} \,\frac{\tau^4\,\E\,\left[\bar{L}^2(X,X')\, \bar{R}^2(X',X'') \right]}{\left(\E\,\left[H^2(X,X')\right]\right)^2}\,\\
& \qquad + \, C_3''\,\frac{1}{n} \,\frac{\tau^4\,\E\,\left[\bar{R}^2(X,X')\, \bar{R}^2(X',X'') \right]}{\left(\E\,\left[H^2(X,X')\right]\right)^2}\\
&= \, C_3 \, \frac{1}{n} \,\frac{ \,\E\,\left[ H^2(X,X')\,H^2(X',X'') \right]}{\left(\E\,\left[H^2(X,X')\right]\right)^2} \,+\, C_3'\frac{1}{n} \,\frac{\tau^2\,\E\,\left[H^2(X,X')\, \bar{R}^2(X',X'') \right]}{\left(\E\,\left[H^2(X,X')\right]\right)^2}\,\\
& \qquad + \, C_3''\,\frac{1}{n} \,\frac{\tau^4\,\E\,\left[\bar{R}^2(X,X')\, \bar{R}^2(X',X'') \right]}{\left(\E\,\left[H^2(X,X')\right]\right)^2}
\end{align*}
for some positive constants $C_3, C_3'$ and $C_3''$. Using H{\"o}lder's inequality it can be seen that the second summand in equation (\ref{var decomp}) scaled by $\widetilde{V_2}^2$ is $o(1)$ as $n, p \to \infty$ under Assumptions \ref{assED_HDMSS} and \ref{assED_HDMSS4}. This completes the proof that $$\frac{\text{var}(\hat{V}_2)}{V_2^2} = o(1)\,.$$ A similar line of argument and the simple observation that
\begin{align*}
\hat{K}(X_k, Y_l)\;&=\;K(X_{k},Y_{l})-\frac{1}{n}\sum^{n}_{a=1}K(X_{a},Y_{l})-\frac{1}{m}\sum^{m}_{b=1}K(X_{k},Y_{b})+\frac{1}{nm}\sum_{a=1}^{n}\sum^{m}_{b=1}K(X_{a},Y_{b})\,\\
&= \;\bar{K}(X_{k},Y_{l})-\frac{1}{n}\sum^{n}_{a=1} \bar{K}(X_{a},Y_{l})-\frac{1}{m}\sum^{m}_{b=1} \bar{K}(X_{k},Y_{b})+\frac{1}{nm}\sum_{a=1}^{n}\sum^{m}_{b=1}\bar{K}(X_{a},Y_{b})\,
\end{align*}
will show that under Assumptions \ref{ass0.5}, \ref{assED_HDMSS} and \ref{assED_HDMSS4}, $$\frac{\text{var}(\hat{V}_1)}{V_1^2} = o(1)\qquad \textrm{and} \qquad \frac{\text{var}(\hat{V}_3)}{V_3^2} = o(1)\,.$$
\end{proof}

\begin{lemma}\label{lemma_supp_4}
Under $H_0$ and Assumptions \ref{ass0.5}, \ref{assED_HDMSS} and \ref{assED_HDMSS4}, as $n, m$ and $p \to \infty$, we have \,$\hat{V}/V \overset{P}{\to} 1 \,.$
\end{lemma}

\begin{proof}
It is enough to show that $$ \E\,\left[\left(\frac{\hat{V}}{V} - 1\right)^2\right] = o(1)\;, \;\; \textrm{i.e.}\,, \;\; \frac{\text{var}(\hat{V}) + \left(\E\,[\hat{V}] - V\right)^2}{V^2} = o(1)\,.$$It suffices to show the following
\begin{align*}
\frac{\text{var}(\hat{V}_i)}{V_i^2} = o(1) \;\;\;\; \textrm{and} \;\;\;\; \frac{\left(\E\,[\hat{V}_i] - V_i\right)^2}{V_i^2} = o(1), \quad 1 \leq i \leq 3.
\end{align*}
The proof can be completed using Lemmas \ref{lemma_supp_2} and \ref{lemma_supp_3}.
\end{proof}

\vspace{0.1in}

\begin{proof}[Proof of Theorem \ref{th_ED_HDMSS}]

The proof essentially follows from Lemma \ref{lemma_supp_1} and \ref{lemma_supp_4}.

\end{proof}

\begin{proof}[Proof of Proposition \ref{unbiased and ustat form}]
The proof of the first part follows similar lines of the proof of Proposition 1 in Sz\'ekely et al.\,(2014), replacing the Euclidean distance between $X$ and $X'$, viz. $\Vert X-X' \Vert_{\tilde{p}}$\,, by $K(X,X')$. The second part of the proposition has a proof similar to Lemma 2.1 in Yao et al.\,(2018) and Section 1.1 in the Supplement of Yao et al.\,(2018).
\end{proof}

\begin{proof}[Proof of Theorem \ref{ACdCov U-stat prop}]
The first two parts of the theorem immediately follow from Proposition 2.6 and Theorem 2.7 in Lyons\,(2013), respectively and the parallel U-statistics theory (see for example Serfling\,(1980)). The third part follows from the first part and the fact that $\cal{D}$ is non-zero for two dependent random vectors.
\end{proof}

\begin{proof}[Proof of Theorem \ref{D pop taylor}]
Following the definition of $\cal{D}(X,Y)$ and applying Proposition \ref{K taylor : ED}, we can write
\begin{align*}
\frac{1}{\tau_{XY}}\,\cal{D}^2 (X,Y) \; &= \; \E \,\frac{K(X,X')}{\tau_X}\, \frac{K(Y,Y')}{\tau_Y} \, + \, \E \,\frac{K(X,X')}{\tau_X} \, \E \,\frac{K(Y,Y')}{\tau_Y} \, - \, 2\, \E \,\frac{K(X,X')}{\tau_X}\, \frac{K(Y,Y'')}{\tau_Y}\\
&= \;\; \E\, \left(1 + \frac{1}{2} L(X, X') + R(X, X') \right) \, \left(1 + \frac{1}{2} L(Y, Y') + R(Y, Y')  \right)\\
& \qquad + \; \E\, \left(1 + \frac{1}{2} L(X, X') + R(X, X') \right) \, \E\,\left(1 + \frac{1}{2} L(Y, Y') + R(Y, Y')  \right)\\
& \qquad - \; 2\, \E\, \left(1 + \frac{1}{2} L(X, X') + R(X, X') \right) \, \left(1 + \frac{1}{2} L(Y, Y'') + R(Y, Y'')  \right)\,\\
&= \;\; L \; + \; R,
\end{align*}
where
\begin{align*}
L \; &= \; \frac{1}{4} \, \left[ \,\E\, L(X,X') L(Y,Y') \,\, + \,\,\E\,L(X,X') \, \E\,L(Y,Y') \,\, - \,\, 2\, \E\, L(X,X') L(Y,Y'') \, \right],
\end{align*}
and
\begin{align*}
R \; =& \; \E\, \left[ \, \frac{1}{2} L(X,X') R(Y,Y') \,+\, \frac{1}{2} R(X,X') L(Y,Y') \, +\, R(X,X') R(Y,Y') \,\right]\\
&  -\,2\,\E\, \left[ \, \frac{1}{2} L(X,X') R(Y,Y'') \,+\, \frac{1}{2} R(X,X') L(Y,Y'') \, +\, R(X,X') R(Y,Y'') \,\right]
\\&+ \E\, R(X,X') \, \E\,R(Y,Y').
\end{align*}
Some simple calculations yield \begin{align*}
L \; &=\; \frac{1}{4\tau^2_{XY}} \, \left\{ \,\E\, [K^2(X,X') K^2(Y,Y')] \,\, + \,\,\E\,[K^2(X,X')] \, \E\,[K^2(Y,Y')] \,\, - \,\, 2\, \E\, [K^2(X,X') K^2(Y,Y'')] \, \right\}\\
&= \; \frac{1}{4\tau^2_{XY}} \,  \dis\sum_{i=1}^p \sum_{j=1}^q \Big\{\,\E\, [\rho_i(X_{(i)}, X_{(i)}') \, \rho_j(Y_{(j)}, Y_{(j)}')] \, + \,  \E\, [\rho_i(X_{(i)}, X_{(i)}')] \, \E\,[\rho_j(Y_{(j)}, Y_{(j)}')] \\ &\qquad \qquad \qquad \qquad -\,2 \,\E\, [\rho_i(X_{(i)}, X_{(i)}') \, \rho_j(Y_{(j)}, Y_{(j)}'')]  \Big\}\\
&= \; \frac{1}{4\tau^2_{XY}} \,  \dis\sum_{i=1}^p \sum_{j=1}^q D^2_{\rho_i,\,\rho_j}(X_{(i)}, Y_{(j)})\,.
\end{align*}
To observe that the remainder term is negligible, note that under Assumption \ref{ass D pop taylor},
\begin{align*}
\E\,[ L(X, X') R(Y,Y')] \; & \leq \; \left(\, \E\, [L(X, X')^2] \, \E\, [R(Y, Y')^2] \,\right)^{1/2} \; = \; O(a_p' b_q'^2)\;,\\
\E\,[ R(X, X') L(Y,Y')] \; & \leq \; \left(\, \E\, [R(X, X')^2 ]\, \E\, [L(Y, Y')^2] \,\right)^{1/2} \; = \; O(a_p'^2 b_q')\;,\\
\E\,[ R(X, X') R(Y,Y')] \; & \leq \; \left(\, \E\, [R(X, X')^2] \, \E\, [R(Y, Y')^2] \,\right)^{1/2} \; = \; O(a_p'^2 b_q'^2)\;,
\end{align*}
Clearly, $\mathcal{R} = \tau_{XY} R = O(\tau_{XY}\, a_p'^2 b_q' + \tau_{XY}\, a_p' b_q'^2)$.
\end{proof}

\begin{proof}[Proof of Theorem \ref{th MDD}]
The proof is essentially similar to the proof of Theorem \ref{D pop taylor}. Note that using Proposition \ref{K taylor : ED}, we can write
\begin{align*}
\frac{1}{\tau_Y}\,\cal{D}^2 (X,Y) \; &= \; \E \,K(X,X')\, \frac{K(Y,Y')}{\tau_Y} \, + \, \E \,K(X,X') \, \E \,\frac{K(Y,Y')}{\tau_Y} \, - \, 2\, \E \,K(X,X')\, \frac{K(Y,Y'')}{\tau_Y}\\
&= \;\; \E\, K(X,X') \, \left(1 + \frac{1}{2} L(Y, Y') + R(Y, Y')  \right)\\
& \qquad + \; \E\, K(X,X') \, \E\,\left(1 + \frac{1}{2} L(Y, Y') + R(Y, Y')  \right)\\
& \qquad - \; 2\, \E\, K(X,X') \, \left(1 + \frac{1}{2} L(Y, Y'') + R(Y, Y'')  \right)\,\\
&= \;\; L \; + \; R,
\end{align*}
where
\begin{align*}
L \; &=\; \frac{1}{2\tau^2_{Y}} \,  \dis \sum_{j=1}^q \Big\{\,\E\, [K(X,X') \, \rho_j(Y_{(j)}, Y_{(j)}')] \, + \,  \E\, [K(X,X') \, \E\,[\rho_j(Y_{(j)}, Y_{(j)}')] \; -\,2 \,\E\, [K(X,X') \, \rho_j(Y_{(j)}, Y_{(j)}'')]  \Big\}\\
&= \; \frac{1}{2\tau^2_{Y}} \,   \sum_{j=1}^q D^2_{K,\,\rho_j}(X, Y_{(j)})\,,
\end{align*}
and
\begin{align*}
R \; =& \; \E\, \left[ \,  K(X,X') R(Y,Y') \,\right]\;+\; \E\,\left[ K(X,X')\right] \, \E \left[ R(Y,Y')\right]\;  -\; 2\,\E\, \left[ \, K(X,X') R(Y,Y'') \,\right]\,.
\end{align*}
Under the assumption that $\E\,[R^2(Y,Y')] = O(b_q'^4)$, using H{\"o}lder's inequality it is easy to see that $\tau_Y R = O(\tau_{Y}\, b_q'^2) = o(1)$.

\end{proof}

\begin{proof}[Proof of Theorem \ref{ACdCov taylor thm}]

Following equation (\ref{expr to be used later}), we have for $1\leq k\neq l \leq n$
\begin{align*}
\widetilde{D}^X_{kl} \;&=\; \frac{\tau_X}{2} \widetilde{L}(X_k, X_l) \;+\; \tau_X \widetilde{R}(X_k, X_l) \;=\; \frac{1}{2\tau_X} \dis \sum_{i=1}^p \wrho_i(X_{k(i)}, X_{l(i)}) \;+\; \tau_X \widetilde{R}(X_k, X_l)\,,\\
\widetilde{D}^Y_{kl} \;&=\; \frac{\tau_Y}{2} \widetilde{L}(Y_k, Y_l) \;+\; \tau_Y \widetilde{R}(Y_k, Y_l) \;=\; \frac{1}{2\tau_Y} \dis \sum_{j=1}^q \wrho_i(Y_{k(j)}, Y_{l(j)}) \;+\; \tau_Y \widetilde{R}(Y_k, Y_l)\,.
\end{align*}
From equation (\ref{ustat dcov}) in the main paper it is easy to check that
\begin{align*}
\widetilde{\cal{D}^2_n} (X,Y) \;&=\; \frac{1}{4\tau_{XY}} \dis \sum_{i=1}^p \sum_{j=1}^q \widetilde{D^2_n}_{\,;\,\rho_i, \rho_{j}} (X_{(i)}, Y_{(j)}) \;+\;   \frac{\tau_{XY}}{2n(n-3)} \dis \sum_{k\neq l} \widetilde{L}(X_k, X_l) \widetilde{R}(Y_k, Y_l) \\
&\qquad + \; \frac{\tau_{XY}}{2n(n-3)} \dis \sum_{k\neq l} \widetilde{L}(Y_k, Y_l) \widetilde{R}(X_k, X_l) \;+\; \frac{\tau_{XY}}{n(n-3)} \dis \sum_{k\neq l} \widetilde{R}(X_k, X_l)\, \widetilde{R}(Y_k, Y_l) \,.
\end{align*}
Under Assumption \ref{ass2.1}, using H{\"o}lder's inequality and power mean inequality, it can be verified that
\begin{align*}
\dis \sum_{k\neq l} \widetilde{L}(X_k, X_l) \widetilde{R}(Y_k, Y_l) \; & \leq \; \left(\, \dis \sum_{k\neq l} \widetilde{L}(X_k, X_l)^2 \, \sum_{k\neq l} \widetilde{R}(Y_k, Y_l)^2 \,\right)^{1/2} \; = \; O_p(a_p b^2_q)\;,\\
\dis \sum_{k\neq l} \widetilde{L}(Y_k, Y_l) \widetilde{R}(X_k, X_l) \; & \leq \; \left(\, \dis \sum_{k\neq l} \widetilde{L}(Y_k, Y_l)^2 \, \sum_{k\neq l} \widetilde{R}(X_k, X_l)^2 \,\right)^{1/2} \; = \; O_p(a^2_p b_q)\;,\\
\dis \sum_{k\neq l} \widetilde{R}(X_k, X_l) \widetilde{R}(Y_k, Y_l) \; & \leq \; \left(\, \dis \sum_{k\neq l} \widetilde{R}(X_k, X_l)^2 \, \sum_{k\neq l} \widetilde{R}(Y_k, Y_l)^2 \,\right)^{1/2} \; = \; O_p(a_p^2 b^2_q)\;.
\end{align*}
This completes the proof of the theorem.
\end{proof}

\begin{proof}[Proof of Theorem \ref{th MDD sample}]
Following equation (\ref{expr to be used later}), we have for $1\leq k\neq l \leq n$
\begin{align*}
\widetilde{D}^Y_{kl} \;&=\;  \frac{1}{2\tau_Y} \dis \sum_{j=1}^q \wrho_j(Y_{k(j)}, Y_{l(j)}) \;+\; \tau_Y \widetilde{R}(Y_k, Y_l)\,,
\end{align*}
and therefore \begin{align*}
\widetilde{\cal{D}^2_n} (X,Y) \;&=\; \frac{1}{2\tau_{Y}} \dis  \sum_{j=1}^q \widetilde{D^2_n}_{\,;\,K, \rho_{j}} (X, Y_{(j)}) \;+\;   \frac{\tau_{Y}}{n(n-3)} \dis \sum_{k\neq l} \widetilde{K}(X_k, X_l) \widetilde{R}(Y_k, Y_l) \,.
\end{align*}
Using power mean inequality, it can be verified that $ \sum_{k\neq l} \widetilde{K}(X_k, X_l) \widetilde{R}(Y_k, Y_l)\,=\, O_p(b^2_q)$. This completes the proof of the theorem.
\end{proof}

\begin{proof}[Proof of Theorem \ref{ACdcov:dist_conv}]
The proof follows similar lines of the proof Theorem 2.2.1 in Zhu et al.\,(2019), with the distance metric being the one from the class of metrics we proposed in equation (\ref{Kdef}).
\end{proof}

\begin{proof}[Proof of Theorem \ref{HDLSS dist conv}]
The proof of the theorem follows similar lines of the proof of Proposition 2.2.2 in Zhu et al.\,(2019).
\end{proof}

\begin{proof}[Proof of Theorem \ref{ACdCov hd}]
The decomposition into the leading term follows the similar lines of the proof of Theorem \ref{ACdCov taylor thm}. The negligibility of the remainder term can be shown by mimicking the proof of Theorem 3.1.1 in Zhu et al.\,(2019).
\end{proof}

\begin{proof}[Proof of Theorem \ref{HDMSS dist conv}]
It essentially follows similar lines of Proposition 3.2.1 in Zhu et al.\,(2019).
\end{proof}

\end{alphasection}

\end{document}